\pdfoutput=1
\documentstyle[amssymb,11pt]{article}
\newcommand{\bas}[1]{{\bf #1}}
\def\LB{{\lbrack\!\lbrack}}
\def\RB{{\rbrack\!\rbrack}}
\def\bracket#1{\LB#1\RB}
\def\of#1#2{{#1}\bracket{#2}}
\def\morph#1{{\cal {#1}}}
\newcommand{\BA}{{\bf A}}
\newcommand{\BB}{{\bf B}}
\newcommand{\bA}{{\bf A}}
\def\bAof#1{\of\bA{#1}}

\newcommand{\F}{\morph{F}}
\newcommand{\subd}{\lhd}
\newcommand{\sub}{\subseteq}
\newcommand{\dom}[1]{{\cal #1}}
\newcommand{\doma}{{\cal A}}
\newcommand{\domb}{{\cal B}}
\newcommand{\domd}{{\cal D}}
\newcommand{\dome}{{\cal E}}
\newcommand{\domu}{{\cal U}}
\newcommand{\domv}{{\cal V}}
\newcommand{\domf}[1]{{\cal #1}^0}
\newcommand\cK{{\sf{K}}}
\newcommand\cP{{\sf{P}}}
\newcommand\cMap{{\sf{Map}}}
\newcommand\cFun{{\sf{Fun}}}
\newcommand{\cat}[1]{\bf #1}
\newcommand{\pair}[2]{\langle #1, #2 \rangle}
\newcommand{\sigstar}{\Sigma ^\star}
\newcommand{\appx}{\sqsubseteq}
\newcommand{\below}{\sqsubseteq}
\newcommand{\appxneq}{\sqsubset}
\newcommand{\lub}{\sqcup}		
\newcommand{\Lub}{\bigsqcup}
\newcommand{\glb}{\sqcap}
\newcommand{\Glb}{\mbox{\Large $\sqcap$}}
\newcommand{\union}{\cup}
\newcommand{\Union}{\bigcup}
\newcommand{\iso}{\approx}
\newcommand{\ident}{{\sf I}}
\newcommand{\ideal}{{\cal I}}
\newcommand{\mand}{\wedge}
\newcommand{\setof}[1]{{\{{#1}\}}}

\newcommand{\map}{\rightarrow}		
\newcommand{\amap}{\Rightarrow}		
\newcommand{\imp}{\:\Rightarrow\:}	
\newcommand{\eop}{$\:\:\Box$}

\newcommand{\such}{\:\vert\:}
\newcommand{\fa}{\forall}
\newcommand{\te}{\exists}
\newcommand{\cpo}{{\it cpo}}
\newcommand{\mknod}[3]{\put(#1,#2){\oval(20,10)}\put(#1,#2){\makebox(0,0){#3}}}
\newcommand{\nln}[5]{\put(#1,#2){\line(#4,#5){#3}}}
\newcommand{\proof}{\noindent{\bf Proof~~}}
\newcommand{\nat}{{\Bbb{N}}}
\newcommand{\fun}{{\Bbb{F}}}
\newcommand{\domn}{{\cal N}}
\newcommand{\mx}{{\sf max}}
\newcommand{\mn}{{\sf min}}
\newcommand{\Exp}{{\cal Ex}}
\def\apply{{\sf apply}}
\def\Nat{{\Bbb N}}

\newcounter{exercnt}[section]
\renewcommand{\theexercnt}{\thesection.\arabic{exercnt}}
\newenvironment{exer}{\begin{trivlist}\refstepcounter{exercnt}
\item[]{\bf Exercise~\theexercnt:} }{\end{trivlist}}
\newenvironment{mdef}[1]{\begin{trivlist}\refstepcounter{exercnt}
\item[]{\bf Definition~\theexercnt:  [#1]} }{\end{trivlist}}
\newenvironment{thm}{\begin{trivlist}\refstepcounter{exercnt}
\item[]{\bf Theorem~\theexercnt:} }{\end{trivlist}}
\newenvironment{clm}{\begin{trivlist}\refstepcounter{exercnt}
\item[]{\bf Claim~\theexercnt:} }{\end{trivlist}}
\newenvironment{lem}{\begin{trivlist}\refstepcounter{exercnt}
\item[]{\bf Lemma~\theexercnt:} }{\end{trivlist}}
\newenvironment{cor}{\begin{trivlist}\refstepcounter{exercnt}
\item[]{\bf Corollary~\theexercnt:} }{\end{trivlist}}
\newenvironment{exm}{\begin{trivlist}\refstepcounter{exercnt}
\item[]{\bf Example~\theexercnt:} }{\end{trivlist}}
\newenvironment{rem}{\begin{trivlist}\refstepcounter{exercnt}
\item[]{\bf Remark~\theexercnt:} }{\end{trivlist}}
\begin{document}
\begin{titlepage}
\vspace{2in}
\begin{center}
{\LARGE \sc
~\\
~\\
~\\
Domain Theory:\\
\vspace{.3cm}
An Introduction
}
\end{center}
\vspace{.5in}
\begin{center}
{\Large 
\begin{tabular}{cc}
Robert Cartwright&Rebecca Parsons\\
{\it Rice University, USA}&{\it ThoughtWorks Inc., USA}
\end{tabular}\\
\vspace{.5cm}
Moez AbdelGawad\\
\vspace{.1cm}
{\it Hunan University, China}\\\vspace{.2cm}{\it SRTA-City, Egypt}}
\end{center}
\vspace{.5in}
\begin{quote}
This monograph is an ongoing revision of the
``Lectures On A
Mathematical Theory of Computation'' by Dana Scott~\cite{ns}. 
Scott's monograph uses a formulation of domains called {\it
neighborhood systems} in which finite elements are selected
subsets of a master set
of objects called ``tokens''.  Since tokens have little intuitive
significance, Scott has discarded neighborhood systems in favor of an
equivalent formulation of domains called {\it information systems} ~\cite{is}.
Unfortunately, he has not rewritten his monograph to reflect this
change.

We have rewritten Scott's monograph in terms of
{\it finitary bases} (see
Cartwright~\cite{tt})
instead of information systems.
A finitary basis is an information system
that is closed under least upper bounds on finite consistent subsets.
This convention ensures that every finite data value is represented
by a single basis object instead of a set of objects.
\end{quote}

\end{titlepage}
\tableofcontents
\newpage
\section{The Rudiments of Domain Theory}
\subsection{Motivation}
Computer programs perform computations
by repeatedly applying primitive operations
to data values.  The set of primitive operations and data values
depends on the particular programming language.  Nearly all programming languages
support a rich collection of data values including {\it atomic}
objects, such as
booleans, integers, characters, and floating point numbers,
and
{\it composite} objects, such as arrays,
records, sequences, tuples,
and infinite streams.  More advanced languages also support functions
and procedures as data values.
To define the meaning of programs in a given language,
we must first define the building blocks---the primitive data values and
operations---from which computations in the language
are constructed.


{\it Domain theory\/} is a comprehensive mathematical framework for
defining the data values and primitive operations of a
programming language.  A critical feature
of domain theory (and expressive programming
languages like Scheme, ML and Haskell) is the fact that {\it program
operations are also data values}; in domain theory
both operations and data values that the operations operate on are elements of computational
{\it domains\/}.  

In a language implementation,
every data value and operation is represented by a finite
configuration of symbols ({\it e.g.}, a bitstring).  However, the
choice of how data is represented should {\it not} affect the observable
behavior of programs.
Otherwise, a programmer
cannot reason about the behavior of
programs independent of their implementations.


To achieve this goal,
we must abandon finite representations for some data values.
The abstract meaning of a procedure, for example, is typically defined
as a mathematical function from an infinite domain to an infinite codomain.
Although the graph of this function is recursively enumerable, it
does not have an effective {\it finite}
canonical representation; otherwise we could decide
the equality of recursively enumerable sets by generating their
canonical descriptions then comparing these descriptions.

Finite data values have canonical representations. Data values that do not have finite canonical representations 
are called {\it infinite} data values.
Some common examples of infinite data values are
functions over an infinite domain, infinite streams
and infinite trees (sometimes also called ``lazy" lists and trees).  
 A data domain that contains infinite data values is called a 
 {\it higher order} data domain.
To describe an infinite data value, 
we must use an {\it infinite sequence}
of progressively better finite approximations, each obviously having a canonical representation.

We can interpret
each finite approximation as a proposition
asserting that a certain property
is true of the approximated value.  By stating enough
different properties (a countably infinite number of them in general), every
infinite data value can be uniquely identified.  

Higher order data domains can also contain ordinary finite data values.  As such, in a higher order data domain there are {\it two} separate kinds of finite values.
\begin{itemize}
\item
First, the finite elements used to approximate infinite values are
legitimate data values themselves.  Even though these approximating
elements are only ``partially defined,'' they can be produced
as the {\it final} results of computations.
For example, a tree of the
form $cons(\alpha,\beta)$, where $\alpha$ and $\beta$ are arbitrary
data values, is a data value
in its own right, because a computation yielding
$cons(\alpha,\beta)$ may subsequently diverge without producing any
information about the values $\alpha$ and $\beta$.
\item
Second, higher order domains may contain ``maximal''
finite elements that do not
approximate any other values.  These ``maximal''
values correspond to conventional finite data values.
For example, in
the domain of potentially infinite
binary trees of integers, the leaf
carrying the integer 42 does not approximate any value
other than itself.  
\end{itemize}

In summary, a framework for defining computational data values and operations
must accommodate infinite elements, partially-defined elements,
and finite maximal elements.
In addition, the framework
should support the construction of more complex values from simpler
values, and it should support a notion of computation on these objects.
This monograph describes a framework---{\it domain
theory}---satisfying all of these properties.
\subsection{Notation}
The following notation is used throughout the monograph:
\begin{quote}
\begin{tabular}{ccl}
$\imp$&means& logical implication\\
$\iff$&means& if and only if (used in mathematical formulas)\\
iff&means& if and only if (used in text)\\
$\appx$&means& approximation ordering\\
$\lub$&means& least upper bound\\
$\nat$&means& the natural numbers\\
\end{tabular}
\end{quote}
\subsection{Basic Definitions}
To support the idea of describing
data values by generating ``better and better'' approximations,
we need to specify an (``information content") ordering relation among the finite
approximations to data values.
The following definitions describe the structure
of the sets of finite approximations used to build domains;
these sets of finite approximations are
called
{\it finitary bases\/}.

\begin{mdef}{Partial Order}
A {\it partial order\/} {\bf B} is a pair $\langle B,\appx\rangle $ 
consisting of $(i)$ 
a (non-empty) set $B$, called the {\it universe} of {\bf B}, and
$(ii)$ a
binary relation $\appx$ on the
set $B$, called the {\it approximation ordering},
that is
\begin{itemize}
\item
{\it reflexive:} $\forall x \in B [x \appx x]$,
\item
{\it antisymmetric:} $\forall x,y \in B [x \appx y] \mbox{ and } 
[y \appx x] \mbox{ implies } x = y$, and
\item
{\it transitive:} $\forall x,y,z \in B [x \appx y] \mbox{ and } 
[y \appx z] \mbox{ implies } x \appx z$.
\end{itemize}
\end{mdef}
\begin{mdef}{Upper Bounds, Lower Bounds, Consistency}
Let $S$ be a subset of (the universe of) a partial order {\bf B}.
An element $b \in B$ is an {\it upper bound} 
of $S$
iff $\fa s \in S \: s \appx b$.
An element $b \in B$ is a {\it lower bound} 
of $S$
iff
$\fa s \in S \: b \appx s$.
$S$ is {\it consistent} (sometimes also called {\it bounded}) iff $S$ has an upper 
bound. An upper bound $b$ of $S$ is
the {\it least upper bound} of $S$ (denoted $\Lub S$)
iff $b$ approximates all upper bounds of $S$.
A lower bound $b$ of $S$ is
the {\it greatest lower bound} of $S$ (denoted $\Glb S$)
iff every lower bound of $S$ approximates
$b$.
\end{mdef}

\begin{rem} In domain theory upper bounds are much more important 
than lower bounds.
\end{rem}

\begin{mdef}{Directed Set, Progressive Set, Chain}
A subset $S$ of a partial order {\bf B}
is {\it directed} iff
every finite subset of $S$ is consistent ({\it i.e.}, has an upper bound) in $S$.
A directed subset $S$ of $\bas{B}$ is {\it progressive} iff
$S$ does not contain a maximum element: $\nexists b \in S[\fa s \in S s\appx b]$.
A directed subset $S$ of $\bas{B}$ is a {\it chain} iff $S$ is
{\it totally ordered}: $\fa a,b \in S a \appx b$
or $b \appx a$.
\end{mdef}
\begin{clm}
The empty set is directed.
\end{clm}
\begin{mdef}{Complete Partial Order}
A {\it complete partial order\/}, abbreviated {\it cpo\/} (or sometimes {\it dcpo}), is a
partial order $\langle B,\appx\rangle$ such
that every directed subset has a least upper bound in $B$.
\end{mdef}
\begin{clm}
A \cpo\ has a least element.
\end{clm}
\subsection{Finitary Bases}
\begin{mdef}{Finitary Basis}
A {\it finitary basis\/} \bas{B} is a 
partial order $\langle
B,\appx\rangle $ such that $B$ is countable and
every finite consistent subset
has a least upper bound in $B$. 
\end{mdef}

We call the elements of a finitary basis
$\bas{B}$ {\it propositions} since they can be 
interpreted as
logical assertions
about domain elements. 
In propositional logic,
the least upper bound of a set of propositions is the 
conjunction of all the propositions in the set. 
Since the empty set, $\emptyset$, is a finite consistent
subset of $\bas{B}$, it has a least upper bound, which is
denoted $\bot$.
The $\bot$ proposition is true for all domain elements; hence it
does not give any ``information'' about an element.

\begin{exm}\label{exmstr}
Let $B=\{\bot ,0\bot ,1\bot , 00,01,10,11\}$ where $0\bot$
describes strings that start with 0 and are indeterminate past that
point; 00 describes the string consisting of two consecutive 0's.
The other propositions are defined similarly. Let
$\appx$ denote the
implication (or, conversely, the approximation) relation between
propositions.
Thus,
$0\bot \appx 00$ and $0\bot \appx 01$.  In pictorial form, the partial
order $\langle B,\appx\rangle $ looks like:
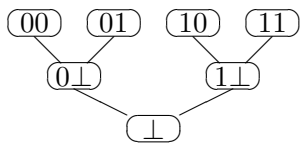
\begin{center}
\begin{picture}(120,55)
\mknod{10}{45}{00} 
\mknod{40}{45}{01} 
\mknod{70}{45}{10} 
\mknod{100}{45}{11} 
\mknod{25}{25}{0$\bot$} 
\mknod{85}{25}{1$\bot$} 
\mknod{55}{5}{$\bot$}
\nln{10}{40}{10}{1}{-1}
\nln{40}{40}{10}{-1}{-1}
\nln{70}{40}{10}{1}{-1}
\nln{100}{40}{10}{-1}{-1}
\nln{25}{20}{20}{2}{-1}
\nln{85}{20}{20}{-2}{-1}
\end{picture}
\end{center}
\end{exm}

$\langle B,\appx\rangle $ is clearly a partial order.
To show that $\langle B,\appx\rangle $ is a finitary basis, we must
show that $B$ is countable and that
all finite bounded ({\it i.e.}, consistent) subsets of $B$ have least
upper bounds.

Since $B$ is finite, it is obviously countable.  It is easy to
confirm that every finite consistent subset has a least
upper bound (by inspection).  In fact, the least upper bound
of any consistent subset $S$ of $B$ is simply the greatest element
of $S$.\footnote{This is a special property of this particular partial order.  It is
not true of finitary bases in general.}
Thus, $\langle B,\appx\rangle $ is a finitary basis.  \eop 

\begin{exm}\label{exmintp}
Let $B=\{(n,m) \: \such \:n,m\in \nat\cup \{\infty\}, n \leq m \}$ where
the proposition $(n,m)$ represents an integer $x$ such that $n \leq x
\leq m$. $\bot$ in this example is the proposition $(0,\infty)$. Let
$\appx$ be defined as
\[(n,m) \appx (j,k) \:\:\iff\:\: n \leq j ~{\rm and}~ k \leq m\]
\end{exm}

For example, $(1,10) \appx (2,6)$ but $(2,6)$ and $(7,12)$ are
incomparable, as are $(2,6)$ and $(4,8)$.  It is easy to confirm that
$B$ is countable and that
$\langle B,\appx\rangle $ is a partial order.
A subset $S$ of $B$ is consistent if there is an integer
for which the proposition in $S$ is true.
Thus, $(2,6)$ and $(4,8)$ are consistent since either 4, 5,
or 6 could be represented by these propositions.  The least upper
bound of these
elements is $(4,6)$.  In general, for a consistent
subset $S = \{(n_i,m_i) \: \such \:i \in I\}$ of $B$,
the least upper bound of $S$ is defined as
\[\Lub S  = (\mx \:  \{n_i \: \such \:i \in I\}, \mn \: 
\{m_i \: \such \:i \in I\})\;.\]
Therefore, $\langle B,\appx\rangle $ is a finitary basis. \eop\\

Given a finitary basis $\bas{B}$, the corresponding
{\it domain\/} $\dom{D}_\bas{B}$ (also called $\domb$)
is constructed by forming all
consistent subsets of $\bas{B}$ that
are ``closed'' under implication and finite conjunction (where
$a \appx b$ corresponds to $b \imp a$ and
$a \lub b$ then corresponds to $a \mand b$).  More precisely, a
consistent subset $S \subseteq \bas{B}$ is an element of
the corresponding domain $\dom{D}_\bas{B}$ iff
\begin{itemize}
\item
$\forall s \in S \;\; \forall b \in \bas{B} \;\; b \appx s \imp b \in S$, and
\item
$\forall r,s \in S \;\; r \lub s \in S$.
\end{itemize}
Corresponding to each basis element/proposition $p \in B$
there is a unique element
$\ideal_p$ =  $\setof{b \in B \: \such \:b \appx p} 
\in \dom{D}_\bas{B}$.  In addition,
$\dom{D}_\bas{B}$ contains elements (``closed'' subsets of $B$)
corresponding to the ``limits'' of all progressive directed
subsets of $\bas{B}$.
This construction
``completes'' the finitary basis $\bas{B}$
by adding limit elements for all progressive directed
subsets of $B$.

In $\dom{D}_\bas{B}$, every
element $d$ is represented by the set of all the propositions in the finitary
basis $\bas{B}$ that describe ({\it i.e.}, approximate) $d$.  
These sets are called {\it ideals\/}.
\begin{mdef}{Ideal}\label{defideal}
For finitary basis $\bas{B}$, a subset $\ideal$ of $B$
is an {\it ideal over} $\bas{B}$ iff
\begin{itemize}
\item
$\ideal$ is downward-closed (implication):
$e\in {\ideal} \imp (\forall b\in B \; b\appx e \imp b\in {\ideal})$
\item
$\ideal$ is closed under least upper bounds on finite subsets (finite conjunction): $\forall r,s \in I \;\; r \lub s \in \ideal$.\footnote{Using induction, it is easy to prove that closure under lubs of pairs implies closure under lubs of finite sets, and trivially vice versa.}
\end{itemize}
\end{mdef}
\subsection{Domains}
Now, we construct domains as partially ordered sets of ideals.
\begin{mdef}{Constructed Domain}\label{1.15}
Let $\bas{B}$ be a finitary basis.
The {\it domain\/} $\dom{D}_\bas{B}$ {\it determined} by $\bas{B}$
is the partial order $\langle D,\appx_D\rangle$ where $D$, the universe of $\dom{D}_\bas{B}$, is
the set of all ideals $\ideal$ over \bas{B}, and $\appx_D$ is the
subset relation.  We will frequently write $\domd$ or $\domb$ instead of
$\dom{D}_\bas{B}$.
\end{mdef}

The proof of the following two claims are easy; they are left to the reader.
\begin{clm}
The least upper bound of 
two ideals $\ideal_1$ and $\ideal_2$, closing over $\ideal_1$ and $\ideal_2$, 
if it exists, is found by unioning $\ideal_1$ and $\ideal_2$
to form an ideal $\ideal_1 \union \ideal_2$ over $\bas{B}$.
\end{clm}
\begin{clm}
The
domain $\domd$ determined by
a finitary basis $\bas{B}$
is a {\it complete partial order\/}.
\end{clm}

Each proposition $b$ in a finitary basis $B$ determines an ideal
$\ideal_b$ consisting of the set of propositions implied by $b$.
An ideal of this form called a {\it principal ideal\/} of $B$.
\begin{mdef}{Principal Ideals}
For finitary basis $\bas{B}=\langle B,\appx\rangle $, the {\it principal ideal}
determined by $b\in B$, is the ideal
\[{\ideal_b}=\{b'\in B \: \such \:b'\appx b\}\;.\]
We will use the notation 
$\ideal_b$
to denote the principal ideal determined by an element $b$ throughout this
monograph.\end{mdef}

Since there is a natural one-to-one correspondence between the
propositions of a finitary basis
$\bas{B}$ and the principal ideals over $\bas{B}$,
the following theorem obviously holds.

\begin{thm}
The principal ideals over a finitary basis $\bas{B}$ form a
finitary basis under the subset ordering.
\end{thm}
\begin{proof}
This is an instance of the universality of the subset relation over partial orders.\eop
\end{proof}\\

Within the domain $\domd$ determined by a finitary
basis $\bas{B}$, the principal ideals are characterized by
an important topological property called {\it finiteness}.

\begin{mdef}{Finite Elements}
An element $e$ of a {\it cpo} $\domd=\langle D,\appx\rangle$ is {\it
finite\/}
iff
for every directed subset $S$ of $D$, $e =
\Lub S \imp e \in S$.\footnote{This property is
	weaker in general than the corresponding
	property (called {\it isolated} or {\it compact})
	that is widely used in topology.  In the context of {\it cpos}, the
	two properties are equivalent.}
\end{mdef}

The set of finite elements in a {\it cpo} $\domd$ is denoted $\domf{D}$. The proof of the following theorem is left to the reader.
\begin{thm}
An element of the domain 
$\domd$ of ideals determined by a finitary basis $\bas{B}$
is finite
iff it is principal.
\end{thm}

In $\domd$,
the principal ideal determined by the least proposition
$\bot$ is the set $\{\bot\}$.  This ideal is the least
element in the domain (viewed as a \cpo).
In contexts where there is no confusion,
we will abuse notation and denote this ideal by the symbol
$\bot$ instead of $\ideal_\bot$.

The next theorem identifies the relationship between an ideal and
all the principal ideals that approximate it.
\begin{thm}\label{idleqthm}  
Let $\domd$ be the domain determined by a finitary basis $\bas{B}$.
For any $\ideal \in \domd$,
$\ideal=\Lub \:\{\ideal' \in \domf{D} \: \such \:\ideal' \appx \ideal \}\, .$
\end{thm}
\proof See Exercise 9.
\eop

The approximation ordering in a partial order
allows us to differentiate
partial elements from total elements.
\begin{mdef}{Partial and Total Elements}
Let $\bas{B}=\langle B,\appx\rangle$ be a partial order.
An element $b \in B$ is {\it partial\/} iff there exists
an element
$b'\in B$ such that $b\neq b'$ and $b\appx b'$ (sometimes written $b \appxneq b'$).
An element $b\in B$ is {\it total\/} iff
for all $b'\in B$, $b\appx b'$ implies $b=b'\, .$
\end{mdef} 

\begin{exm}
The domain determined by the finitary basis defined in Example~\ref{exmstr}
consists only of elements for each proposition in the basis.  The
four total elements are the principal ideals for the propositions
$00,01,10$, and $11$.  In general, a finite basis determines a domain
with this property (all ideals are principal ideals). \eop
\end{exm}
\begin{exm}\label{domnat}
The domain determined by the basis defined in 
Example~\ref{exmintp}
contains total elements for each of the natural numbers.
These elements are the principal ideals for propositions of the form
$(n,n)$.  In this case as well, there are no ideals formed that are
not principal. \eop
\end{exm}
\begin{exm}\label{inftreeexm}
Let $\Sigma = \{0,1\}$, and let $\sigstar$ be the set of all finite
strings over $\Sigma$ with $\epsilon$ denoting the empty string.
$\sigstar$ forms a finitary basis under the prefix ordering on
strings. $\epsilon$ is the least element in 
$\sigstar$.  The domain $\dom{S}$ determined by
$\sigstar$ 
contains principal ideals for all the finite
bitstrings.
In addition, $\dom{S}$ contains nonprincipal ideals
corresponding to all infinite bitstrings.  Given any infinite
bitstring $s$, the corresponding ideal $\ideal_s$ is the
set of all finite prefixes of $s$.  In fact, these prefixes form
a chain.\footnote{Nonprincipal ideals in other
domains are not necessarily chains.  Strings are a special case
because finite elements can ``grow'' in only one direction.
In contrast, the ideals corresponding to infinite trees---other
than vines---are not chains.}
\eop
\end{exm}

If we view {\it cpo}s abstractly, the names we associate with
particular elements in the universe are unimportant.  Consequently,
we introduce the notion of {\it isomorphism}: two domains
are {\it isomorphic} iff they have exactly the same structure.

\begin{mdef}{Isomorphic Partial Orders}
Two partial orders $\BA$
and $\BB$ are 
{\it isomorphic\/}, denoted
$\BA \iso \BB$, iff there exists a
one-to-one onto function 
$m : A \map B$
that preserves the approximation ordering:
\[ \forall a,b\in A \; a\appx_\BA b \iff m(a)\appx_\BB m(b)\,. \]
\end{mdef}

\begin{thm}
Let $\domd$ be the domain determined by a finitary
basis $\bas{B}$.  $\domf{D}$ forms a finitary basis $\bas{B'}$
under the approximation ordering $\appx$ (restricted
to $\domf{D}$).  Moreover, the
domain $\dome$ determined by the
finitary basis $\bas{B'}$ is isomorphic to $\domd$.
\end{thm}

\proof  
Since the finite
elements of $\domd$ are precisely the principal ideals, it is easy to
see that
$\bas{B'}$ is isomorphic to $\bas{B}$.  Hence, $\bas{B'}$
is a finitary basis and $\dome$ is isomorphic to $\domd$.
The isomorphism between $\domd$ and $\dome$ is given by the
mapping $\delta:\domd\map\dome$ is defined by the equation
\[\delta(d)=\{e\in\domf{D} \: \such \:e\appx d\}\, .\]  
\eop

The preceding theorem justifies the following comprehensive
definition for domains.
\begin{mdef}{Domain}
A {\it cpo} $\domd=\langle D,\appx\rangle$ is a {\it domain}
iff 
\begin{itemize}
\item
$\domd^0$ forms a finitary basis under the approximation ordering
$\appx$ restricted to $\domd^0$, and
\item
$\domd$ is isomorphic to 
the domain $\dome$ determined by
$\domd^0$.
\end{itemize}
\end{mdef}
In other words, a domain is a partial order that is isomorphic
to a constructed domain.

To conclude this section, 
we state some closure properties on $\domd$
to provide more intuition about the approximation ordering.
\begin{thm}\label{union}  Let $\domd$ be the domain determined
by a finitary basis $\bas{B}$.
For any subset $S$ of $\domd$,
the
following properties hold:
\begin{enumerate}
\item $\bigcap S \in \domd$ and
$\bigcap S = \Glb S\,.$
\item 
if $S$ is directed, then
$\bigcup S \in \domd$ and
$\bigcup S = \Lub S\,.$

\end{enumerate}
\end{thm}
\proof The conditions for ideals specified in
Definition~\ref{defideal} must be satisfied for these properties to hold.
The intersection case is trivial.  The union
case requires the directedness restriction since ideals require closure under lubs.
\eop 

For the remainder of this monograph, we will ignore the distinction
between principal ideals and the corresponding elements ({\it i.e.}, propositions) of the
finitary basis whenever it is convenient.

\begin{center}
	{\bf Exercises}
\end{center}
\begin{exer}
Let
\[B=\{s_n\: \such \:s_n=\{m\in \nat  \: \such \:m\geq n\},~ n\in \nat
\}\]
What is the approximation ordering for $B$? Verify that $B$ is a
finitary basis.  What are the total elements 
of the domain determined by $B$?  Draw a partial picture demonstrating the
approximation ordering in the basis.
\end{exer}
\begin{exer}
Example~\ref{exmstr} can be generalized to allow strings of any
finite length.  Give the finitary basis for this general case.
What is the approximation ordering?  What does the domain look
like?  What are the total elements in the domain?  Draw a partial
picture of the approximation ordering.
\end{exer}
\begin{exer}
Let $B$ be all finite subsets of $\nat$ with the subset relation
as the approximation relation.  Verify that this is a finitary basis.
What is the domain determined by $B$?  What
are the total elements?  Draw a partial picture of the domain.
\end{exer}
\begin{exer}
Construct two non-isomorphic infinite domains in which all elements are
finite but there is no infinite chain of elements
({\it i.e.}, no sequence $\langle x_n\rangle ^\infty_{n=0}$ with $x_n \appxneq x_{n+1}$---{\it i.e.}, 
$x_n \appx x_{n+1}$ but 
$x_n \neq x_{n+1}$---
for all $n$). 
\end{exer}
\begin{exer}
Let $B$ be the set of all non-empty open intervals on the real
line with rational endpoints plus a ``bottom'' element.  What would a reasonable approximation ordering be? Verify
that $B$ is a finitary basis.  For any real number $r$, show that
\[\{i\in B \: \such \:r\in i\} \union \{ \bot \}\]
is an ideal element.  Is it a total element?  What are the total
elements? (Hint: When $r$ is rational consider all intervals with
$r$ as a right-hand end point.)
\end{exer}
\begin{exer}
Let $\bas{D}$ be a finitary basis for domain $\domd$.  Define a new
basis $\bas{D}'=\{\downarrow X\: \such \:X\in\bas{D}\}$ where $\downarrow
X=\{Y\in\bas{D}\: \such \:X\appx Y\}$.  Show that $\bas{D}'$ is a finitary basis
and that $\bas{D}$ and $\bas{D'}$ are isomorphic.
\end{exer}
\begin{exer}
Let $\langle B,\appx\rangle $ be a finitary basis where
\[B=\{X_0,X_1,\ldots ,X_n,\ldots \}.\]
Suppose that consistency of finite sequences of elements is decidable.
Let
\begin{eqnarray*}
Y_0&=&X_0\\
Y_{n+1}&=&
	\left\{\begin{array}{ll}
		X_{n+1} & 
		\mbox{if $X_{n+1}$ is consistent with $Y_0,Y_1,\ldots,Y_n$}\\
		Y_n & \mbox{otherwise}\,.
		\end{array}
	\right.
\end{eqnarray*}
Show that $\{Y_0,\ldots ,Y_n,\ldots\}$ is a total element in the domain
determined by $B$.  (Hint: Show that $Y_0,\ldots ,Y_{n-1}$ is consistent
for all $n$.) Show that all ideals can be determined by such sequences.
\end{exer}
\begin{exer}
Devise a finitary basis $\bas{B}$ with more than two elements such
that every pair of elements in $B$ is consistent, but
$B$ is not consistent.
\end{exer}
\begin{exer}
Prove Theorem~\ref{idleqthm}.
\end{exer}
\newpage
\section{Operations on Data}
\subsection{Motivation}
Since program operations perform computations {\it incrementally} on
data values that correspond to ideals (sets of approximations),
operations must obey some critical
{\it topological} constraints.
For any approximation $x'$ to the input value $x$, a program 
operation $f$ must produce the output $f(x')$.
Since program output cannot be withdrawn, every program operation $f$
is a {\it monotonic} function: 
$x_1 \appx x_2$ implies $f(x_1) \appx (x_2)$.

We can describe this process in more detail by examining the structure
of computations.
Recall that every value in a domain $\domd$ can be interpreted as
a set of finite elements in $\domd$ that is closed under implication.
When an operation $f$ is applied to the input value $x$, $f$
gathers information about $x$ by asking the program computing $x$ to
generate a countable chain of finite elements $C$ where
$\Lub \:  \{\ideal_c  \: \such \: c \in C\} = x$.
For the sake of simplicity, we can force the chain $C$ describing
the input value $x$ to be infinite:
if $C$ is finite, we can convert it to an equivalent infinite chain by
repeating the last element.  Then we can view operation $f$ as a function
on infinite streams
that
repeatedly ``reads'' the next element in 
an infinite chain $C$ and ``writes'' the next
element in an infinite chain $C'$ where 
$\Lub \:  \{\ideal_{c'}  \: \such \: c' \in C'\} = f(x)$.
Any such function $f$ on $\domd$ is clearly monotonic.
In addition, $f$ obeys the stronger property
that for any directed set $S$, $f(\Lub S) = \Lub \:  \{f(s) \: \such \: s \in S\}$.
This property is called {\it continuity}.\footnote{For the computational motivation behind continuity, the interested reader is referred to Stoy's detailed and highly-readable account in~\cite{stoy}, in particular the derivation of Condition~6.39 on page 99 of~\cite{stoy}, and the following discussion of its implications.}

The formulation of
computable operations as functions on streams of finite elements
is concrete and
intuitive, but it is not canonical.  Due to the elements of domains in general not being totally-ordered but partially-ordered, there are many different
functions on streams of finite elements
corresponding to the same continuous function $f$ over a domain $\domd$.
For this reason, we will use a slightly different model of incremental
computation as the
basis for defining the form of operations on domains.

To produce a canonical representation for
computable operations, we represent values as {\it ideals}
rather than chains of finite elements.  In
addition, we allow 
computations to be performed in parallel, producing finite answers incrementally
in non-deterministic order.
It is important to emphasize that the result
of every computation---an
ideal $\ideal$---is still deterministic; only the order in which the
elements of $\ideal$ are enumerated is non-deterministic.
When an operation $f$ is applied to an input value $x$, $f$
gathers information about $x$ by asking for the enumeration
of the ideal of finite elements $I_x =
\{d \in \domf{D}  \: \such \: d \appx x\}$.
In response to each input approximation $d \appx x$, $f$
enumerates the ideal
$I_{f(d)} =
\{e \in \domf{D}  \: \such \: e \appx f(d)\}$.
Since $I_{f(d)}$ may be infinite, each enumeration must be an independent
computation.  The operation $f$ merges all of these enumerations
yielding an enumeration of the ideal
$I_{f(x)} =
\{e \in \domf{D} \: \such \: e \appx f(x)\}$.

A computable operation $f$ mapping domain $\dom{A}$, with basis $\bas{A}$, into domain $\dom{B}$, with basis $\bas{B}$, can be
formalized as a consistent relation $F \subseteq \bas{A} \times \bas{B}$ (a subset of the Cartesian product of the two basis)
such that 
\begin{itemize}
\item
the image $F(a) = \setof{b \in \bas{B} | a \, F \, b}$
of any input element $a \in \bas{A}$ is an ideal, and
\item
$F$ is monotonic: $a \appx a' \imp F(a) \subseteq F(a')$.
\end{itemize}
These closure properties ensure that
the relation $F$ {\it uniquely} identifies a continuous function $f$ on $D$.
Relations ({\it i.e.}, subsets of $\bas{A} \times \bas{B}$) satisfying these closure properties are called
{\it approximable mappings}.

\subsection{Approximable Mappings and Continuous Functions}
The following set of definitions restates the
preceding descriptions
in more rigorous form.
\begin{mdef}{Approximable Mapping}
Let $\dom{A}$ and $\dom{B}$ be the domains
determined by finitary bases
$\bas{A}$ and $\bas{B}$, respectively.
An {\it approximable mapping}
$F \sub \bas{A} \times \bas{B}$ is a binary relation over
$\bas{A} \times \bas{B}$ such that 
\begin{enumerate}
\item $\bot_A\,F\,\bot_B$
\item If $a\,F\,b$ and $a\,F\,b'$ then $a\,F\,(b \lub b')$
\item If $a\,F\,b$ and $b'\appx_B b$, then $a\,F\,b'$
\item If $a\,F\,b$ and $a\appx_A a'$, then $a'\,F\,b$
\end{enumerate}
The partial order of approximable mappings
$F \sub \bas{A} \times \bas{B}$ under the subset relation
is denoted by the expression $\cMap(\bas{A},\bas{B})$.
\end{mdef}
Conditions 1, 2, and 3 force the image of an input ideal
to be an ideal.
Condition 4
states that the function on ideals associated with $F$ is
monotonic.

\begin{mdef}{Continuous Function}
Let $\dom{A}$ and $\dom{B}$ be the domains
determined by finitary bases
$\bas{A}$ and $\bas{B}$, respectively.
A function $f:\dom{A} \map \dom{B}$ is {\it continuous}
iff for any ideal $\ideal$ in $\dom{A}$,
$f(\ideal) = \Lub \:  \{f(\ideal_a) \: \such \: a \in \ideal\}$.
The partial ordering $\appx_B$ from $\dom{B}$ determines a
partial ordering $\appx$ on continuous functions:
\[ f \appx g \iff \fa x \in \dom{A} \; f(x) \appx_{\dom{B}} 
g(x) \,.\]
The partial order consisting of the continuous functions from
$\dom{A}$ to $\dom{B}$ under the pointwise ordering
is denoted by the expression $\dom{A} \map_c \dom{B}$ (or, sometimes, by the expression $\cFun(\dom{A},\dom{B})$).
\end{mdef}

It is easy to show that continuous functions satisfy a stronger
condition than the definition given above.

\begin{thm}\label{2.3}
If a function $f:\dom{A} \map \dom{B}$ is continuous,
then for every directed subset $S$ of $\dom{A}$,
$f(\Lub S) =
\Lub \:  \{f(\ideal) \: \such \: \ideal \in S\}$.
\end{thm}

\proof By Theorem \ref{union},
$\Lub S$ is simply $\Union S$.
Since $f$ is continuous,
$f(\Lub S) = \Lub \:  \{f(\ideal_a) \: \such \: \exists \ideal \in S \; a \in \ideal\}$.
Similarly, for every $\ideal \in \dom{A}$,
$f(\ideal) = \Lub \:  \{f(\ideal_a) \: \such \: a \in \ideal\}$.
Hence,
$\Lub \:  \{f(\ideal) \: \such \: \ideal \in S\} =
\Lub \:  \{\Lub \:  \{f(\ideal_a) \: \such \: a \in \ideal\} \: \such \: \ideal \in S\} =
\Lub \:  \{f(\ideal_a) \: \such \: \exists \ideal \in S \; a \in \ideal\}$.
\eop\\

Every approximable mapping $F$ over the finitary basis 
$\bas{A} \times \bas{B}$
determines a continuous function $f: \dom{A} \map \dom{B}$.  Similarly,
every continuous function $f: \dom{A} \map \dom{B}$ determines
an approximable mapping $F$ over the finitary basis 
$\bas{A} \times \bas{B}$.
\begin{mdef}{Image of Approximable Mapping}
For approximable mapping 
\[F \sub \bas{A} \times \bas{B}\]
the {\it image\/}
of $d\in \dom{A}$
under $F$  (denoted $\apply(F,d)$) is the ideal
$\{b\in\bas{B} \: \such \: \exists a\in\bas{A}, 
\: a\in d \: \mand a\,F\,b\}\,.$
The {\it function}
$f:\dom{A}\map\dom{B}$ {\it determined by} $F$
is defined by the equation:
\[f(d) = \apply(F,d)\,.\]
\end{mdef}
\begin{rem}
It is easy to confirm that $\apply(F,d)$ is an element of $\dom{B}$
and that the function $f$ is continuous.
Given any ideal $d \in \dom{A}$, 
$\apply(F,d)$ 
is the subset of $\bas{B}$ consisting of all
the elements related by $F$ to finite elements in $d$.
The set $\apply(F,d)$ is an ideal in $\dom{B}$
since
$(i)$ the set $\{b \in \bas{B} \: \such \: a \, F \, b\}$ is downward-closed for
all $a \in \bas{A}$, and 
$(ii)$ $a \, F \, b \mand a \, F \, b'$ implies $a \, F \, (b \lub b')$.
The continuity of $f$ is an immediate consequence of
the definition of $f$ and the definition of continuity.
\end{rem}

The following theorem establishes that the partial
order of
approximable mappings over $\bas{A} \times \bas{B}$ is
isomorphic to the partial order 
of continuous functions in $\dom{A}\map_c\domb$.
\begin{thm}\label{2.7} Let $\bas{A}$ and $\bas{B}$ be
finitary bases.  The partial order $\cMap(\bas{A},\bas{B})$ consisting
of the set of approximable mappings over
$\bas{A}$ and $\bas{B}$ is isomorphic to
the partial order $\dom{A} \map_c \dom{B}$
of continuous functions mapping
$\dom{A}$ into $\dom{B}$.  The isomorphism is witnessed
by the function
$\fun: \cMap(\bas{A},\bas{B}) \map (\dom{A} \map_c \dom{B})$
defined by
\[ \fun(F) = f\]
where $f$ is the function defined by the equation
\[f(d) = \apply(F,d)\]
for all $d \in \dom{A}$.
\end{thm}

\proof 
The theorem is an immediate consequence
of the following lemma. \eop
\begin{lem}\label{2.8}
\begin{enumerate}
\item
For any approximable mappings
$F,G \subseteq \bas{A} \times \bas{B}$
\begin{enumerate}
\item
$\fa a\in\bas{A}, b\in\bas{B} \;\;
a\,F\,b \iff \ideal_b \appx \fun(F)(\ideal_a)$.
\item $F\subseteq G \iff 
\fa a\in\bas{A} \;\; \fun(F)(a) \appx \fun(G)(a)$
\end{enumerate}
\item
The function $\fun: \cMap(\bas{A},\bas{B}) \map (\dom{A} \map_c \dom{B})$
is one-to-one and onto.
\end{enumerate}
\end{lem}
\proof (lemma)
\begin{enumerate}
\item
Part ($a$) is the 
immediate
consequence of 
the definition of $f$ ($b \in f(\ideal_a) \iff a \, F \, b$) and
the fact that $f(\ideal_a)$ is downward closed.
Part ($b$) follows directly from Part ($a$):
$F \subseteq G \iff \fa a \in \dom{A} \;\;
\{b \: \such \: a \, F \, b\} \subseteq  \{b \: \such \:
a \, G \, b\}$.
But the latter holds iff 
$\fa a\in\bas{A} \;\; (f(a) \subseteq g(a) \iff
f(a) \appx g(a))$.
\item
Assume $\fun$ is not one-to-one.  Then there are distinct
approximable mappings $F$ and $G$ such that
$\fun(F) = \fun(G)$.  Since $\fun(F) = \fun(G)$,
\[\fa a \in \bas{A}, b \in \bas{B} \;\;
(\ideal_b \appx \fun(F)(\ideal_a) \iff
\ideal_b \appx \fun(G)(\ideal_a))\:.\]
By Part 1 of the lemma, 
\[\fa a \in \bas{A}, b \in \bas{B} \;\;
(a\,F\,b \iff 
\ideal_b \appx \fun(F)(\ideal_a) \iff
\ideal_b \appx \fun(G)(\ideal_a) \iff
a\,G\,b)\:.\]
\end{enumerate}

We can prove that $\fun$ is {\it onto} as follows.  
Let $f$ be an arbitary continuous function
in $\dom{A}\map\dom{B}$.  Define the relation
$F \sub \bas{A} \times \bas{B}$ by the rule
\[a \: F \: b \iff \ideal_b \appx f(\ideal_a)\;.\]
It is easy to verify that $F$ is  an approximable mapping.
By Part 1 of the lemma,
\[a \: F \: b \iff \ideal_b \appx \fun(F)(\ideal_a)\;.\]
Hence
\[\ideal_b \appx f(\ideal_a) \iff \ideal_b \appx \fun(F)(\ideal_a)\;,\]
implying that $f$ and $\fun(F)$ agree on finite inputs.
Since $f$ and $\fun(F)$ are continuous, they are equal.
\eop\\

The following examples show how approximable mappings
and continuous functions are related.
\begin{exm}\label{2.9}
Let $\dom{B}$ be the domain of infinite strings from the previous section
and let $\dom{T}$ be the truth value domain with two total elements, {\tt
true} and {\tt false}, where 
$\bot_{\dom{T}}$ denotes that there is insufficient
information to determine the outcome.  Let
$p:\dom{B}\map\dom{T}$ be the function defined by the equation:

\[
p(x) =
	\left\{ \begin{array}{ll}
		{\tt true}  & \mbox{if $x = 0^n1y$ where $n$ is even}\\
		{\tt false} & \mbox{if $x = 0^n1y$ where $n$ is odd}\\
		\bot_{\dom{T}}   & \mbox{otherwise}
		\end{array}
	\right.
\]

The function $p$ determines whether or not there 
are an even number of 0's before
the first 1 in the string.  
If there is no 1 in the string, the
result is $\bot_{\dom{T}}$.  It is easy to
show that $p$ is continuous.  
The corresponding binary relation $P$
is defined by the rule:
\begin{eqnarray*}
a\:P\:b &\iff& (b \appx_{\dom{T}}\bot_{\dom{T}}) \; \lor \;\\
& & (0^{2n}1\appx_B a \land b\appx_{\dom{T}}{\tt true}) \; \lor \;\\
 & & (0^{2n+1}1\appx_B a \land b\appx_{\dom{T}}{\tt false})
\end{eqnarray*}

The reader should verify that $P$ is an approximable
mapping and that $p$ is the continuous function determined
by $P$. \eop
\end{exm}
\begin{exm}\label{2.10}
Given the domain $\dom{B}$ 
from the previous example, let 
$g:\dom{B}\map\dom{B}$ be the function defined by the equation:

\[
g(x) =
	\left\{ \begin{array}{ll}
		0^{n+1}y & \mbox{if $x = 0^n1^k0y$}\\
		\bot_D	& \mbox{otherwise}
		\end{array}
	\right. 
\]

The function $g$ eliminates the first substring of the form
$1^k \; (k > 0)$ from
the input string $x$.  If $x = 1^\infty$, the infinite string of ones,
then $g(x) = \bot_D$.  Similarly, if
$x = 0^n1^\infty$, then $g(x) = \bot_D$.  The reader should confirm
that $g$ is continuous and determine the approximable mapping $G$
corresponding to $g$. \eop
\end{exm}

Approximable mappings and continuous functions
can be composed and manipulated just like any other relations and functions.
In particular, the composition operators for approximable
mappings and continuous functions behave
as expected. In fact, they form {\it categories\/}.

\subsection{Categories of Approx. Mappings and Cont. Functions}
Approximable
mappings and continuous functions form categories over finitary bases and domains, respectively. 
\begin{thm}\label{2.11}
The approximable mappings form a {\it category\/} over finitary bases
where the
{\it identity mapping\/} for finitary basis $B$,
$\ident_B \subseteq \bas{B} \times \bas{B}$, is defined for
$a,b\in \bas{B}$ as
\[a\:\ident_B \:b \iff b\appx a\]
and the composition $G \circ F \subseteq \bas{B_1} \times \bas{B_3}$ 
of approximable mappings
$F \subseteq \bas{B_1}\times\bas{B_2}$ and $G \subseteq
\bas{B_2}\times\bas{B_3}$ is defined for
$a\in\bas{B_1}$ and $c\in\bas{B_3}$ by the rule:
\[a \: (G\circ F) \:c \iff \exists b\in\bas{B_2} \: a\,F\,b \mand b\,G\,c\,.\]
\end{thm}
To show that this structure is a category, we must
establish the following properties:
\begin{enumerate}
\item the identity mappings are approximable mappings,
\item the identity mappings composed with an approximable mapping
defines the original mapping,
\item the mappings formed by composition are approximable mappings, and
\item composition of approximable mappings is {\it associative\/}.
\end{enumerate}
\proof Let $F \sub \bas{B_1}\times\bas{B_2}$ and $G \sub \bas{B_2}\times\bas{B_3}$ be
approximable mappings. Let ${\ident}_1, {\ident}_2$ be identity mappings
for $\bas{B_1}$ and $\bas{B_2}$ respectively.  
\begin{enumerate}
\item The verification that the identity mappings satisfy the
requirements for approximable mappings is straightforward and left
to the reader.
\item To show $F\circ {\ident}_1$ and ${\ident}_2\circ F$ are
approximable mappings, we prove the following equivalence:
\[F\circ \ident_1 = \ident_2\circ F = F\]
For $a\in\bas{B_1}$ and $b\in\bas{B_2}$,
\[a\:(F\circ \ident_1) \:b \iff \exists c\in\domd_1 \;\; (c\appx a\mand
c\,F\,b)\, .\]  
By the definition of approximable mappings,
this property
holds iff $a\,F\,b$, implying that
$F$ and $F\circ \ident_1$ are the same relation.  The proof of
the other half of
the equivalence is similar.
\item 
We must show that the relation
$G\circ F$ is approximable given that the relations
$F$ and $G$ are approximable.  To prove the first condition, we observe
that $\bot_1\,F\,\bot_2$ and
$\bot_2\,G\,\bot_3$ by assumption,
implying that $\bot_1\:(G\circ F)\:\bot_3$.  Proving the second
condition requires a bit more work.
If $a\,(G\circ F)\,c$ and $a\,(G\circ F)\,c'$, then by the definition of
composition, $a\, F\,b$ and $b\,G\,c$ for some $b$ and $a\, F\,b'$ and
$b'\,G\,c'$ for some $b'$.   Since $F$ and $G$ are approximable mappings,
$a\,F\,(b\lub b')$ and since $b'\appx (b\lub b')$, it must be true that
$(b\lub b')\,G\,c$.  By an analogous argument, $(b\lub b')\,G\,c'$.
Therefore, $(b\lub b')\,G\,(c\lub c')$ since $G$ is an approximable
mapping, implying that $a\,(G\circ F)\,(c\lub c')$.  The final
condition asserts that $G\circ F$ is monotonic.  We can prove this
as follows.  If $a\appx a'$,
$c'\appx c$ and $a\,(G\circ F)\,c$, then $a\,F\,b$ and $b\,G\,c$ for
some $b$.  So $a'\,F\,b$ and $b\,G\,c'$ and thus $a'\,(G\circ F)\,c'$.
Thus, $G\circ F$ satisfies the conditions of an approximable mapping.
\item Associativity of composition implies that for approximable
mapping $H$ with $F, G$ as above and $H:\domd_3\map\domd_4$,
$H\circ (G\circ F) = (H\circ G)\circ F$. Assume $a\,(H\circ (G\circ
F))\,z$. Then,
\[\begin{array}{lcl}
a\,(H\circ (G\circ F))\,z&\iff&\exists c\in\domd_3 \; a\,(G\circ
F)\,c\mand c\,H\,z\\
&\iff& \exists c\in\domd_3 \, \exists b\in\domd_2 \;  a\,F\,b\mand
b\,G\,c\mand c\,H\,z\\
&\iff& \exists b\in\domd_2 \: \exists c\in\domd_3 \; a\,F\,b\mand
b\,G\,c\mand c\:H\,z\\
&\iff& \exists b\in\domd_2 \:  a\,F\,b\mand b\,(H\circ G)\,z\\
&\iff& a\,((H\circ G)\circ F)\,z
\end{array}\]
\end{enumerate}
\eop

Since finitary
bases correspond to domains and
approximable mappings correspond to continuous functions, 
we
can restate the same theorem in terms of domains and continuous
functions.

\begin{cor}\label{2.12}
The continuous functions form a {\it category\/} over domains
determined by finitary bases;
the
{\it identity function\/} for domain $\dom{B}$,
$\ident_B: \dom{B} \map \dom{B}$, is defined for
by the equation
\[\ident_B(d) = d\]
and the composition $g \circ f \in \dom{B}_1 \map \dom{B}_3$ 
of continuous functions
$f: \dom{B}_1\map\dom{B}_2$ and $g:
\dom{B}_2\map\dom{B}_3$ is defined for
$a\in\dom{B}_1$ by the equation
\[(g\circ f)(a) = g(f(a))\,.\]
\end{cor}

\proof The corollary is an immediate consequence of the preceding theorem
and two facts:
\begin{itemize}
\item
The partial order of finitary bases is isomorphic to the partial order of
domains determined by finitary bases; the ideal completion mapping
established the isomorphism.
\item
The partial order of approximable mappings over $\bas{A} \times
\bas{B}$ is isomorphic to the 
partial order of continuous functions in $\dom{A} \map \dom{B}$.
\end{itemize}
\eop\\

Based on Theorem~\ref{2.11} and Corollary~\ref{2.12}, we define $\cat{FB}$ as the category having finitary bases as its objects and approximable mappings between finitary bases as its arrows, and we define $\cat{Dom}$ as the category having domains as its objects and continuous functions over domains as its arrows.
\subsection{Domain Isomorphisms}
Isomorphisms between domains are important.  We briefly state and
prove
one of their most important properties.
\begin{thm}\label{2.13}
Every isomorphism between domains is characterized by an approximable
mapping between the finitary bases.  Additionally, finite elements are
always mapped to finite elements.
\end{thm}
\proof Let $f:\dom{D}\map\dome$ be a one-to-one and onto function that
preserves the approximation ordering.  Using the earlier theorem
characterizing approximable mappings and their associated functions, we
can define the mapping as $a\,F\,b \iff \ideal_b\appx f(\ideal_a)$ where
$\ideal_a,\ideal_b$ are the principal ideals for $a,b$.  As shown in
Exercise~\ref{2.8}, monotone functions on finite elements always
determine approximable mappings.  Thus, we need to show that the
function described by this mapping, using the function image
construction defined
earlier, is indeed the original function $f$.  To show this, the
following equivalence must be established for $a\in\domd$:
\[f(a) = \{b\in\bas{E}\: \such \: \exists a'\in a \: 
\ideal_b\appx f(\ideal_{a'})\}\]
The right-hand side of this equation, call it $e$,  is an ideal---for
a proof of 
this, see Exercise~\ref{2.11}.  Since $f$ is an onto function,
there must be some $d\in\domd$ such that $f(d)=e$.
Since $a'\in a$, $\ideal_{a'}\appx a$ holds.  Thus, $f(\ideal_{a'})\appx
f(a)$.  Since this holds for all $a'\in a$, $f(d)\appx f(a)$. Now,
since $f$ is an order-preserving function,
$d\appx a$.  In addition, since $a'\in a$, $f(\ideal_{a'})\appx
f(d)$ by the definition of $f(d)$ so $\ideal_{a'}\appx d$.  Thus, $a'\in d$
and thus $a\appx d$ since $a'$ is an arbitrary element of $a$.  Thus,
$a=d$ and $f(a) = f(d)$ as desired.

To show that finite elements are mapped to finite elements, let $\ideal_a\in
\domd$ for $a\in\bas{D}$.  Since $f$ is one-to-one and onto, every
$b\in f(\ideal_a)$ has a unique $\ideal_{b'}\appx \ideal_a$ such that
$f(\ideal_{b'})=\ideal_b$.  This element is found using the inverse of $f$, which
must exist.  Now, let
\[z=\Lub \: \{\ideal_{b'}\: \such \: b\in f(\ideal_a)\}\]
Since $p\appx q$ implies $\ideal_{p'}\appx \ideal_{q'}$, $z$ is also an ideal
(see Exercise~\ref{2.11} again).  Since  $\ideal_{b'}\appx \ideal_a$
holds for each $\ideal_{b'}$,
$z\appx \ideal_a$ must also hold.  Also, since each $\ideal_{b'}\appx z$,
$\ideal_b=f(\ideal_{b'})\appx f(z)$.  Therefore, $b\in f(z)$.  Since $b$ is an
arbitrary element in $f(\ideal_a)$, $f(\ideal_a)\appx f(z)$ must hold and thus
$\ideal_a\appx z$.  Therefore, $\ideal_a=z$ and $a\in z$.  But then $a\in
\ideal_{c'}$ for some $c\in f(\ideal_a)$ by the definition of $z$.
Thus,  $\ideal_a 
\appx \ideal_{c'}$ and  $f(\ideal_a)\appx \ideal_c$.  Since $c$ was
chosen such that 
$\ideal_c\appx f(\ideal_a)$, $\ideal_{c'}\appx \ideal_a$ and therefore
$\ideal_c=f(\ideal_a)$ and  
$f(\ideal_a)$ is finite.  The same argument holds for the inverse of $f$;
therefore, the isomorphism preserves the finiteness of elements.
\eop
\vspace{.25in}
\begin{center}
{\bf Exercises}
\end{center}
\begin{exer}\label{2.14}
Show that the partial order of monotonic functions mapping
$\domd^0$ to $\dome^0$ (using the pointwise ordering)
is isomorphic to the partial order of approximable mappings
$f:\bas{D}\times\bas{E}$. 
\end{exer}
\begin{exer}\label{2.15}
Prove that, if $F \sub \bas{D}\times\bas{E}$ is an approximable
mapping, then the corresponding function $f:\domd\map\dome$ satisfies
the following equation:
\[f(x) = \Lub \: \{e \: \such \: \exists d \in x \; d F e\}\]
for all $x\in\domd$.
\end{exer}
\begin{exer}\label{2.16}
Prove the following claim: if
$F,G \sub \bas{D}\times\bas{E}$ are approximable mappings, then there exists
$H \sub \bas{D}\times\bas{E}$ such that $H = F \cap G = F \glb G$.
\end{exer}
\begin{exer}\label{2.17}
Let $\langle I,\leq\rangle $ be a non-empty partial order that is
directed and let $\langle D,\appx\rangle$ be a finitary basis.  Suppose
that $a:I\map\domd$ is defined such that $i\leq j\imp a(i)\appx a(j)$
for all $i,j\in I$.  Show that 
\[\bigcup\{a(i)\: \such \: i\in I\}\]
is an ideal in $\domd$.  
This says that the domain is closed under directed unions.
Prove also that for an approximable mapping $f:\bas{D}\map\bas{E}$,
then for any directed union,
\[f(\bigcup\{a(i)\: \such \: i\in I\}) = \bigcup\{f(a(i))\: \such \: i\in I\}\]
This says that approximable mappings preserve directed unions.  If an
elementwise function preserves directed unions, must it come from an
approximable mapping? (Hint: See Exercise~\ref{2.9}??).
\end{exer}
\begin{exer}\label{2.18}
Let $\langle I,\leq\rangle $ be a directed partial order with
$f_i:\bas{D}\map\bas{E}$ as a family of approximable mappings indexed
by $i\in I$.  And assume $i\leq j \imp f_i(x)\appx f_j(x)$ for all
$i,j\in I$ and all $x\in\domd$.  Show that there is an approximable
mapping $g:\bas{D}\map\bas{E}$ where
\[g(x)=\bigcup\{f_i(x)\: \such \: i\in I\}\]
for all $x\in\domd$.
\end{exer}
\begin{exer}\label{2.19}
Let $f:\domd\map\dome$ be an isomorphism between domains.  Let
$\phi:\bas{D}\map\bas{E}$ be the one-to-one correspondence from
Theorem~\ref{2.7} where
\[f(\ideal_a)=\ideal_{\phi(a)}\]
for $a\in\bas{D}$.  Show that the approximable mapping determined by
$f$ is the relationship $\phi(x)\appx b$.  Show also that if $a$ and
$a'$ are 
consistent in $\domd$ then $\phi(a\lub a') = \phi(a) \lub \phi(a')$.
Show how this means that isomorphisms between domains correspond to
isomorphisms between the bases for the domains.  
\end{exer}
\begin{exer}\label{2.20}
Show that the mapping defined in Example~\ref{2.10} is approximable.
Is it uniquely determined by the following equations or are some
missing?

\[\begin{array}{lcl}
g(0x) &=& 0g(x)\\
g(11x) &=& g(1x)\\
g(10x) &=& 0x\\
g(1) &=& \bot
\end{array}\]
\end{exer}
\begin{exer}\label{2.21}
Define in words the effect of the approximable mapping
$h:\bas{B}\map\bas{B}$ using the bases defined in Example~\ref{2.10}
where
\[\begin{array}{lcl}
h(0x)&=&00h(x)\\
h(1x)&=&10h(x)
\end{array}\]
for all $x\in\dom{B}$.  Is $h$ an isomorphism?  Does there exist a map
$k:\bas{B}\map\bas{B}$ such that $k\circ h = \ident_B$ and is $k$ a
one-to-one function?
\end{exer}
\begin{exer}\label{2.22}
Generalize the definition of approximable mappings to define mappings
\[f:\bas{D_1}\times\bas{D_2}\map\bas{D_3}\]
of two variables.  (Hint: A mapping $f$ can be a ternary relation
$f\subseteq\bas{D_1}\times\bas{D_2}\times\bas{D_3}$ where the relation
among the basis elements is denoted $(a,b)\,F\,c$.)  State a modified
version of the theorem characterizing these mappings and their
corresponding functions.
\end{exer}
\begin{exer}\label{2.23}
Modify the construction of the domain $\dom{B}$ from
Example~\ref{2.9} to construct a domain $\dom{C}$ with both finite and
infinite total elements ($\dom{B}\subseteq\dom{C}$).  Define an
approximable map, $C$,  on this domain corresponding to the concatenation of
two strings.  (Hint: Use 011 as an finite total element, 011$\bot$
as the corresponding finite partial element.)  Recall that $\epsilon$,
the empty sequence, is different from $\bot$, the undefined
sequence.  Concatenation should be defined such that if $x$ is an
infinite element from $\dom{C}$, then $\fa
y\in\bas{C} \; (x,y)\,C\,x$.  How does concatenation behave on partial
elements on the left?
\end{exer}
\begin{exer}\label{2.24}
Let $\bas{A}$ and $\bas{B}$ be arbitrary finitary bases.  Prove that
the partial order of approximable mappings over $\bas{A} \times
\bas{B}$ is a domain.  (Hint: The finite elements are the closures
of finite consistent relations.)  Prove that the
partial order of continuous functions in $\dom{A} \map \dom{B}$ is
a domain.
\end{exer}
\newpage
\section{Domain Constructors}
Now that the notion of a
domain has been defined, we need to develop
convenient methods for constructing specific domains.  
The strategy that we will follow is to define the simplest domains ({\it i.e.}, flat domains) directly
(as initial or term algebras) 
and to construct more complex
domains by applying domain constructors to simpler domains.
Since
domains are completely determined by finitary bases, we will
focus on how to construct composite finitary bases from 
simpler ones.  These constructions obviously determine
corresponding constructions on domains.
The two most important constructions on finitary bases are (1)
Cartesian products of finitary bases and (2) approximable mappings
on finitary bases constructed using a function-space constructor.
\subsection{Cartesian Products}
\begin{mdef}{Product Basis}\label{3.1}
Let $\bas{D}$ and $\bas{E}$ be the finitary bases
generating domains $\domd$ and $\dome$.  The {\it product basis},
$\bas{D\times E}$, is the partial order consisting
of the universe
\[D\times E = \{[d,e]\such d\in \bas{D},e\in \bas{E}\}\]
and the approximation ordering
\[[d,e] \appx [i,j] \iff d\appx_D i~{\rm and}~e\appx_E j.\]
\end{mdef}
\begin{thm}\label{3.2} The product 
basis of two finitary bases, as defined above,
is a finitary basis.  
\end{thm}
\proof 
Let $\bas{D}$ and 
$\bas{E}$ be finitary
bases and let 
$\bas{D\times E}$ be defined as above.  Since
\bas{D} and \bas{E} are countable, the universe of $\bas{D}\times\bas{E}$
must be countable.  It is easy to show that $
\bas{D}\times\bas{E}$ is a 
partial order.  By the construction, the bottom
element of the product basis is $[\bot_D,\bot_E]$.  For any
finite bounded subset $R$ of $\bas{D}\times\bas{E}$ where
$R=\{[d_i,e_i]\}$, the lub of $R$ is $[\lub\{d_i\},\lub\{e_i\}]$ which
must be defined since \bas{D} and \bas{E} are finitary bases and for
$R$ to be bounded, each of the component sets must be bounded.  \eop

It is straightforward to
define projection
mappings on product bases,
corresponding to the standard projection functions defined on 
Cartesian products of sets.

\begin{mdef}{Projection Mappings}\label{3.3}
For a finitary basis 
$\bas{D}\times\bas{E}$, {\it projection
mappings\/} $\cP_0 \sub (\bas{D}\times\bas{E}) \times \bas{D}$ and
$P_1 \sub (\bas{D}\times\bas{E}) \times \bas{E}$ are the relations
defined by the rules
\begin{eqnarray*}
[d,e] \: P_0 \:d' & \iff & d' \appx_D d \\ \relax
[d,e] \: P_1 \:e' & \iff & e' \appx_E e
\end{eqnarray*}
where $d$,$d'$ are arbitrary
elements of $\bas{D}$ and $e$, $e'$ are arbitrary
elements of $\bas{E}$.

Let $\bas{A}$, $\bas{D}$, and $\bas{E}$ be finitary bases
and let 
$F \sub \bas{A}\times\bas{D}$ and
$G \sub \bas{A}\times\bas{E}$ be approximable mappings.
The 
{\it paired mapping\/} 
$\langle F,G\rangle \sub \bas{A}\times(\bas{D}\times\bas{E})$
is the relation defined by the rule
\[a\:\langle F,G\rangle \:[d,e] \iff a\:F\:d\mand a\:G\:e\]
for all $a\in\bas{A},d\in\bas{D}$, and all $e\in\bas{E}$.  
\end{mdef}

It is an easy exercise to show that 
projection mappings and paired mappings
are approximable mappings (as
defined in the previous section).
\begin{thm}\label{3.4}
The mappings $\cP_0$, $\cP_1$, and $\langle F,G\rangle $ are approximable mappings if
$F,G$ are. In addition,
\begin{enumerate}
\item 
$\cP_0\circ \langle F,G\rangle  = F$ and 
$\cP_1\circ \langle F,G\rangle  = G$.
\item 
For $[d,e]\in \bas{D}\times\bas{E}$ and $d' \in \bas{D}$, $[d,e] \:
\cP_0 \: d' \iff d' \appx d$.
\item
For $[d,e]\in \bas{D}\times\bas{E}$ and $e' \in \bas{E}$, $[d,e] \:
\cP_1 \: e' \iff e' \appx d$.
\item 
For approximable mapping $H \sub \bas{A}\times(\bas{D}\times\bas{E})$,
$H=\langle(\cP_0\circ H),(\cP_1\circ H)\rangle$.
\item 
For $a\in\bas{A}$ and $[d,e] \in \bas{D}\times\bas{E}$,
$[a, [d,e]] \in \langle F,G\rangle \iff
[a,d] \in F \land  [a,e] \in G$.
\end{enumerate}
\end{thm}
\proof The proof is left as an exercise to the reader.  
\eop

The projection mappings and paired mappings on finitary
bases obviously correspond
to continuous functions on the corresponding domains.
We will denote the continous
functions corresponding to $\cP_0$ and $\cP_1$ by
the symbols $p_1$ and $p_2$.  Similarly, we will denote
the function corresponding to the paired mapping $\langle F, G
\rangle$ by the expression $\langle f, g \rangle$.

It should be clear that the definition of projection mappings
and paired mappings
can be generalized
to products of more than two domains.
This generalization enables us to treat
a multi-ary continous function (or
approximable mapping) as a special form of
a unary continuous function (or approximable mapping) since
multi-ary inputs can be treated as single objects in a product
domain.  Moreover, it is easy to show that a
relation $R \sub (\bas{A_1} \times \ldots \times \bas{A_n}) \times
\bas{B}$ 
of arity $n+1$ (as in
Exercise~\ref{2.19}) is an approximable mapping
iff
every restriction of $R$ to a single input argument (by fixing
the other input arguments) is an approximable mapping.

\begin{thm}\label{3.5} A relation
$F \sub (\bas{A}\times\bas{B})\times\bas{C}$ is
an approximable mapping
iff for every $a\in\bas{A}$ and every $b\in\bas{B}$,
the derived relations
\begin{eqnarray*}
F_{a,*} & = & \setof{ [y,z] \; \such \; [[a,y],z] \in F}\\
F_{*,b} & = & \setof{ [x,z] \; \such \; [[x,b],z] \in F}
\end{eqnarray*}
are approximable mappings.
\end{thm}
\proof
Before we prove the theorem, we need to introduce the
class of constant relations $\cK_e \sub \bas{D}\times\bas{E}$
for arbitrary finitary bases $\bas{D}$ and $\bas{E}$ and show that
they are approximable mappings.

\begin{lem}\label{3.6} 
For each $e \in \bas{E}$, let the ``constant'' relation
$\cK_e \sub \bas{D} \times \bas{E}$ be defined by the equation
\[\cK_e = \setof {[d,e'] \; \such \; d \in \bas{D}, e' \appx e}\,.\]
In other words,
\[d \: \cK_e \: e' \; \iff \; e' \appx e \, .\]
For $e\in\bas{E}$, the constant relation
$\cK_e \sub \bas{D} \times \bas{E}$ is approximable.
\end{lem}

\proof (lemma) The proof of this lemma is left to the reader.
\eop 

To prove the 
``if'' direction of the theorem, we observe that we can construct
the relations $F_{a,*}$  and $F_{*,b}$ for all
$a \in \bas{A}$ and $b \in \bas{B}$
by composing and pairing
primitive approximable mappings.  In particular,
$F_{a,*}$ is the relation
\[F \circ \langle K_a, \ident_B \rangle\]
where $\ident_B$ denotes the identity relation
on $\bas{B}$.
Similary,
$F_{*,b}$ is the relation
\[F \circ \langle \ident_A, K_b\rangle\]
where $\ident_A$ denotes the identity relation
on $\bas{A}$.

To prove the ``only-if''
direction, we assume that for all $a \in \bas{A}$ and
$b \in \bas{B}$,
the relations $F_{a,*}$ and $F_{*,b}$ are approximable.
We must show that the four closure properties for approximable
mappings hold for $F$.
\begin{enumerate}
\item
Since $F_{\bot_A,*}$ is approximable, $[\bot_B,\bot_C] \in F_{\bot_A,*}$,
which implies that $[[\bot_A,\bot_B],\bot_C] \in F$.
\item
If $[[x,y],z] \in F$ and $[[x,y],z'] \in F$, then
$[y,z] \in F_{x,*}$ and $[y,z'] \in F_{x,*}$.  Since
$F_{x,*}$ is approximable, $[y,z \lub z'] \in F_{x,*}$, implying
$[[x,y],z \lub z'] \in F$.
\item
If $[[x,y],z] \in F$ and $z' \appx z$, then
$[y,z] \in F_{x,*}$.  Since
$F_{x,*}$ is approximable, $[y,z'] \in F_{x,*}$, implying
$[[x,y],z'] \in F$.
\item
If $[[x,y],z] \in F$ and $[x,y] \appx [x',y']$,
then $[y,z] \in F_{x,*}$, $x \appx x'$, and $y \appx y'$.
Since $F_{x,*}$ is approximable, $[y',z] \in F_{x,*}$, implying
$[x,y'],z] \in F$, which is equivalent to $[x,z] \in F_{*,y'}$.
Since $F_{*,y'}$ is approximable, $[x',z] \in F_{*,y'}$, implying
$[[x',y'],z] \in F$.
\end{enumerate}
\eop

The same result can be restated in terms of continous functions.

\begin{thm}\label{3.7} A function of two arguments,
$f:\dom{A}\times\dom{B}\map\dom{C}$ is continuous
iff for every $a\in\dom{A}$ and every $b\in\dom{B}$,
the unary functions
\[x\mapsto f[a,x]~{\rm and}~y\mapsto f[y,b]\]
are continuous.
\end{thm}
\proof
Immediate from the previous theorem and the fact that
the domain of approximable mappings over
$(\bas{A} \times \bas{B}) \times \bas{C}$ is isomorphic
to the domain of continuous functions over
$(\dom{A} \times \dom{B}) \times \dom{C}$.
\eop

\subsection{Multiary Function Composition}
The composition of functions, as defined in Theorem~\ref{2.11}, can be
generalized to functions of several arguments.  But we need some
new syntactic machinery to describe more general forms of function 
composition.

\begin{mdef}{Cartesian Types} Let $S$ be a set of symbols used to denote
primitive types.  
The set $S^*$ of {\it Cartesian types} over $S$
consists of the set of expressions denoting
all finite non-empty Cartesian products over primitive types
in $S$:
\[S^* ::= S \; \such \; S \times S \; \such \; \ldots \; .\]
A {\it signature} $\Sigma$ is a pair $\pair SO$
consisting of a set $S$ of type names
$\setof{s_1, \ldots, s_m}$ used
to denote domains and a set $O=\setof{o_i^{\rho_i\map\sigma_i}
\; \such \; 1 \leq i \leq m, \; \rho_i \in S^*, \; \sigma_i \in S}$ of function symbols
used
to denote first-order functions over the domains $S$.
Let $V=\setof{v_i^{\tau} \; \such \; \tau \in S,
\; i \in \Nat}$ be a countably-infinite set
of symbols (variables) distinct
from the symbols in $\Sigma$.
The {\it typed expressions} over $\Sigma$ (denoted
$\Exp(\Sigma)$)
is the set of ``typed'' terms
determined by the following inductive definition:
\begin{enumerate}
\item
$v_i^{\tau} \in V$ is a term of type $\tau$,
\item
for $M_1^{\tau_1}, \dots, M_n^{\tau_n} \in \Exp(\Sigma)$
and $o^{(\tau_1 \times \dots \times \tau_n) \to \tau_0} \in O$
then 
\[
o^{(\tau_1 \times \dots \times \tau_n) \to \tau_0}(M_1^{\tau_1}, \dots, 
M_n^{\tau_n})^{\tau_0}\]
is a term of type $\tau_0$.
\end{enumerate}
We will restrict our attention to terms where every instance
of a variable $v_i$ in a term has the same type $\tau$.
To simplify notation,
we will drop the type superscripts from terms whenever they can be easily 
inferred from context.
\end{mdef}

\begin{mdef}{Finitary Algebra}
A {\it finitary algebra} with signature $\Sigma$ is a function $\bA$ mapping
\begin{itemize}
\item
each primitive type $\tau \in S$ to a finitary basis $\bAof{\tau}$,
\item
each operation type $\tau^1 \times \dots \times \tau^n \in S^*$ 
to the finitary basis $\bAof{\tau^1} \times \dots \times \bAof{\tau^n}$,
\item
each function symbol 
$o_i^{\rho_i \map \sigma_i} \in O$
to an approximable mapping
$\bAof{o_i} \sub (\bAof{\rho_i} \times \bAof{\sigma_i})$.
(Recall that $\bAof{\rho_i}$ is a product 
basis.)
\end{itemize}
\end{mdef}

\begin{mdef}{Closed Term}
A term $M \in \Exp(\Sigma)$ is {\it closed} iff it contains no
variables in $V$.
\end{mdef}

The finitary 
algebra $\bA$ implicitly assigns a meaning to every closed term $M$
in $\Exp(\Sigma)$.  This extension is inductively defined by the equation:
\[\bAof{o[M_1,\dots,M_n]} = \bAof{o}[\bAof{M_1}, \dots, \bAof{M_n}] =\]\[
\setof{ b_0 \; \such \; \te [b_1,\dots,b_n] \in \bAof{\rho_i} \;
[b_1,\dots,b_n] \bAof{o} b_0} \, .\]
We can extend $\bA$ to terms $M$ with free variables by
implicitly abstracting over all of the free variables in $M$.

\begin{mdef} {Meaning of Terms}
Let $M$ be a term in $\Exp{(\Sigma)}$
and let $l = x_1^{\tau_1}, \dots, x_n^{\tau_n}$ be
a list of distinct variables in $V$ containing all the free variables
of $M$.  Let $\bA$ be a finitary algebra with signature $\Sigma$ and
for each tuple $[d_1, \dots, d_n] \in \bAof{\tau_1} \times \dots
\times \bAof{\tau_n}$, let $\bA_\setof{x_1 := d_1,\dots, x_n := d_n}$ denote
the algebra $\bA$ extended by defining 
\[\bAof{x_i} = d_i\]
for $1 \leq i \leq n$.
The meaning of $M$ with respect
to $l$, denoted $\bAof{x_1^{\tau_1}, \dots, x_n^{\tau_n} \, \mapsto \,M}$, 
is the relation $F_M \sub (\bAof{\tau_1} \times
\bAof{\tau_n}) \times \bAof{\tau_0}$ defined by 
the equation:
\[
F_M[d_1, \dots, d_n] = \of{\bA_\setof{x_1 := d_1,\dots,x_1 := d_n}}{M}
\]
\end{mdef}

The relation denoted by 
$\bAof{x_1^{\tau_1}, \dots, x_n^{\tau_n} \, \mapsto \, M}$ 
is often called a {\it substitution\/}.  
The following
theorem shows that the relation
$\bAof{x_1^{\tau_1}, \dots, x_n^{\tau_n} \, \mapsto \, M}$ 
is approximable.
\begin{thm}\label{3.12}
(Closure of continuous functions under substitution)
Let $M$ be a term in $\Exp{(\Sigma)}$
and let $l = x_1^{\tau_1}, \dots, x_n^{\tau_n}$ be
a list of distinct variables in $V$ containing all the free variables
of $M$.
Let $\bA$ be a finitary algebra with signature $\Sigma$.
Then the relation $F_M$ denoted by the expression
\[x_1^{\tau_1}, \dots, x_n^{\tau_n} \mapsto M \] 
is approximable.
\end{thm}
\proof The proof proceeds by induction on the structure of $M$.
The base cases are easy.  If $M$ is a variable $x_i$,
the relation $F_M$ is
simply the projection mapping $\cP_i$.
If $M$ is a constant $c$ of type $\tau$, then
$F_M$ is the contant relation $\cK_c$ of arity $n$.
The induction step is also straightforward.
Let $M$ have the form $g[M_1^{\sigma_1},\ldots,M_m^{\sigma_m}]$.  
By the induction
hypothesis,
\[x_1^{\tau_1}, \dots, x_n^{\tau_n} \mapsto M_i^{\sigma_i} \]
denotes an approximable mapping $F_{M_i} \sub
(\bAof{\tau_1} \times \bAof{\tau_n}) \times \bAof{\sigma_i}$.
But $F_M$ is simply the composition of the approximable mapping
$\bAof{g}$ with the mapping $\langle F_{M_1}, \dots F_{M_m} \rangle$.
Theorem~\ref{2.11} tells us that the composition must be approximable.
\eop

The preceding generalization of composition obviously carries over
to continuous functions.  The details are left to the reader.
\subsection{Function Spaces}
The next domain constructor, the function-space constructor, allows
approximable mappings (or, equivalently,
continuous functions)
to be regarded as objects.  
In this framework, standard operations on approximable mappings such as
{\it application\/} and {\it composition\/} are approximable mappings too.
Indeed, the definitions of ideals
and of approximable mappings are quite similar.
The space of approximable mappings
is built by looking at
the actions of mappings on finite sets, and then using progressively
larger finite sets to construct the mappings in the limit.  
To this end, the
notion of a {\it finite step mapping\/} is required. 
\begin{mdef}{Finite Step Mapping}\label{3.13}
Let $\bas{A}$ and $\bas{B}$ be finitary bases.
An approximable mapping $F \sub \bas{A} \times \bas{B}$ is a
{\it finite step mapping} iff there exists
a finite set $S\sub
\bas{A}\times\bas{B}$ and $F$ is the least
approximable mapping such that $S \sub F$.
\end{mdef}
It is easy to show that for every {\it consistent} finite set $S \sub
\bas{A}\times\bas{B}$, a least mapping $F$ always exists.
$F$
is simply the closure of $S$ under the four conditions that
an approximable mapping must satisfy.
The
least approximable mapping respecting the empty set is the 
relation $\setof {\pair a{\bot_B} \; \such \; a \in \bas{A}}$.

The space of approximable mappings is built from finite step
mappings.
\begin{mdef}{Partial Order of Finite Step Mappings}\label{3.8.2}
For finitary bases $\bas{A}$ and $\bas{B}$ the
{\it mapping basis\/} is the partial order
$\bas{A}\amap\bas{B}$ consisting of
\begin{itemize}
\item
the universe of all finite step
mappings, and
\item
the approximation ordering
\[F \appx G \iff \forall a\in\bas{A} \; F(a) \appx_B G(a) \, .\]
\end{itemize}
\end{mdef}

The following theorem establishes that the constructor $\amap$
maps finitary bases into finitary bases.
\begin{thm}\label{3.9}
Let $\bas{A}$ and $\bas{B}$ be finitary bases.  Then, the mapping basis
$\bas{A}\amap\bas{B}$ is a finitary
basis.
\end{thm}
\proof Since the elements are finite subsets of a countable set, the
basis must be countable.  It is easy to confirm that
$\bas{A}\amap\bas{B}$ is a partial order; this task is left
to the reader.  We must show that every finite consistent subset
of $\bas{A}\amap\bas{B}$ has a least upper bound in
$\bas{A}\amap\bas{B}$.
Let $\S$ be a finite consistent subset of the universe of
$\bas{A}\amap\bas{B}$.  Each element of $\S$ is a
set of ordered pairs $\pair ab$ that meets the approximable
mapping closure conditions.
Since $\S$ is consistent, it has an upper bound $\S' \in 
\bas{A}\amap\bas{B}$.
Let $U = \Union \S$.  Clearly, $U \sub \S'$.  But $U$ may not be
approximable.  Let $S$ be the intersection of all relations
in $\bas{A}\amap\bas{B}$ above $\S$.  Clearly $U \sub S$, implying
$S$ is a superset of every element of $\S$.
It is easy to verify that $S$ is approximable, because all the
approximable mapping
closure conditions are preserved by infinite intersections.
\eop

\begin{mdef}{Function Domain}
We will denote
the domain
of ideals determined by the finitary basis $\bas{A}\amap\bas{B}$
by the expression $\dom{A}\amap\dom{B}$.  The justification
for this notation will be explained shortly.
\end{mdef}

Since the partial order of approximable mappings is isomorphic
to the partial order of continuous functions, the preceding
definitions and theorems about approximable mappings can be restated
in terms of continuous functions.

\begin{mdef}{Finite Step Function}
Let $\dom{A}$ and $\dom{B}$ be the domains determined
by the finitary bases $\bas{A}$ and $\bas{B}$, respectively.
A continuous function $f$ in $\dom{A}\map_c\dom{B}$ is
{\it finite} iff there exists a finite step mapping
$F \sub \bas{A} \times \bas{B}$ such that $f$ is the function
determined by $F$.
\end{mdef}

\begin{mdef}{Function Basis}
For domains $\dom{A}$ and $\dom{B}$, the
{\it function basis\/} is the partial order
$(\dom{A}\map_c\dom{B})^0$
consisting of
\begin{itemize}
\item
a universe of all finite step
functions, and
\item
the approximation order
\[f\appx g \iff \forall a\in\dom{A} \; f(a) \appx_\domb g(a) \, .\]
\end{itemize}
\end{mdef}

\begin{cor} (to Theorem~\ref{3.9})
For domains $\dom{A}$ and $\dom{B}$,
the function basis
$(\dom{A}\map_c\dom{B})^0$
is a finitary
basis.
\end{cor}

We can prove that
the domain constructed by generating the ideals over
$\bas{A}\amap\bas{B}$ is isomorphic to the partial
order $\cMap(\bas{A},\bas{B})$
of approximable mappings defined in Section 2.
This result is not surprising; it merely demonstrates
that $\cMap(\bas{A},\bas{B})$
is a domain and that we have identified the finite elements 
correctly in defining
$\bas{A}\amap\bas{B}$.

\begin{thm}\label{3.10} The domain of ideals 
determined by $\bas{A}\amap\bas{B}$
is isomorphic to the partial order of
the approximable mappings $\cMap(\bas{A},\bas{B})$.
Hence, $\cMap(\bas{A},\bas{B})$ is a domain.
\end{thm}
\proof 
We must establish an isomorphism between the domain determined
by $\bas{A}\amap\bas{B}$ and the partial order of mappings
from $\bas{A}$ to $\bas{B}$.  Let $h: \dom{A}\amap\dom{B} \map
\cMap(\bas{A},\bas{B})$ be the function 
defined by rule
\[h \: \F = \Union \setof{F \in \F} \,.\]
It is easy to confirm that
the relation on the right hand side of the preceding equation is an
approximable mapping: if it violated any of the closure properties,
so would a finite approximation in $\F$.
We must prove that the function $h$ is one-to-one and onto.
To prove the former, we note that
each pair of distinct ideals has a witness $\pair ab$ that belongs
to a set in one ideal but not in any set in the other.  Hence,
the images of the two ideals are distinct.  The function $h$
is onto because every approximable mapping is the image of
the set of finite step maps that approximate it.
\eop

The preceding theorem can be restated in terms of continuous functions.
\begin{cor} (to Theorem~\ref{3.10})
The domain of ideals 
determined by the finitary basis $(\dom{A}\map_c\dom{B})^0$
is isomorphic to the partial order of
continuous functions $\dom{A} \map_c \dom{B}$.
Hence, $\dom{A} \map_c \dom{B}$ is a domain.
\end{cor}

Now that we have defined the approximable mapping and continous
function domain constructions,
we can show that operators on maps and functions
introduced in Section 2 are continuous functions.
\begin{thm}\label{3.11}
Given finitary bases, $\bas{A}$ and $\bas{B}$, there is an
approximable mapping
\[Apply:((\bas{A}\amap\bas{B}) \times \bas{A}) \times \bas{B}\]
such that for all $F:\bas{A}\amap\bas{B}$ and $a\in\bas{A}$,
\[Apply[F,a] = F(a)\,.\]
Recall that for any approximable mapping $G \sub \bas{C} \times
\bas{D}$ and any element $c \in \bas{C}$
\[G(c) = \setof{d \; \such \; c \: G \: d}\,.\]
\end{thm}

\proof For $F\in(\bas{A}\amap\bas{B})$, $a\in\bas{A}$ and $b\in\bas{B}$,
define the $Apply$ relation 
as follows:
\[[F,a]\:Apply\:b \iff a\:F\:b\,.\]
It is easy to verify that $Apply$ is an approximable mapping; this
step is left to the reader.
From the definition of $Apply$, we deduce
\[Apply[F,a] = \setof{b \; \such \; [F,a]\: Apply \: b} =
\setof{b \; \such \; a \: F \: b} = F(a)\,.\]
\eop

This theorem can be restated in terms of continuous functions.

\begin{cor}\label{3.11.1}
Given domains, $\dom{A}$ and $\dom{B}$, there is a
continuous function
\[apply:((\dom{A}\map_c\dom{B}) \times \dom{A}) \map_c \dom{B}\]
such that for all $f:\dom{A}\map_c\dom{B}$ and $a\in\dom{A}$,
\[apply[f,a] = f(a)\,.\]
\end{cor}

\proof (of corollary).
Let $apply: ((\dom{A}\map_c \dom{B}) \times \dom{A}) \map_c \dom{B}$
be the continuous function (on functions rather
than relations!) corresponding to $Apply$.
From the definition of $apply$ and
Theorem~\ref{2.7} which relates approximable mappings on
finitary bases to 
continous functions over the corresponding domains, we know that
\[apply[f,\ideal_A]=\{b\in\bas{B}\such\exists F'\in
(\bas{A}\amap\bas{B}) \ ; \exists
a\in \ideal_A \mand F' \sub  F \mand [F',a]\:Apply\:b\}\]
where $F$ denotes the approximable mapping corresponding to $f$.
Since $f$ is the continuous function corresponding to $F$,
\[f(\ideal_A) = \{b\in\bas{B}\such \exists a\in\ideal_A \; a\:F\:b\}\]
So, by the definition of the $Apply$ relation, $apply[f,\ideal_A]\sub
f(\ideal_A)$.  For every $b\in f(\ideal_A)$, there exists
$a \in \ideal_A$ such that $a\:F\:b$.  Let $F'$ be the least
approximable mapping such that $a \: F' b$.  By definition, $F'$ is
a finite step mapping.  Hence $b \in apply[f,\ideal_A]$, implying
$f(\ideal_A) \sub apply[f,\ideal_A]$.  Therefore, $
f(\ideal_A) = apply[f,\ideal_A]$ for arbitrary $\ideal_A$. \eop

The preceding theorem and
corollary demonstrate that approximable mappings and continuous functions
can operate on other approximable mappings or continuous functions
just like other data objects.  The next theorem
shows that the {\it currying\/} operation is a continuous
function.

\begin{mdef} {The Curry Operator}
Let \bas{A}, \bas{B}, and \bas{C} be finitary bases.
Given an approximable mapping
$G$ in the
basis $(\bas{A}\times\bas{B})\amap\bas{C}$, 
\[Curry_G:\bas{A}\amap(\bas{B}\amap\bas{C})\]
is the relation defined by the equation
\[Curry_G(a) = \setof {F \in \bas{B}\amap\bas{C} \; \such \;
\fa [b,c] \in F \; [a,b] \: G \: c}\]
for all $a \in \bas{A}$.
Similarly, given any continuous function
$g: (\dom{A}\times\dom{B})\map_c\dom{C}$,
\[curry_g:\dom{A}\map_c(\dom{B}\map_c\dom{C})\]
is the function defined by the equation
\[curry_g(\ideal_A) = (y \mapsto g[\ideal_A,y])\,.\]
By Theorem~2.7??, $(y \mapsto g[\ideal_A,y])$ is a 
continous function.
\end{mdef}

\begin{lem}
$Curry_G$ is an approximable mapping
and $curry_g$ is the continuous function determined
by $Curry_G$.
\end{lem}

\proof A straightforward exercise.\eop\\

It is more convenient to discuss the {\it currying} operation
in the context of continuous functions than approximable mappings.

\begin{thm}
Let $g \in (\dom{A}\times\dom{B})\map_c\dom{C}$ and
$h \in  (\dom{A}\map_c(\dom{B}\map_c\dom{C})$.
The $curry$ operation satisfies the following two equations:
\begin{eqnarray*}
apply\circ\langle curry_g\circ p_0,p_1\rangle  & = & g\\
curry_{apply\circ\langle h\circ p_0,p_1\rangle} & = & h\, .
\end{eqnarray*}
In addition, the function
\[curry:(\dom{A}\times\dom{B}\map\dom{C})\map
(\dom{A}\map_c(\dom{B}\map_c\dom{C}))\]
defined by the equation
\[curry(g)(\ideal_A)(\ideal_B) = curry_g(\ideal_A)(\ideal_B)\]
is continuous.
\end{thm}
\proof Let $g$ be any continuous function
in the domain $(\dom{A}\times\dom{B})\map_c\dom{C}$.
Recall that 
\[curry_g(a) = (y \mapsto g[a,y])\,.\]
Using this definition and the definition of operations in the
first equation, we can deduce
\[\begin{array}{lcl}
apply\circ\langle curry_g\circ p_0,p_1\rangle [a,b]
& = & apply[\langle curry_g \circ p_0, p_1 \rangle [a,b]] \\
& = & apply[(curry_g \circ p_0)[a,b], p_1[a,b]] \\
& = & apply[curry_g p_0[a,b], b] \\
& = & apply[curry_g \: a, b] \\
& = & curry_g \: a \: b \\
& = & g[a,b]\,.
\end{array}\]
Hence, the first equation holds.

The second equation follows almost immediately from
the first.  
Define $g':(\dom{A}\times\dom{B})\map_c\dom{C}$ by the equation
\[g'[a,b] = h \; a \; b\,.\]
The function $g'$ is defined so that $curry_{g'} = h$.  This
fact is easy to prove.  For
$a \in \dom{A}$:
\begin{eqnarray*}
curry_{g'}(a) & = & (y \mapsto g'[a,y])\\
		& = & (y \mapsto h (a)(y))\\
		& = & h(a) \, .
\end{eqnarray*}
Since $h = curry_{g'}$, the first equation implies that
\begin{eqnarray*}
apply\circ\langle h\circ p_0,p_1\rangle & = &
apply\circ\langle curry_{g'}\circ p_0,p_1\rangle \\
& = & g'\,.
\end{eqnarray*}
Hence,
\[
curry_{apply}\circ\langle h\circ p_0,p_1\rangle = curry_{g'} = h\,.
\]

These two equations show that $(\dom{A}\times\dom{B})\map_c\dom{C}$
is isomorphic to $(\dom{A}\map_c(\dom{B}\map_c\dom{C})$ under the
$curry$ operation.  In addition, the
definition of $curry$ shows that
\[curry(g) \below curry(g') \iff g \below g'\,.\]
Hence, $curry$ is an isomorphism.  Moreover,
$curry$ must be continuous. \eop

The same theorem can be restated in terms of approximable mappings.

\begin{cor}
The relation $Curry_G$ satisfies the following two equations:
\begin{eqnarray*}
Apply\circ\langle Curry_G \circ \cP_0,\cP_1\rangle  & = & G\\
Curry_{Apply\circ\langle G\circ \cP_0,\cP_1\rangle} & = & G\, .
\end{eqnarray*}
In addition, the relation
\[Curry:(\bas{A}\times\bas{B})\amap\bas{C})\amap
(\bas{A} \amap (\bas{B}\amap\bas{C}))\]
defined by the equation
\[Curry(G) = \setof{[a,F] \; \such \; a \in \bas{A}, \: 
F \in (\bas{B}\amap\bas{C}), \; \fa [b,c] \in F \:
[a,b]\: G \: c }\]
is approximable.
\end{cor}

\begin{table}
	\begin{tabular}{|l|l|l|l|l|l|}
		\hline
		Domain   &  Elem. &  Finite Elem.   & Cont. Func.  & Fin. Step Func. & Func. Dom. \\
		\hline
		Fin. Bas. & Ideal & Princ. Ideal & Appr. Map. & Fin. Step Map. & Map. Basis  \\ \hline
	\end{tabular}
	\caption{\label{tabdomfb}Domains and Finitary Bases}
\end{table}
Table~\ref{tabdomfb} summarizes the main elements of the correspondence between 
domains and 
finitary bases. Whenever convenient, in the following sections we take liberty to confuse corresponding notions. Context and notation should make clear which category is meant.

\vspace{.25in}
\begin{center}
{\bf Exercises}
\end{center}
\begin{exer}\label{3.14}
We assume that there is a countable basis.  Thus, the basis elements
could without loss of generality be defined in terms of $\{0,1\}^*$.
Show that the product space $\bas{A}\times\bas{B}$ could be defined as
a finitary basis over $\{0,1\}^*$ such that
\[\bas{A}\times\bas{B}=\{[0a,1b]\such a\in\bas{A},b\in\bas{B}\}\]
Give the appropriate definition for the elements in the domain. Also
show that there exists an approximable mapping
$diag:\bas{D}\map\bas{D}\times\bas{D}$ where $diag(x)= [x,x]$ for all
$x\in\domd$. 
\end{exer}
\begin{exer}\label{3.15}
Establish some standard isomorphisms:
\begin{enumerate}
\item $\bas{A}\times\bas{B}\iso\bas{B}\times\bas{A}$
\item
$\bas{A}\times(\bas{B}\times\bas{C})\iso(\bas{A}\times\bas{B})\times\bas{C}$
\item
$\bas{A}\iso\bas{A'},\bas{B}\iso\bas{B'}\imp\bas{A}\times\bas{B}\iso\bas{A'}\times\bas{B'}$
\end{enumerate}
for all finitary bases.
\end{exer}
\begin{exer}\label{3.16}
Let $B\sub \{0,1\}^*$ be a finitary basis.  Define 
\[B^\infty=\bigcup\limits_{n=0}^\infty 1^n0B\]
Thus, $B^\infty$ contains infinitely many disjoint copies of $B$.  Now
let $D^\infty$ be the least family of subsets over $\{0,1\}^*$ such
that 
\begin{enumerate}
\item $B^\infty\in D^\infty$
\item if $b\in \bas{B}$ and $d\in D^\infty$, then $0X\cup 1Y\in
D^\infty$.
\end{enumerate}
Show that, with the superset relation as the approximation ordering,
$D^\infty$ is a finitary basis.  State any assumptions that must be
made.  Show then that $D^\infty\iso D\times D^\infty$.
\end{exer}
\begin{exer}\label{3.18}
Using the product construction as a guide, generate a definition for the
separated sum structure $\bas{A}+\bas{B}$.  Show that there are mappings
$in_A:\bas{A}\map\bas{A}+\bas{B}$, $in_B:\bas{B}\map\bas{A}+\bas{B}$,
$out_A:\bas{A}+\bas{B}\map\bas{A}$, and
$out_B:\bas{A}+\bas{B}\map\bas{B}$ such that $out_A\circ in_A = \ident_A$
where $\ident_A$ is the identity function on $\bas{A}$.
State any necessary assumptions to ensure this function equation is
true.
\end{exer}
\begin{exer}\label{3.19}
For approximable mappings $f:\bas{A}\map\bas{A'}$ and
$g:\bas{B}\map\bas{B'}$, show that there exist approximable mappings,
$f\times g:\bas{A}\times\bas{B}\map\bas{A'}\times\bas{B'}$ and
$f+g:\bas{A}+\bas{B}\map\bas{A'}+\bas{B'}$ such that
\[(f\times g)[a,b] = [f \: a, g \: b]\]
and thus
\[f\times g = \langle f\circ p_0,g\circ p_1\rangle \]
Show also that 
\[out_A\circ (f+g)\circ in_A = f\]
and
\[out_B\circ (f+g)\circ in_B = g\]
Is $f+g$ uniquely determined by the last two equations?
\end{exer}
\begin{exer}\label{3.22}
Prove that the composition operator is an approximable mapping.  That
is, show that
$comp:(\bas{B}\map\bas{C})\times(\bas{A}\map\bas{B})\map(\bas{A}\map\bas{C})$
is an approximable mapping where for $f:\bas{A}\map\bas{B}$ and
$g:\bas{B}\map\bas{C}$, $comp[g,f] = g\circ f$.  Show this using the
approach used in showing the result for $apply$ and $curry$.  That is,
define the relation and then build the function from $apply$ and $curry$,
using $\circ$ and paired functions.  (Hint: Fill in mappings according
to the following sequence of domains.)
\[\begin{array}{c}
(\bas{A}\map\bas{B})\times\bas{A}\map\bas{B}\\
(\bas{B}\map\bas{C})\times((\bas{A}\map\bas{B})\times\bas{A})\map(\bas{B}\map\bas{C})\times\bas{B}\\
((\bas{B}\map\bas{C})\times(\bas{A}\map\bas{B}))\times\bas{A}\map(\bas{B}\map\bas{C})\times\bas{B}\\
((\bas{B}\map\bas{C})\times(\bas{A}\map\bas{B}))\times\bas{A}\map\bas{C}\\
(\bas{B}\map\bas{C})\times(\bas{A}\map\bas{B})\map(\bas{A}\map\bas{C}).
\end{array}\]
This map shows only one possible solution.
\end{exer}
\begin{exer}\label{3.26}
Show that for every domain $\domd$ there is an approximable mapping
\[cond:\bas{T}\times\bas{D}\times\bas{D}\map\bas{D}\]
called the conditional operator such that
\begin{enumerate}
\item $cond[true,a,b]=a$
\item $cond[false,a,b]=b$
\item $cond[\bot_T,a,b]=\bot_D$
\end{enumerate}
and $\bas{T}=\{\bot_T,true,false\}$ such that $\bot_T\appx true$,
$\bot_T\appx false$, but $true$ and $false$ are incomparable.
(Hint: Define a $Cond$ relation).
\end{exer}
\def\far{\map}
\def\bottom{\bot}
\newpage
\section{Fixed Points and Recursion}
\subsection{Fixed Points}
Functions can now be constructed by composing basic functions.
However, we wish to be able to define functions recursively as well.
The technique of recursive definition will also be useful for defining
domains as we will see in Section~\ref{sec6}.  Recursion can be thought of as (possibly infinite) iterated
function composition. The primary result for interpreting recursive definitions is the following {\it Fixed Point
Theorem\/}. 
\begin{thm}\label{4.1}
For any continuous function $f:\domd\far\domd$ determined by an approximable mapping $F:\bas{D}\far\bas{D}$,
there exists a least element $x\in\domd$ such that 
\[f(x) = x.\]
\end{thm}
\proof Let $f^n$ stand for the function $f$ composed with itself $n$
times, and similarly for $F^n$.  Thus, for
\[\begin{array}{lcl}
f^0&=&I_\domd~\\
f^{n+1}&=&f\circ f^n\\
F^0&=&\ident_D~{\rm and}\\
F^{n+1}&=&F\circ F^n
\end{array}\]
we define
\[x = \{d\in\bas{D}\such\exists n\in\nat. \bot\:F^n\:d\}.\]
To show that $x\in\domd$, we must show it to be an ideal.  Map $F$
is an approximable mapping, so $\bot\in x$ since
$\bot\:F\:\bot$.  For $d\in x$ and $d'\appx d$, $d'\in x$ must hold
since, for $d\in x$, there must exist an $a\in \bas{D}$ such that
$a\:F\:d$.  But by the definition of an approximable mapping,
$a\:F\:d'$ must hold as well so $d'\in x$. Closure under lubs is
direct since $F$ must include lubs to be approximable.

To see that $f(x)=x$, or equivalently $x\:F\:x$, note that for any $d\in x$, if $d\:F\:d'$, then
$d'\in x$.  Thus, $f(x)\appx x$.  Now, $x$ is constructed to
be the least element in $\domd$ with this property. To see this is
true, let $a\in\domd$ such that $f(a)\appx a$.  We want to show that
$x\appx a$.  Let $d\in x$ be an arbitrary element.  Therefore, there
exists an $n$ such that $\bot\: F^n\:d$ and therefore
\[\bot\:F\:d_1\:F\:d_2\:\ldots\:F\:d_{n-1}\:F\:d.\]
Since $\bot\in a$, $d_1\in f(a)$.  Thus, since $f(a)\appx a$,
$d_1\in a$.  Thus, $d_2\in f(a)$ and therefore $d_2\in a$.  Using
induction on $n$, we can show that $d\in f(a)$.  Therefore, $d\in a$
and thus $x\appx a$.  

Since $f$ is
monotonic and $f(x)\appx x$, $f(f(x))\appx f(x)$.  Since $x$ is the
least element with this property, $x\appx f(x)$ and thus $x=f(x)$.
\eop

Since the element $x$ above is the least element, it must be unique.
Thus we have defined a function mapping the domain $\domd\far\domd$
into the domain $\domd$.  The next step is to show that this mapping
is approximable.
\begin{thm}\label{4.2}
For any domain $\domd$, there is an approximable mapping
\[fix:(\bas{D}\far\bas{D})\far\bas{D}\]
such that if $f:\bas{D}\far\bas{D}$ is an approximable mapping,
\[fix(f) = f(fix(f))\]
and for $x\in\domd$,
\[f(x)\appx x \imp fix(f)\appx x\]
This property implies that $fix$ is unique.  The function $fix$ is
characterized by the equation
\[fix(f)=\bigcup\limits_{n=0}^\infty f^n(\bottom)\]
for all $f:\bas{D}\far\bas{D}$.
\end{thm}
\proof The final equation can be simplified to
\[fix(f) = \{d\in\bas{D}\such \exists n\in\nat.\bot\:f^n\:d\}\]
which is the equation used in the previous theorem to define the fixed
point. Using the formula from Exercise~\ref{2.9} on the above
definition for $fix$ yields the following equation to be shown:
\[fix(f)=\bigcup \{fix(\ideal_F)\such\exists
F\in(\bas{D}\far\bas{D}).F\in f\}\]
where $\ideal_F$ denotes the ideal for $F$ in $\bas{D}\far\bas{D}$.

From its definition, $fix$ is monotonic since, if $f\appx g$,
then $fix(f)\appx fix(g)$ since $f^n\appx g^n$.  Since $F\in f$,
$\ideal_F\appx f$ and since $fix$ is monotonic, $fix(\ideal_F)\appx fix(f)$.  

Let $x\in fix(f)$.  Thus, there is a finite sequence of elements
such that $\bot\:f\:x_1\:f\:\ldots\:f\:x'\:f\:x$.  Define $F$ as the
basis element encompassing the step functions required for this
sequence.  Clearly, $F\in f$.  In addition, this same sequence exists in
$fix(\ideal_F)$ since we constructed $F$ to contain it,  and thus, $x\in
fix(\ideal_F)$ and $fix(f)\appx fix(\ideal_F)$.  The equality is therefore
established.  

The first equality is direct from the Fixed Point Theorem since the same
definition is used.  Assume $f(x)\appx x$ for some $x\in\domd$.  Since
$\bot\in x$, $x\neq\emptyset$.  Since $f$ is an approximable
mapping, for $x'\in x$ and $x'\:f\:y$, $y\in x$ must hold.  By
induction, for any $\bot\:f\:y$, $y\in x$ must hold.  Thus, $fix(f)\appx x$.

To see that the operator is unique, define another operator $fax$ that
satisfies the first two equations.  It can easily be shown that
\[\begin{array}{lcll}
fix(f)&\appx& fax(f)~{\rm and}\\
fax(f)&\appx&fix(f)
\end{array}\]
Thus the two operators are the same. \eop

\subsection{Recursive Definitions}
Recursion has played a part already in the definitions above.  Recall that
$f^n$ was defined for all $n\in\nat$.  More complex examples of
recursion are given below.
\begin{exm}\label{4.3}
Define a basis $\bas{N}=\langle N,\appx_N\rangle $ where
\[N=\{\{n\}\such n\in \nat\}\cup \{\nat\}\]
and the approximation ordering is the superset relation.  This
generates a flat domain with $\bottom=\{\{\nat\}\}$ and the total elements
being in a one-to-one correspondence with the natural numbers.  Using
the construction outlined in Exercise~\ref{3.16}, construct the basis
$F=N^\infty$.  Its corresponding domain is the domain of partial functions over the
natural numbers.  To see this, let $\Phi$ be the set of all finite
partial functions $\varphi\subseteq \nat\times\nat$.  Define
\[\uparrow\varphi=\{\psi\in\Phi\such\varphi\subseteq\psi\}\]
Consider the finitary basis $\langle F',\appx_F'\rangle $ where 
\[F'=\{\uparrow\varphi\such\varphi\in\Phi\}\]
and the approximation order is the superset relation.  The reader
should satisfy himself that $F'$ and $F$ are isomorphic and that the
elements are the partial functions.  The total elements are the total
functions over the natural numbers.  

The domains $\dom{F}$ and $(\dom{N}\far\dom{N})$ are not isomorphic.
However, the following mapping $val:F\times\bas{N}\far\bas{N}$ can be
defined as follows:
\[(\uparrow\varphi,\{n\})\:val\:\{m\} \iff (n,m)\in\varphi\]
and
\[(\uparrow\varphi,\nat)\:val\:\nat\]
Define also as the ideal for $m\in\dom{N}$,
\[\hat{m} = \{\{m\},\nat\}\]
It is easy to show then that for $\pi\in\dom{F}$ and $n\in\dom{N}$ we
have
\[\begin{array}{lcll}
val(\pi,\hat{n})&=&\hat{\pi(n)}&{\rm if}~\pi(n)\neq\bottom\\
&=&\bottom&{\rm otherwise}
\end{array}\]
Thus, \[curry(val):\bas{F}\far(\bas{N}\far\bas{N})\]
is a one-to-one function on elements.  (The problem is that
($\bas{N}\far\bas{N}$) has more elements than \bas{F} does as the
reader should verify for himself).  

Now, what about mappings $f:\bas{F}\far\bas{F}$?  Consider the
function
\[\begin{array}{lcll}
f(\pi)(n)&=&0&{\rm if}~n=0\\
&=&\pi(n-1)+n-1&{\rm for}~n>0
\end{array}\]
If $\pi$ is a total function, $f(\pi)$ is a total function.
If $\pi(k)$ is undefined, then $f(\pi)(k+1)$ is undefined.
The function $f$ is approximable since it is completely determined by
its actions on partial functions.  That is
\[f(\pi)=\bigcup\{f(\varphi)\such\exists\varphi\in\Phi.
\varphi\subseteq\pi\}\]
The Fixed Point Theorem defines a least fixed point for any
approximable mapping.  Let $\sigma=f(\sigma)$.  Now, $\sigma(0)=0$ and
\[\begin{array}{lcl}
\sigma(n+1)&=&f(\sigma)(n+1)\\
&=&\sigma(n)+n
\end{array}\]
By induction, $\sigma(n)=\sum\limits_{i=0}^n i$ and therefore, $\sigma$ is a
total function.  Thus, $f$ has a unique fixed point.  

Now, in looking at $(\bas{N}\far\bas{N})$, we have
$\hat{0}\in\dom{N}$ (The symbols $n$ and $\hat{n}$ will no longer be
distinguished, but the usage should be clear from context.).  Now define
the two mappings, $succ,pred:\bas{N}\far\bas{N}$ as approximable
mappings such that
\[\begin{array}{lcl}
n\:succ\: m& \iff& \exists p\in\nat.n\appx p,m\appx p+1\\
n\:pred\: m& \iff& \exists p+1\in\nat.n\appx p+1,m\appx p
\end{array}\]
In more familiar terms, the same functions are defined as
\[\begin{array}{lcll}
succ(n)&=&n+1\\
pred(n)&=&n-1&{\rm if}~n>0\\
&=&\bottom&{\rm if}~n=0
\end{array}\]
The mapping $zero:\bas{N}\far\bas{T}$ is also defined such that
\[\begin{array}{lcll}
zero(n)&=&true&{\rm if}~n=0\\
&=&false&{\rm if}~n>0
\end{array}\]
where $\dom{T}$ is the domain of truth value defined in an earlier
section. The {\it structured domain\/} $\langle N,0,succ,pred,zero\rangle $ is
called ``The Domain of the Integers'' in the present context.  The
function element $\sigma$ defined as the fixed point of the mapping
$f$ can now be defined directly as a mapping
$\sigma:\bas{N}\far\bas{N}$ as follows:
\[\sigma(n)=cond(zero(n),0,\sigma(pred(n))+pred(n))\]
where the function $+$ must be suitably defined.  Recall that $cond$
was defined earlier as part of the structure of the domain $\dom{T}$.
This equation is called a {\it functional equation\/}; the next
section will give another notation, the $\lambda-calculus$ for writing
such equations. \eop
\end{exm}
\begin{exm}\label{4.4}
The domain $\dom{B}$ defined in Example~\ref{2.3} contained only
infinite elements as total elements.  A related domain, $\dom{C}$
defined in Exercise~\ref{2.21}, can be regarded as a generalization on
$\dom{N}$. To demonstrate this, the structured domain corresponding to
the domain of integers must be presented.  The total elements in
$\dom{C}$ are denoted $\sigma$ while the partial elements are denoted
$\sigma\bottom$ for any $\sigma\in\{0,1\}^*$.  

The empty sequence $\epsilon$ assumes the role of the number $0$ in
$\dom{N}$.  Two approximable mappings can serve as the successor
function: $x\mapsto 0x$ denoted $succ_0$ and $x\mapsto 1x$ denoted
$succ_1$.  The predecessor function is filled by the $tail$ mapping
defined as follows: 
\[\begin{array}{lcll}
tail(0x)& =& x,\\
tail(1x)& =& x&{\rm and}\\
tail(\epsilon)& =& \bottom.
\end{array}\]
The $zero$ predicate is defined using the $empty$ mapping defined as
follows:
\[\begin{array}{lcll}
empty(0x)& =& false,\\
empty(1x)& =& false&{\rm and}\\
empty(\epsilon)& =& true.
\end{array}\]
To distinguish the other types of elements in $\dom{C}$, the following
mappings are also defined:
\[\begin{array}{lcll}
zero(0x)& =& true,\\
zero(1x)& =& false&{\rm and}\\
zero(\epsilon)& =& false.\\
one(0x)& =& false,\\
one(1x)& =& true&{\rm and}\\
one(\epsilon)& =& false.
\end{array}\]
The reader should verify the conditions for an approximable
mapping are met by these functions.

An element of $\dom{C}$ can be defined using a fixed point equation.
For example, the total element representing an infinite sequence of
alternating zeroes and ones is defined by the fixed point of the
equation
$a=01a$.
This same element is defined with the equation
$a=0101a$.
(Is the element defined as
$b=010b$
the same as the previous two?)

Approximable mappings in $\dom{C}\far\dom{C}$ can also be defined
using equations.  For example, the mapping
\[\begin{array}{lcll}
d(\epsilon) &= &\epsilon,\\
d(0x)&=&00d(x)&{\rm and}\\
d(1x)&=&11d(x)
\end{array}\]
can be characterized with the functional equation
\[d(x)=cond(empty(x),\epsilon,cond(zero(x),succ_0(succ_0(d(tail(x)))),succ_1(succ_1(d(tail(x))))))\]

The concatenation function of Exercise~\ref{2.21} over
$\dom{C}\times\dom{C}\far\dom{C}$ can be defined with the functional
equation
\[C(x,y)=cond(empty(x),y,cond(zero(x),succ_0(C(tail(x),y)),succ_1(C(tail(x),y))))\]
The reader should verify that this definition is consistent with the
properties required in the exercise.
\end{exm}
These definitions all use {\it recursion\/}.  They rely
on the object being defined for a base case ($\epsilon$ for example) or
on earlier values ($tail(x)$ for example).  These equations
characterize the object being defined, but unless a theorem is proved
to show that a solution to the equation exists, the definition is
meaningless.  However, the Fixed Point Theorem for domains was established
earlier in this section.  Thus, solutions exist to these equations
provided that the variables in the equation range over domains and any
other functions appearing in the equation are known to be continuous
(that is approximable).  

\subsection{Peano's Axioms}
To illustrate one use of the Fixed Point Theorem as well as show the
use of recursion in a more familiar setting, we will show that all
second order models of Peano's axioms are isomorphic.  Recall that
\begin{mdef}{Model for Peano's Axioms}\label{4.5}
A structured set $\langle \nat,0,succ\rangle $ for $0\in\nat$ and
$succ:\nat\times\nat$ is a {\it model for Peano's axioms\/} if all the
following conditions are satisfied:
\begin{enumerate}
\item $\forall n\in\nat . 0\neq succ(n)$
\item $\forall n,m \in\nat .succ(n)=succ(m)\imp n=m$
\item $\forall x\subseteq\nat .0\in x\mand succ(x)\subseteq x\imp
x=\nat$
\end{enumerate}
where $succ(x)=\{succ(n)\such n\in x\}$.  
The final clause is usually referred to as the principle of
mathematical induction. 
\end{mdef}
\begin{thm}\label{4.6}
All second order models of Peano's axioms are isomorphic.
\end{thm}
\proof Let $\langle N,0,+\rangle $ and $\langle M,\bullet,\#\rangle $ be models for Peano's
axioms.  Let $N\times M$ be the cartesian product of the two sets and
let ${\cal P}(N\times M)$ be the powerset of $N\times M$.  Recall from
Exercise~\ref{1.15} that the powerset can be viewed as a domain with
the subset relation as the approximation order.  Define the following
mapping:
\[u\mapsto \{(0,\bullet)\}\cup\{(+(n),\#(m))\such (n,m)\in u\}\]
The reader should verify that this mapping is approximable.  Since it
is indeed approximable, a fixed point exists for the function.  Let
$r$ be the least fixed point:
\[r=\{(0,\bullet)\}\cup\{(+(n),\#(m))\such (n,m)\in r\}\]
But $r$ defines a binary relation which establishes the isomorphism.
To see that $r$ is an isomorphism, the one-to-one and onto features
must be established.  By construction,
\begin{enumerate}
\item $0\:r\:\bullet$ and
\item $n\:r\:m \imp +(n)\:r\:\#(m)$.
\end{enumerate}

Now, the sets $\{(0,\bullet)\}$ and $\{(+(n),\#(m))\such (n,m)\in r\}$
are disjoint by  the first axiom.  Therefore, $0$ corresponds
to only one element in $m$.  Let $x\subseteq N$ be the set of all
elements of $N$ that correspond to only one element in $m$.  Clearly,
$0\in x$.  Now, for some $y\in x$ let $z\in M$ be the element in $M$
that $y$ uniquely corresponds to (that is $y\:r\:z$).  But this means
that $+(y)\:r\#(z)$ by the construction of the relation.  If there
exists $w\in M$ such that $+(y)\:r\:w$ and since $(+(y),w)\neq
(0,\bullet)$, the fixed point equation implies that $(+(y)=+(n_0))$
and $(w=\#(m_0))$ for some $(n_0,m_0)\in r$.  But then by the second
axiom, $y=n_0$ and since $y\in x$, $z=m_0$.  Thus, $\#(z)$ is the
unique element corresponding to $+(y)$. The third axiom can now be
applied, and thus every element in $N$ corresponds to a unique element
in $M$.  The roles of $N$ and $M$ can be reversed in this proof.
Therefore, it can also be shown that every element of $M$ corresponds
to a unique element in $N$.  Thus, $r$ is a one-to-one and onto
correspondence. 
\eop
\vspace{.25in}
\begin{center}
{\bf Exercises}
\end{center}
\begin{exer}\label{4.7}
In Theorem~\ref{4.2}, an equation was given to find the least fixed
point of a function $f:\domd\far\domd$.  Suppose that for
$a\in\dom{D}$, $a\appx f(a)$.  Will the fixed point $x=f(x)$ be such
that $a\appx x$?  (Hint:  How do we know that
$\bigcup\limits_{n=0}^\infty f^n(a)\in\domd$?)
\end{exer}
\begin{exer}\label{4.8}
Let $f:\domd\far\domd$ and $S\subseteq \domd$ satisfy
\begin{enumerate}
\item $\bottom\in S$
\item $x\in S\imp f(x)\in S$
\item $[\forall n .\{x_n\}\subseteq S \mand x_n\appx x_{n+1}]\imp
\bigcup\limits_{n=0}^\infty x_n \in S$
\end{enumerate}
Conclude that $fix(f)\in S$.  This is sometimes called the principle
of fixed point induction.  Apply this method to the set
\[S=\{x\in\domd\such a(x)= b(x)\}\]
where $a,b:\domd\far\domd$ are approximable, $a(\bottom)=b(\bottom)$,
and $f\circ a=a\circ f$ and $f\circ b=b\circ f$.
\end{exer}
\begin{exer}\label{4.9}
Show that there is an approximable operator
\[\Psi:((\domd\far\domd)\far\domd)\far((\domd\far\domd)\far\domd)\]
such that for $\Theta:(\domd\far\domd)\far\domd$ and
$f:\domd\far\domd$,
\[\Psi(\Theta) (f) = f(\Theta(f))\]
Prove also that $fix:(\domd\far\domd)\far\domd$ is the least fixed
point of $\Psi$.
\end{exer}
\begin{exer}\label{4.10}
Given a domain $\domd$ and an element $a\in\domd$, construct the
domain $\domd_a$ where
\[\domd_a=\{x\in\domd\such x\appx a\}\]
Show that if $f:\domd\far\domd$ is approximable, then $f$ can be
restricted to another approximable map
$f':\domd_{fix(f)}\far\domd_{fix(f)}$ where $\forall x\in
\domd_{fix(f)}.f'(x)=f(x)$
How many fixed points does $f'$ have in $\domd_{fix(f)}$?
\end{exer}
\begin{exer}\label{4.11}
The mapping ${\bf fix}$ can be viewed as assigning a fixed point operator to
any domain $\domd$.  Show that ${\bf fix}$ can be uniquely characterized by
the following conditions on an assignment $\domd\leadsto F_D$:
\begin{enumerate}
\item $F_D:(\domd\far\domd)\far\domd$
\item $F_D(f)=f(F_D(f))$ for all $f:\domd\far\domd$
\item when $f_0:\domd_0\far\domd_0$ and $f_1:\domd_1\far\domd_1$ are
given and $h:\domd_0\far\domd_1$ is such that $h(\bottom)=\bottom$ and
$h\circ f_0=f_1\circ h$, then
\[h(F_{D_0}(f_0)) = F_{D_1}(f_1).\]
\end{enumerate}
Hint:  Apply Exercise~\ref{4.7} to show ${\bf fix}$ satisfies the
conditions.  For the other direction, apply Exercise~\ref{4.10}.
\end{exer}
\begin{exer}\label{4.12}
Must an approximable function have a maximum fixed point?  Give an
example of an approximable function that has many fixed points. 
\end{exer}
\begin{exer}\label{4.14}
Must a monotone function $f:{\cal P}(A)\far{\cal P}(A)$ have a maximum
fixed point? (Recall: ${\cal P}(A)$ is the powerset of the set $A$).
\end{exer}
\begin{exer}\label{4.18}
Verify the assertions made in the first example of this section. 
\end{exer}
\begin{exer}\label{4.19}
Verify the assertions made in the second example, in particular those
in the discussion of ``Peano's Axioms''.  Show that the predicate
function $one:\dom{C}\far\dom{T}$ could be defined using a fixed point
equation from the other functions in the structure.
\end{exer}
\begin{exer}\label{4.20}
Prove that
\[fix(f\circ g)=f(fix(g\circ f))\]
for approximable functions $f,g:\domd\far\domd$.
\end{exer}
\begin{exer}\label{4.21}
Show that the less-than-or-equal-to relation $l\subseteq
\nat\times\nat$ is uniquely determined by 
\[l=\{(n,n)\such n\in\nat\}\cup \{(n,succ(m)\such (n,m)\in l\}\]
for the structure called the ``Domain of Integers''.
\end{exer}
\begin{exer}\label{4.22}
Let $N^*$ be a structured set satisfying only the first two of the
axioms referred to as ``Peano's''.  Must there be a subset $S\subseteq
N^*$ such that all three axioms are satisfied?  (Hint: Use a least
fixed point from ${\cal P}(N^*)$).
\end{exer}
\begin{exer}\label{4.23}
Let $f:\domd\far\domd$ be an approximable map.  Let
$a_n:\domd\far\domd$ be a sequence of approximable maps such that
\begin{enumerate}
\item $a_0(x)=\bottom$ for all $x\in\domd$
\item $a_n\appx a_{n+1}$ for all $n\in\nat$
\item $\bigcup\limits_{n=0}^\infty a_n = \ident_D$ in $\domd\far\domd$
\item $a_{n+1}\circ f = a_{n+1}\circ f\circ a_n$ for all $n\in\nat$
\end{enumerate}
Show that $f$ has a unique fixed point. (Hint: Show that if $x=f(x)$
then $a_n(x)\appx a_n(fix(f))$ for all $n\in\nat$.  Show this by
induction on $n$.)
\end{exer}
\newpage
\section{Typed $\lambda$-Calculus}
As shown in the previous section, functions can be characterized by
recursion equations which combine previously defined functions with
the function being defined.  The expression of these functions is
simplified in this section by introducing a notation for specifying a
function without having to give the function a name.  The notation
used is that of the {\it typed\/} $\lambda$-Calculus; a function is
defined using a $\lambda$-{\it abstraction\/}.  

\subsection{Definition of Typed $\lambda$-Calculus}
An informal characterization of the $\lambda$-calculus suffices for
this section; more formal descriptions are available elsewhere in the
literature~\cite{blc}.  Thus, examples are used to introduce the notation.

An infinite number of variables, $x$,$y$,$z$,$\ldots$ of various types
are required.  While a variable has a certain type, type subscripts
will not be used due to the notational complexity.  A distinction
must also be made between type symbols and domains.  The domain
$\dom{A}\times\dom{B}$ does not uniquely determine the
component domains  $\dom{A}$ and $\dom{B}$ even though these domains are
uniquely determined by the symbol for the domain.  The domain is the
{\it meaning\/} that we attribute to the symbol.   

In addition to variables, constants are also present.  For example,
the symbol $0$ is used to represent the zero element from the domain
$\domn$. Another constant, present in each domain by virtue of
Theorem~\ref{4.2}, is $fix^\domd$, the least fixed point operator for
domain $\domd$ of type $(\domd\far\domd)\far\domd$.  The constants and
variables are the atomic (non-compound) terms.  Types can be
associated with all atomic terms.

There are several constructions for compound terms.  First, given
$\tau,\ldots,\sigma$, a list of terms, the {\it ordered tuple\/}
\[\langle \tau,\ldots,\sigma\rangle \]
is a compound term.  If the types of $\tau,\ldots,\sigma$ are
$\dom{A},\ldots,\dom{B}$, the type of the tuple is
$\dom{A}\times\ldots\times\dom{B}$ since the tuple is to be an element
of this domain.  The tuple notation for combining functions given
earlier should be disregarded here.

The next construction is function application.  If the term $\tau$ has
type $\dom{A}\far\dom{B}$ and the term $\sigma$ has the type
$\dom{A}$, then the compound term
\[\tau(\sigma)\]
has type $\dom{B}$.  Function application denotes the value of a
function at a given input.  The notation
$\tau(\sigma_0,\ldots,\sigma_n)$  abbreviates
$\tau(\langle \sigma_0,\ldots,\sigma_n\rangle )$.  Functions applied to tuples
allows us to represent applications of multi-variate functions. 

The $\lambda$-abstraction is used to define functions.  Let
$x_0,\ldots,x_n$ be a list of distinct variables of type
$\domd_0,\ldots,\domd_n$.  Let $\tau$ be a term of some type
$\domd_{n+1}$.  $\tau$ can be thought of as a function of $n+1$
variables with type
$(\domd_0\times\ldots\times\domd_n)\far\domd_{n+1}$.  The name 
for this function is written
\[\lambda x_0,\ldots,x_n.\tau\]
This expression denotes the entire function.  To look at some familiar
functions in the new notation, consider
\[\lambda x,y.x\]
This notation is read ``lambda ex wye (pause) ex''.  If the types of
$x$ and $y$ are $\dom{A}$ and $\dom{B}$ respectively, the function has
type $(\dom{A}\times\dom{B})\far\dom{A}$.  This function is the first
projection function $p_0$.  This function and the second projection
function can be defined by the following equations:
\[\begin{array}{lcl}
p_0&=&\lambda x,y.x\\
p_1&=&\lambda x,y.y
\end{array}\]
Recalling the function tuple notation introduced in an earlier
section, the following equation holds:
\[\langle f,g\rangle =\lambda w.\langle f(w),g(w)\rangle \]
which defines a function of type $\domd_1\far(\domd_2\times\domd_3)$.

Other familiar functions are defined by the following equations:
\[\begin{array}{lcl}
eval&=&\lambda f,x.f(x)\\
curry&=&=\lambda g\lambda x\lambda y.g(x,y)
\end{array}\]
The $curry$ example shows that this notation can be iterated.  A
distinction is thus made between the terms $\lambda x,y.x$ and
$\lambda x\lambda y.x$ which have the types
$\domd_0\times\domd_1\far\domd_0$ and 
$\domd_0\far\domd_1\far\domd_0$ respectively.  Thus, the following
equation also holds:
\[curry(\lambda x,y.\tau)=\lambda x\lambda y.\tau\]
which relates the multi-variate form to  the iterated or {\it
curried\/} form.  Another true equation is
\[fix={\bf fix}(\lambda F\lambda f.f(F(f)))\]
where $fix$ has type $(\domd\far\domd)\far\domd$ and {\bf fix} has type
\[((((\domd\far\domd)\far\domd)\far((\domd\far\domd)\far\domd))\far((\domd\far\domd)\far\domd))\]
This is the content of Exercise~\ref{4.9}.  

This notation can now be used to define functions using recursion
equations.  For example, the function $\sigma$ in Example~\ref{4.3} can be
characterized by the following equation:
\[\sigma=fix(\lambda f\lambda n.cond(zero(n),0,f(pred(n))+pred(n))\]
which states that $\sigma$ is the least recursively defined function
$f$ whose value at $n$ is $cond(\ldots)$.  The variable $f$ occurs in
the body of the $cond$ expression, but this is just the point of a
recursive definition.  $f$ is defined in terms of its value on
``smaller'' input values. The use of the fixed point operator makes
the definition explicit by forcing there to be a unique solution to
the equation.

In an abstraction $\lambda x,y,z.\tau$, the variables $x$,$y$, and $z$
are said to be {\it bound\/} in the term $\tau$.  Any other variables
in $\tau$ are said to be {\it free\/} variables in $\tau$ unless they
are bound elsewhere in $\tau$.  Bound variables are simply
placeholders for values; the particular variable name chosen is
irrelevant.  Thus, the equation
\[\lambda x.\tau=\lambda y.\tau[y/x]\]
is true provided $y$ is not free in $\tau$.  The notation $\tau[y/x]$
specifies the substitution of $y$ for $x$ everywhere $x$ occurs in
$\tau$.  The notation $\tau[\sigma/x]$ for the substitution of the
term $\sigma$ for the variable $x$ is also legitimate.  

\subsection{Semantics of Typed $\lambda$-Calculus}
To show that the equations above with $\lambda$--terms are indeed
meaningful, the following theorem relating $\lambda$--terms and
approximable mappings must be proved.
\begin{thm}\label{5.1}
Every typed $\lambda$--term defines an approximable function of its
free variables.  
\end{thm}
\proof Induction on the length of the term and its structure will be
used in this proof.  
\begin{description}
\item[Variables] Direct since $x\mapsto x$ is an approximable function.
\item[Constants] Direct since $x\mapsto k$ is an approximable function
for constant $k$.
\item[Tuples] Let $\tau=\langle \sigma_0,\ldots,\sigma_n\rangle $.  Since the
$\sigma_i$ terms are less complex, they are approximable functions of
their free variables by the induction hypothesis.  Using
Theorem~\ref{3.4} (generalized to the multi-variate case)  then,
$\tau$ which takes tuples as values also defines an approximable
function. 
\item[Application] Let $\tau=\sigma_0(\sigma_1)$.  We assume that the
types of the terms are appropriately matched.  The $\sigma_i$ terms
define approximable functions again by the induction hypothesis.
Recalling the earlier equations, the value of $\tau$ is the same as
the value of $eval(\sigma_0,\sigma_1)$.  Since $eval$ is approximable,
Theorem~\ref{3.7} shows that the term defines an approximable
function.
\item[Abstraction] Let $\tau=\lambda x.\sigma$.  By the induction
hypothesis, $\sigma$ defines a function of its free variables.  Let
those free variables be of types $\domd_0,\ldots,\domd_n$ where
$\domd_n$ is the type of $x$.  Then $\sigma$ defines an approximable
function
\[g:\domd_0\times\ldots\times\domd_n\far\domd'\]
 where $\domd'$
is the type of $\sigma$.  Using Theorem~\ref{3.12}, the function
\[curry(g):\domd_0\times\ldots\times\domd_{n-1}\far(\domd_n\far\domd')\]
yields an approximable function, but this is just the function defined
by $\tau$.
The reader can generalize this proof for multiple bound variables in
$\tau$. 
\end{description}
\eop

Given this, the equation $\tau=\sigma$ states that the two terms
define the same approximable function of their free variables. As an
example,
\[\lambda x.\tau=\lambda y.\tau[y/x]\]
provided $y$ is not free in $\tau$ since the generation of the
approximable function did not depend on the name $x$ but only on its
location in $\tau$.  Other equations such as these are given in the
exercises.  The most basic rule is shown below.
\begin{thm}\label{5.2} For appropriately typed terms, the following equation is
true:
\[(\lambda x_0,\ldots,x_{1}.\tau)(\sigma_0,\ldots,\sigma_{n-1})=
\tau[\sigma_0/x_0,\ldots,\sigma_{n-1}/x_{n-1}]\]
\end{thm}
\proof The proof is given for $n=1$ and proceeds again by induction on
the length of the term and the structure of the term. 
\begin{description}
\item[Variables] This means $(\lambda x.x)(\sigma)=\sigma$ must be
true which it is.
\item[Constants] This requires $(\lambda x.k)(\sigma)= k$ must be true
which it is for any constant $k$.
\item[Tuples] Let $\tau=\langle \tau_0,\tau_1\rangle $.  This requires that
\[(\lambda x.\langle \tau_0,\tau_1\rangle )(\sigma) =
\langle \tau_0[\sigma/x],\tau_1[\sigma/x]\rangle \]
must be true. This equation holds since the left-hand side can be
transformed using the following true equation:
\[(\lambda x.\langle \tau_0,\tau_1\rangle )(\sigma) =
\langle (\lambda x.\tau_0)(\sigma),(\lambda x.\tau_1)(\sigma)\rangle \]
Then the inductive hypothesis is applied to the $\tau_i$ terms.
\item[Applications] Let $\tau=\tau_0(\tau_1)$.  Then, the result
requires that the equation
\[(\lambda x.\tau_0(\tau_1))(\sigma) =
\tau_0[\sigma/x](\tau_1[\sigma/x])\]
hold true. To see that this is true, examine the approximable
functions for the left-hand side of the equation.
\[\begin{array}{lcl}
\tau_0&\mapsto&\bar{V},x\far t_0\\
\tau_1&\mapsto&\bar{V},x\far t_1\\
\sigma&\mapsto&\bar{V}\far s\\
\mbox{so}\\
(\lambda x.\tau_0(\tau_1))(\sigma)&\mapsto&\bar{V}\far[(x\far
t_0(t_1))(s)]\\
&=&\bar{V},x\far[(x\far t_0)(s)]([(x\far t_1)(s)])
\end{array}\]
From this last term, we use the induction hypothesis.  To see why the
last step holds, start with the set representing the left-hand side
and using the aprroximable mappings for the terms:
\[\begin{array}{cl}
&(\lambda x.\tau_0(\tau_1))(\sigma)\\
\mapsto&\bar{V}\far[(x\far t_0(t_1))(s)]\\
=&\{b\such \exists a.a\in s\mand a\:[x\far t_0(t_1)]\: b\}\\
=&\{b\such \exists a. a\in s\mand a\:\{(x,u)\such v\in x\far t_1)\mand
v\:(x\far t_0)\: u\}\: b\}\\
=&\{b\such \exists a.a\in s\mand v\in (x\far t_1)(a)\mand v\:
(x\far t_0)(a)\:b\}\\
=&\{b\such \exists a,c.a\in s\mand a\:(x\far t_1)\: v\mand a\:(x\far
t_0)\:c\mand v\:c\:b\}\\
=&\{b\such v\in [(x\far t_1)(s)]\mand c\in (x\far t_0)(s) \mand
v\:c\:b\}\\
=&\{b\such v\in [(x\far t_1)(s)]\mand v\:[(x\far t_0)(s)]\:b\}\\
=&[(x\far t_0)(s)]([(x\far t_1)(s)])
\end{array}\]
\item[Abstractions] Let $\tau=\lambda y.\tau_0$.  The required
equation is
\[(\lambda x.\lambda y.\tau_0)(\sigma)=\lambda y.\tau_0[\sigma/x]\]
provided that $y$ is not free in $\sigma$.  The following true
equation applies here:
\[(\lambda x.\lambda y.\tau)(\sigma)=\lambda y.((\lambda
x.\tau)(\sigma))\]
To see that this equation holds, let $g$ be a function of $n+2$ free
variables defined by $\tau$.  By Theorem~\ref{5.1}, the term $\lambda
x.\lambda y.\tau$ defines the function $curry(curry(g))$ of $n$
variables. Call this function $h$.  Thus,
\[h(v)(\sigma)(y) = g(v,\sigma,y)\]
where $v$ is the list of the other free variables.  Using a combinator
$inv$ which inverts the order of the last two arguments, 
\[h(v)(\sigma)(y)=curry(inv(g))(v,y)(\sigma)\]
But, $curry(inv(g))$ is the function defined by $\lambda x.\tau$.  Thus,
we have shown that
\[(\lambda x.\lambda y.\tau)(\sigma)(y)=(\lambda x.\tau)(\sigma)\]
is a true equation.  If $y$ is not free in $\alpha$ and
$\alpha(y)=\beta$ is true, then $\alpha=\lambda y.\beta$ must also be
true.
\end{description}
 \eop

If $\tau'$ is the term $\lambda x,y.\tau$, then $\tau'(x,y)$ is the
same as $\tau$.  This specifies that $x$ and $y$ are not free in
$\tau$. This notation is used in the proof of the following theorem.
\begin{thm}\label{5.3}
The least fixed point of
\[\lambda x,y.\langle \tau(x,y),\sigma(x,y)\rangle\]
is the pair with coordinates $fix(\lambda x.\tau(x,fix(\lambda
y.\sigma(x,y))))$ and $fix(\lambda y.\sigma(fix(\lambda
x.\tau(x,y)),y))$. 
\end{thm}
\proof We are thus assuming that $x$ and $y$ are {\bf not} free in
$\tau$ and $\sigma$.  The purpose here is to find the least solution
to the pair of equations:
\[x=\tau(x,y)~{\rm and}~y=\sigma(x,y)\]
This generalizes the fixed point equation to two variables.  More
variables could be included using the same method. Let 
\[y_*=fix(\lambda y.\sigma(fix(\lambda x.\tau(x,y)),y))\]
and
\[x_*=fix(\lambda x.\tau(x,y))\]
Then,
\[x_*=\tau(x_*,y_*)\]
and
\[\begin{array}{lcl}
y_*&=&\sigma(fix(\lambda x.\tau(x,y_*),y_*))\\
&=&\sigma(x_*,y_*).
\end{array}\]
This shows that the pair $\langle x_*,y_*\rangle $ is one fixed point.  Now, let
$\langle x_0,y_0\rangle $ be the least solution. (Why must a least solution exist?
Hint: Consider a suitable mapping of type
$(\domd_x\times\domd_y)\far(\domd_x\times\domd_y)$.) Thus, we know
that $x_0=\tau(x_0,y_0)$, $y_0=\sigma(x_0,y_0)$, and that $x_0\appx
x_*$ and $y_0\appx y_*$.  But this means that $\tau(x_0,y_0)\appx x_0$
and thus $fix(\lambda x.\tau(x,y_0))\appx x_0$ and consequently
\[\sigma(fix(\lambda x.\tau(x,y_0),y_0))\appx \sigma(x_0,y_0)\appx
y_0\]
By the fixed point definition of $y_*$, $y_*\appx y_0$ must hold as
well so $y_0=y_*$.  Thus,
\[x_*=fix(\lambda x.\tau(x,y*))=fix(\lambda x.\tau(x,y_0))\appx x_0.\]
Thus, $x*=x_0$ must also hold.  A similar argument holds for
$x_0$.\eop

The purpose of the above proof is to demonstrate the use of least fixed
points in proofs. The following are also true equations:
\[fix(\lambda x.\tau(x))=\tau(fix(\lambda x.\tau(x)))\]
and
\[\tau(y)\appx y\imp fix(\lambda x.\tau(x))\appx y\]
if $x$ is not free in $\tau$.  These equations combined with the
monotonicity of functions were the methods used in the proof above.
Another example is the proof of the following theorem.
\begin{thm}\label{5.4}
Let $x$,$y$, and $\tau(x,y)$ be of type $\domd$ and let
$g:\domd\far\domd$ be a function.  Then the equation
\[\lambda x.fix(\lambda y.\tau(x,y))=fix(\lambda g.\lambda
x.\tau(x,g(x)))\]
holds.
\end{thm}
\proof Let $f$ be the function on the left-hand side.  Then,
\[f(x)=fix(\lambda y.\tau(x,y))=\tau(x,f(x))\]
holds using the equations stated above.  Therefore,
\[f=\lambda x.\tau(x,f(x))\]
and thus 
\[g_0=fix(\lambda g.\lambda x.\tau(x,g(x)))\appx f. \]
By the definition of $g_0$ we have
\[g_0(x)=\tau(x,g_0(x))\]
for any given $x$.  By the definition of $f$ we find that
\[f(x)=fix(\lambda y.\tau(x,y))\appx g_0(x)\]
must hold for all $x$.  Thus $f\appx g_0$ and the equation is
true.\eop

This proof illustrates the use of inclusion and equations between
functions.  The following principle was used:
\[(\forall x.\tau\appx\sigma)\imp \lambda x.\tau\appx\lambda
x.\sigma\]
This is a restatement of the first part of Theorem~\ref{3.13}.  

\subsection{Combinators and Recursive Functions}
Below is a list of various combinators with their definitions in
$\lambda$-notation.  The meanings of those combinators not previously
mentioned should be clear.
\[\begin{array}{lcl}
p_0&=&\lambda x,y.x\\
p_1&=&\lambda x,y.y\\
pair&=&\lambda x.\lambda y.\langle x,y\rangle \\
n-tuple&=&\lambda x_0\lambda\ldots\lambda
x_{n-1}.\langle x_0,\ldots,x_{n-1}\rangle\\
diag&=&\lambda x.\langle x,x\rangle \\
funpair&=&\lambda f.\lambda g.\lambda x.\langle f(x),g(x)\rangle \\
proj^n_i&=&\lambda x_0,\ldots,x_{n-1}.x_i\\
inv^n_{i,j}&=&\lambda x_0,\ldots,x_i,\ldots,x_j,\ldots,x_{n-1}.\langle x0,\ldots,x_j,\ldots,x_i,\ldots,x_{n-1}\rangle \\
eval&=&\lambda f,x.f(x)\\
curry&=&\lambda g.\lambda x.\lambda y.g(x,y)\\
comp&=&\lambda f,g.\lambda x.g(f(x))\\
const&=&\lambda k.\lambda x.k\\
{\bf fix}&=&\lambda f.fix(\lambda x.f(x))
\end{array}\]
These combinators are actually {\it schemes} for combinators since no types
have been specified and thus the equations are ambiguous.  Each scheme
generates an infinite number of combinators for all the various types.

One interest in combinators is that they allow expressions without
variables---if enough combinators are used.  This is useful at times
but can be clumsy.  However, defining a combinator when the same
combination of symbols repeatedly appears is also useful.  

There are some familiar combinators that do not appear in the table.
Combinators such as $cond$, $pred$, and $succ$ cannot be defined in the pure
$\lambda$-calculus but are instead specific to certain domains.  They
are thus regarded as primitives.  A large number of other functions
can be defined using these primitives and the $\lambda$-notation, as
the following theorem shows.
\begin{thm}\label{5.6}
For every partial recursive function $h:\nat\far\nat$, there is a
$\lambda$-term $\tau$ of type $\dom{N}\far\dom{N}$ such that the only constants
occurring in $\tau$ are $cond$, $succ$, $pred$, $zero$, and $0$ and if
$h(n)=m$ then $\tau(n)=m$.  If $h(n)$ is undefined, then
$\tau(n)=\bottom$ holds.  $\tau(\bottom)=\bottom$ is also true.
\end{thm}
\proof It is convenient in the proof to work with strict functions
$f:\domn^k\far\domn$ such that if any input is $\bottom$, the result
of the function is $\bottom$.  The composition of strict functions is
easily shown to be strict.  It is also easy to see that any partial
function $g:\nat^k\far\nat$ can be extended to a strict approximable
function $\bar{g}:\domn^k\far\domn$ which yields the same values on
inputs for which $g$ is defined.  Other input values yield $\bottom$.
We want to show that $\bar{g}$ is definable with a
$\lambda$-expression. 

First we must show that {\it primitive\/} recursive functions have
$\lambda$-definitions.  Primitive recursive functions are formed from
starting functions using composition and the scheme of primitive
recursion.  The starting functions are the constant function for zero
and the identity and projection functions.  These functions, however,
must be strict so the term $\lambda x,y.x$ is not sufficient for a
projection function.  The following device reduces a function to its
strict form.  Let $\lambda x.cond(zero(x),x,x)$ be a function with $x$
of type $\domn$.  This is the strict identity function.  The strict
projection function attempted above can be defined as
\[\lambda x,y.cond(zero(y),x,x)\]
The three variable projection function can be defined as
\[\lambda x,y,z.cond(zero(x),cond(zero(z),y,y),cond(zero(z),y,y))\]

While not very elegant, this device does produce strict functions.
Strict functions are closed under substitution and composition.  Any
substitution of a group of functions into another function can be
defined with a $\lambda$-term if the functions themselves can be so
defined. Thus, we need to show that functions obtained by primitive
recursion are definable.  Let $f:\domn\far\domn$, and
$g:\domn^3\far\domn$ be total functions with $\bar{f}$ and $\bar{g}$
being $\lambda$-definable.  We obtain the function
$h:\domn^2\far\domn$ by primitive recursion where
\[\begin{array}{lcl}
h(0,m)&=&f(m)\\
h(n+1,m)&=&g(n,m,h(n,m))
\end{array}\]
for all $n,m\in\domn$.  The $\lambda$-term for $\bar{h}$ is
\[fix(\lambda k.\lambda
x,y.cond(zero(x),\bar{f}(y),\bar{g}(pred(x),y,k(pred(x),y))))\]
Note that the fixed point operator for the domain $\domn^2\far\domn$
was used.  The variables $x$ and $y$ are of type $\domn$.  The $cond$
function is used to encode the function requirements.  The fixed point
function is easily seen to be strict and this function is $\bar{h}$.  

Primitive recursive functions are now $\lambda$-definable.  To obtain
{\it partial\/} ({\it i.e.}, general) recursive functions, the $\mu$-scheme (the least
number operator) is used.  Let $f(n,m)$ be a primitive recursive
function. Then, define $h$, a partial function, as $h(m) =$ the least
$n$ such that $f(n,m)=0$.  This is written as $h(m)=\mu n.f(n,m)=0$.
Since $\bar{f}$ is $\lambda$-definable as has just been shown, let
\[\bar{g}=fix(\lambda g.\lambda
x,y.cond(zero(\bar{f}(x,y)),x,g(succ(x),y)))\]
Then, the desired function $\bar{h}$ is defined as $\bar{h}=\lambda
y.\bar{g}(0,y)$.  It is easy to see that this is a strict function.
Note that, if $h(m)$ is defined, clearly $h(m)=\bar{g}(0,m)$ is also
defined.  If $h(m)$ is undefined, it is also true that
$\bar{g}(0,m)=\bottom$ due to the fixed point construction but it is
less obvious. This argument is left to the reader.\eop

Theorem~\ref{5.6} does not claim that all $\lambda$-terms define
partial recursive functions although this is also true.
Further examples of recursion are found
in the exercises.
\vspace{.25in}
\vspace{.25in}
\vspace{.25in}
\begin{center}
{\bf Exercises}
\end{center}
\begin{exer}\label{5.7}
Find the definitions of 
\[\lambda x,y.\tau~{\rm and}~\sigma(x,y)\]
which use only $\lambda v$ with one variable and applications only to
one argument at a time.  Note that use must be made of the combinators
$p_0$, $p_1$, and $pair$.  Generalize the result to functions of many
variables. 
\end{exer}
\begin{exer}\label{5.8}
The table of combinators was meant to show how combinators could be
defined in terms of $\lambda$-expressions.  Can the tables be turned
to show that, with enough combinators available, every
$\lambda$-expression can be defined by combining combinators using
application as the only mode of combination?
\end{exer}
\begin{exer}\label{5.9}
Suppose that $f,g:\domd\far\domd$ are approximable and $f\circ
g=g\circ f$.  Show that $f$ and $g$ have a least common fixed point
$x=f(x)=g(x)$.  (Hint: See Exercise~\ref{4.20}.)  If, in addition,
$f(\bottom)=g(\bottom)$, show that $fix(f)=fix(g)$.  Will
$fix(f)=fix(f^2)$? What if the assumption is weakened to $f\circ
g=g^2\circ f$?
\end{exer}
\begin{exer}\label{5.11}
For any domain $\domd$, $\domd^\infty$ can be regarded as consisting
of bottomless stacks of elements of $\domd$.  Using this view, define
the following combinators with their obvious meaning:
$head:\domd^\infty\far\domd$, $tail:\domd^\infty\far\domd^\infty$ and
$push:\domd\times\domd^\infty\far\domd^\infty$.  Using the fixed point
theorem, argue that there is a combinator $diag:\domd\far\domd^\infty$
where for all $x\in\domd$, $diag(x)=\langle x\rangle _{n=0}^\infty$.  (Hint: Try a
recursive definition, such as 
\[diag(x)=push(x,diag(x))\]
but be sure to prove that all terms of $diag(x)$ are $x$.)  Also
introduce by an appropriate recursive definition a combinator
$map:(\domd\far\domd)^\infty\times\domd\far\domd^\infty$  where for
elements of the proper type
\[map(\langle f_n\rangle _{n=0}^\infty,x)=\langle f_n(x)\rangle _{n=0}^\infty\]
\end{exer}
\begin{exer}\label{5.12}
For any domain $\domd$ introduce, as a least fixed point, a combinator
\[while:(\domd\far\dom{T})\times(\domd\far\domd)\far(\domd\far\domd)\]
by the recursion equation 
\[while(p,f)(x)=cond(p(x),while(p,f)(f(x)),x)\]
Prove that 
\[while(p,while(p,f))=while(p,f)\]
Show how $while$ could be used to obtain the least number operator,$\mu$,
mentioned in the proof of Theorem~\ref{5.6}.  Generalize this idea to
define a combinator
\[find:\domd^\infty\times(\domd\far\dom{T})\far\domd\]
which means ``find the first term in the sequence (if any) which
satisfies the given predicate''.
\end{exer}
\begin{exer}\label{5.13}
Prove the existence of a one-one function
$num:\nat\times\nat\leftrightarrow\nat$ such that
\[\begin{array}{lcl}
num(0,0)&=&0\\
num(n,m+1)&=&num(n+1,m)\\
num(n+1,0)&=&num(0,n)+1
\end{array}\]
Draw a descriptive picture (an infinite matrix) for the function.
Find a closed form for the values if possible.  Use the function to
prove the isomorphism between
${\cal P}(\nat)$,${\cal P}(\nat\times\nat)$, and ${\cal
P}(\nat)\times{\cal P}(\nat)$.
\end{exer}
\begin{exer}\label{5.14}
Show that there are approximable mappings
\[graph:({\cal P}(\nat)\far{\cal P}(\nat))\far{\cal P}(\nat)\]
and
\[fun:{\cal P}(\nat)\far({\cal P}(\nat)\far{\cal P}(\nat))\]
where $fun\circ graph = \lambda f.f$ and $graph\circ fun\appx \lambda
x.x$. (Hint: Using the notation
$[n_0,\ldots,n_k]=num(n_0,[n_1,\ldots,n_k])$, two such combinators can be
given by the formulas
\[\begin{array}{lcl}
fun(u)(x)&=&\{m\such \exists n_0,\ldots,n_{k-1}\in
x.[n_0+1,\ldots,n_{k-1}+1,0,m]\in u\}\\
graph(f)&=&\{[n_0+1,\ldots,n_{k-1}+1,0,m]\such m\in
f(\{n_0,\ldots,n_{k-1}\})\}
\end{array}\]
where $k$ is a variable - meaning all finite sequences are to be
considered.)
\end{exer}
\def\llub{\lub}
\def\Lub{\bigsqcup}
\newpage
\section{Introduction to Domain Equations}\label{sec6}
As stressed in the introduction, the notion of computation with
potentially infinite elements is an integral part of domain theory.
The previous sections have defined the notion of functions over
domains, as well as a notation for expressing these functions.  In
addition, the notion of computation through series of
approximations has been addressed.  This computation is possible since
the functions defined have been approximable and thus continuous.
This section addresses the construction of more complex domains with
infinite elements.  The next section looks specifically at the notion
of computability with respect to these infinite elements.  The last
section looks at another approach to domain construction.

New domains have been constructed from
existing ones using domain constructors such as the product
construction ($\times$), the function space construction ($\far$) and
the sum construction ($+$) of Exercise~\ref{3.18}.  These
constructors can be {\it iterated\/} similar to the way that
function application was iterated to form recursive
function definitions.  In this way, domains can be characterized using
recursion equations, called {\it domain equations\/}. 

\subsection{Domain Equations}
A domain equation represents an isomorphism between the domain as a 
whole and the combination of domains that comprise it.  These recursive
domains are frequently termed {\it reflexive\/} domains since, as in the
following example, the domain contains a copy of itself in its structure. 

\begin{exm}\label{6.1n}
Consider the following domain equation:
\[\dom{T}=\dom{A}+(\dom{T}\times\dom{T})\]
where $\dom{A}$ is a previously defined domain.  This domain can be
thought of as containing atomic elements from $\dom{A}$ or pairs of
elements of $\dom{T}$.  
\begin{itemize}
	\item What do the elements of this domain look like?  In particular, what
	are the finite elements of this domain?  \item How is the domain
	constructed?  
	\item What is an appropriate approximation ordering for the
	domain?  
	\item What do lubs in this domain look like?  \item What is the
	appropriate notion of consistency?  
	\item Does this domain even
	exist? In other words, are we certain a solution to this domain equation exists?
	\item And if a solution to the equation exists, is it a unique solution?
\end{itemize}
Each of these questions is examined below.

The domain equation tells us that an element of the domain is either
an element from $\dom{A}$ or is a pair of ``smaller'' elements from
$\dom{T}$.  One method of constructing a sum domain is using pairs
where some distinguished element denotes what type an element is.
Thus, for some $a\in\dom{A}$, the pair $\langle \pi,a\rangle $ might
represent the element in $\dom{T}$ for the given element $a$. For some
$s,t\in\dom{T}$, the pair $\langle \langle s,t\rangle ,\pi\rangle $
might then represent the 
element in $\dom{T}$ for the pair $s,t$.
Thus, $\pi$ is the distinguished element, and the location of $\pi$ in
the pair specifies the type of the element.   The finite elements
are either elements in $\dom{T}$ representing the (finite) elements of
$\dom{A}$ or the pair elements from $\dom{T}$ whose components are
also finite elements in $\dom{T}$.  

The question then arises about infinite elements.  Are there infinite
elements in this domain?  Consider the following fixed point equation
for some element for $a\in\dom{A}$:
\[x=\langle \langle a,x\rangle ,\pi\rangle .\]
The fixed point of this equation is the infinite product of the
element $a$.  Does this element fit the definition for $\dom{T}$?
From the informal description of the elements of $\dom{T}$ given so
far, $x$ does qualify as a member of $\dom{T}$.  

Now that some intuition has been developed about
this domain, a formal construction is required.  Let
$\langle \bas{A},\appx_A\rangle $ be the finitary basis used to generate the domain
$\dom{A}$. Let $\pi$ be an object such that $\pi\notin\bas{A}$.
Define the bottom element of the finitary basis \bas{T} as
$\Delta_T=\langle \pi,\pi\rangle $. Next, all the elements of $\dom{A}$ must be
included so define an element in $\bas{T}$ for each $a\in\bas{A}$ as
$\langle \pi,a\rangle $. Finally, pair elements for all elements in $\bas{T}$ must
exist in $\bas{T}$ to complete the construction.  The set $\bas{T}$
can be defined inductively as the {\it least\/} set such that:
\begin{enumerate}
\item $\Delta_T\in\bas{T}$,
\item $\langle \pi,a\rangle \in\bas{T}$ whenever $a\in\bas{A}$,
\item $\langle \langle \Delta_T,s\rangle ,\pi\rangle \in\bas{T}$ whenever $s\in\bas{T}$ (necessary??),
\item $\langle \langle t,\Delta_T\rangle ,\pi\rangle \in\bas{T}$ whenever $t\in\bas{T}$ (necessary??), and
\item $\langle \langle t,s\rangle ,\pi\rangle \in\bas{T}$ whenever $s, t\in\bas{T}$.
\end{enumerate}
The set can also be characterized by the following fixed point
equation:
\[\bas{T}=\{\Delta_T\}\cup\{\langle \pi,a\rangle \such
a\in\bas{A}\}\cup\{\langle \langle \Delta_T,s\rangle ,\pi\rangle \such
s\in\bas{T}\}\cup\{\langle \langle t,\Delta_T\rangle ,\pi\rangle \such t\in\bas{T}\}\cup\{\langle \langle t,s\rangle ,\pi\rangle \such s,t\in\bas{T}\}.\]
A solution must exist for this equation by the fixed point theorem.

Now that the basis elements have been defined, we must show how to
find lubs.  We will again use an inductive definition.
\begin{enumerate}
\item $\langle\pi,\pi\rangle \llub t=t$ for all $t\in \bas{T}$
\item For $a,b\in \bas{A}$,
$\langle\pi,a\rangle\llub\langle\pi,b\rangle=\langle\pi,a\llub
b\rangle$ if $a\llub b$ exists in $\bas{A}$
\item $\langle\langle s,t\rangle,\pi\rangle\llub\langle\langle
s',t'\rangle,\pi\rangle=\langle\langle s\llub s',t\llub
t'\rangle,\pi\rangle$ if $s\llub s'$ and $t\llub t'$ exist in
$\bas{T}$.
\item The lub $\langle\pi,a\rangle\llub\langle\langle
s,t\rangle,\pi\rangle$ does not exist.
\end{enumerate}

Next, the notion of consistency needs to be explored.  From the
definition of lubs given above, the following sets are consistent:
\begin{enumerate}
\item The empty set is consistent.
\item Everything is consistent with the bottom element.
\item A set of elements all from the basis \bas{A} is consistent in
\bas{T} if the set of elements is consistent in \bas{A}.
\item A set of product elements in \bas{T} is consistent if the left
component elements are consistent and the right component elements are
consistent. 
\end{enumerate}
These conditions derive from the sum and product nature of the domain.

The approximation ordering in the basis has the following inductive
definition:
\begin{enumerate}
\item $\Delta_T\appx_T s$ for all $s\in\bas{T}$
\item $y\appx_Tu\llub\Delta_T$ whenever $y\appx_Tu$
\item $\langle \pi,a\rangle \appx_T\langle \pi,b\rangle $ whenever $a\appx_Ab$
\item $\langle \langle s,t\rangle ,\pi\rangle \appx_T\langle \langle u,v\rangle ,\pi\rangle $ whenever $s\appx_Tu$ and
$t\appx_Tv$
\end{enumerate}

The next step is to verify that $\bas{T}$ is indeed a finitary basis.
The basis is still countable.
The approximation is clearly a partial order.  The existence of
lubs of finite bounded ({\it i.e.}, consistent) subsets must be verified.  The
definition of consistency gives us the requirements for a bounded
subset.  Each of the conditions for consistency are examined inductively
since the definitions are all inductive:
\begin{enumerate}
\item The lub of the empty set is the bottom element
$\Delta_T$. 
\item The lub of a set containing the bottom element is the lub of the
set without the bottom element which must exist by the induction
hypothesis.  
\item The lub of a set of elements all from the
$\bas{A}$ is the element in $\bas{T}$ for the lub in $\bas{A}$.  This
element must exist since $\bas{A}$ is a finitary basis and all
elements from $\bas{A}$ have corresponding elements in $\bas{T}$.
\item The lub of a set of product elements is the pair of the lub of
the left components and the lub of the right components.  These exist
by the induction hypothesis.
\end{enumerate}
Thus, a finitary basis has been created; the domain is formed as
always from the basis.  The solution to the domain equation has been
found since any element in the domain $\dom{T}$ is either an element
representing an element in $\dom{A}$ or is the product of two other
elements in $\dom{T}$.  Similarly, any element found on the left-hand
side must also be in the domain $\dom{T}$ by the construction.  Thus,
the domain $\dom{T}$ is {\it identical\/} to the domain
$\dom{A}+(\dom{T}\times\dom{T})$. 

To look at the question concerning the existence and uniqueness of
the solution to this domain equation, recall the fixed point theorem.
This theorem states that a fixed point set exists for any approximable
mapping over a domain.

\subsection{Subdomains}
In Section~\ref{univdom}, the concept of a {\it
universal domain\/} is introduced.  A universal domain is a domain
which contains all other domains as sub-domains.  These sub-domains
are, roughly speaking, the image of approximable functions over the
universal domain.  The domain equation for $\dom{T}$ can be viewed as
an approximable mapping over the universal domain.  As such, the fixed
point theorem states that a least fixed point set for the function
does exist and is unique. Sub-domains are defined formally below.

Looking again at the informal discussion concerning the elements of
the domain $\dom{T}$, the infinite element proposed does fit into the
formal definition for elements of $\dom{T}$.  This element
is an infinite tree with all left sub-trees
containing only the element $a$.
For this infinite element to be computable, it must be the lub
 of some ascending chain of finite approximations to it.  The
element $x$ can, in fact, be defined by the following
ascending sequence of finite trees:
\[\begin{array}{lcl}
x_0&=&\bottom\\
x_{n+1}&=&\langle \langle a,x_n\rangle ,\pi\rangle \\
x&=&\Lub^\infty_{n=0}x_n
\end{array}\]

Thus, using domain equations, a domain has been defined recursively.
This domain includes infinite as well as finite elements and allows
computation on the infinite elements to proceed using the finite
approximations, as with the more conventionally defined domains
presented earlier.
\end{exm}

The final topic of this section is the notion of a sub-domain.
Informally, a sub-domain is a structured part of a larger domain.
Earlier, a domain was described as a sub-domain of the universal
domain.  Thus, the sub-domain starts with a subset of the
elements of the larger domain while retaining the approximation
ordering, consistency relation and lub relation,
suitably restricted to the subset elements.

\begin{mdef}{Sub-Domain}\label{6.10}
A domain $\langle \dom{R},\appx_R\rangle $ is a {\it sub-domain\/} of a domain
$\langle \dom{D},\appx_D\rangle $, denoted $\dom{R}\subd\domd$ iff  
\begin{enumerate}
\item $\dom{R}\subseteq\domd$ - The elements of $\dom{R}$ are
a subset of the elements of $\dom{D}$.
\item $\bottom_R=\bottom_D$ - The bottom elements are the same.
\item For $x,y\in\dom{R}$, $x\appx_Ry\iff x\appx_Dy$ - The
approximation ordering  for $\dom{R}$ is the approximation ordering for
$\domd$ restricted to elements in $\dom{R}$.
\item For $x,y,z\in\dom{R}$, $x\llub_Ry=z$  iff $x\llub_Dy=z$  - The
lub relation for $\dom{R}$ is the lub
relation for $\domd$ restricted to elements in $\dom{R}$.
\item $\dom{R}$ is a domain.
\end{enumerate}
\end{mdef}

Equivalently, a sub-domain can be thought of as the image of an
approximable function which approximates the identity function (also
termed a {\it projection\/}).  The notion of a sub-domain is used in
the final section in the discussions about the universal domain. This
mapping between the domains can be formalized as follows:
\begin{thm}\label{6.12}
If $\domd\subd\dom{E}$, then there exists a projection pair of
approximable mappings $i:\domd\far\dom{E}$ and $j:\dom{E}\far\domd$
where $j\circ i={\ident}_\domd$ and $i\circ j\appx {\ident}_\dome$ where $i$
and $j$  are
determined by the following equations:
\[\begin{array}{lcl}
i(x)&=&\{y\in\bas{E}\such \exists z\in x.z\appx y\}\\
j(y)&=&\{x\in\bas{D}\such x\in y\}
\end{array}\]
for all $x\in\domd$ and $y\in\dom{E}$.
\end{thm}
The proof is left as an exercise.

By the definition of a sub-domain, it should be clear that
\[\domd_0\subd\dom{E}\mand\domd_1\subd\dom{E}\imp(\domd_0\subd\domd_1\iff\domd_0\subseteq\domd_1)\]
Using this observation, the sub-domains of a domain can be ordered.
Indeed, the following theorem is a consequence of this ordering.
\begin{thm}\label{6.11}
For a given domain $\domd$, the set of sub-domains $\{\domd_0\such
\domd_0\subd\domd\}$ form a domain.
\end{thm}
The proof proceeds using the inclusion relation defined as an
approximation ordering and is left as an exercise.

Finally, a converse of Theorem~\ref{6.12} can also be established:
\begin{thm}\label{6.15}
For two domains $\domd$ and $\dome$, if there exists a projection pair
$i:\domd\far\dome$ and $j:\dome\far\domd$ with $j\circ i=\ident_\domd$ and
$i\circ j\appx \ident_\dome$, then $\exists\domd'\subd\dome$ where $\domd\iso\domd'$.
\end{thm}
\proof We show that $i$ maps finite elements to finite
elements and that $\domd'$ is the image of $\domd$ in $\dome$.  

For some $x\in\bas{D}$ with $\ideal_x$ as the principal ideal of $x$, we
can write 
\[i(\ideal_x)=\lub \{\ideal_y\such y\in i(\ideal_x)\}\]
Applying $j$ to both sides we get
\[\ideal_x=j\circ i(\ideal_x)=\lub \{j(\ideal_y)\such y\in i(\ideal_x)\}\]
since $j\circ i=\ident_D$ and $j$ is continuous by assumption.  But, since
$x\in \ideal_x$, $x\in 
j(\ideal_y)$ for some $y\in i(\ideal_x)$.  This means that
\[\ideal_x\subseteq j(\ideal_y)\]
and thus
\[i(\ideal_x)\subseteq i\circ j(\ideal_y) \subseteq \ideal_y\]
Since $\ideal_y\subseteq i(\ideal_x)$ must hold by the construction, $i(\ideal_x) =
\ideal_y$. This proves that finite elements are mapped to finite elements. 

Next, consider the value for $i(\bottom_D)$.  Since $\bottom_D\appx_D
j(\bottom_E)$, $i(\bottom_D)\appx \bottom_E$.  Thus
$i(\bottom_D)=\bottom_E$. Thus, $\domd$ is isomorphic to the image of
$i$ in $\dome$.  We still must show that $\domd'$ is a domain.  Thus,
we need to show that if a lub exists in $\dome$ for a
finite subset in $\domd'$, then the lub is also in
$\domd'$.  
Let $y',z'\in\bas{D}'$ and $y'\llub z'=x'\in\bas{E}$.  Then, there exists
$y,z\in\bas{D}$ such that $i(\ideal_y)=\ideal_{y'}$ and $i(\ideal_z)=\ideal_{z'}$ 
which implies that $\ideal_y=j(\ideal_{y'})$ and $\ideal_z=j(\ideal_{z'})$. 
Since $\ideal_{y'}\appx \ideal_{x'}$ and $j(\ideal_{y'})\appx j(\ideal_{x'})$ by
monotonicity, $y\in j(\ideal_{x'})$ must hold.  By the same reasoning,
$z\in j(\ideal_{x'})$.  But then $x=y\llub z\in j(\ideal_{x'})$ must also hold
and thus $y\llub z \in \domd$ since the element $j(\ideal_{x'})$ must be an
ideal. But,
\[\begin{array}{lcl}
\ideal_y\appx \ideal_x&\imp& \ideal_{y'}\appx i(\ideal_{x})\\
\ideal_z\appx \ideal_x&\imp& \ideal_{z'}\appx i(\ideal_{x})
\end{array}\]
This implies that $y'\llub z'=x'\in i(\ideal_x)$.  We already know that
$x\in j(\ideal_{x'})$ so $i(\ideal_x)\appx \ideal_{x'}$.  Thus, $i(\ideal_x)=\ideal_{x'}$ and
thus, $x'\in\bas{D'}$.\eop
\vspace{.25in}
\begin{center}
{\bf Exercises}
\end{center}
\begin{exer}\label{6.24}
Show that there must exist domains satisfying 
\[\begin{array}{lcll}
\doma&=&\doma+(\doma\times\domb)&{\rm and}\\
\domb&=&\doma+\domb
\end{array}\]
Decide what the elements will look like and define $\doma$ and $\domb$
using simultaneous fixed points.
\end{exer}
\begin{exer}
Prove Theorem~\ref{6.11}
\end{exer}
\begin{exer}
Prove Theorem~\ref{6.12}
\end{exer}
\begin{exer}\label{6.28}
Show that if $\doma$ and $\domb$ are finite systems, that
\[\domd\unlhd\dom{E}\unlhd\domd\imp \domd\iso\dom{E}\]
where $\domd\iso\domd'$ and $\domd'\subd\dom{E}$ is denoted
$\domd'\unlhd\dom{E}$. 
\end{exer}
\newpage
\section{Computability in Effectively Given Domains}\label{sec7}
In the previous sections, we gave considerable emphasis to the notion
of computation using increasingly accurate approximations of the input
and output. 
This section defines this notion of computability more formally.  
In Section 5, we found that partial functions over the natural
numbers were expressible in the $\lambda$-notation.  This relationship
characterizes computation for a particular domain.  To describe computation
over domains in general, a broader definition is required.

The way a domain is presented impacts the way computations are
performed over it. Indeed, 
the theorems of recursive function theory~\cite{rogers} rely in part on the normal
presentation of the natural numbers.  A presentation for a domain is
an enumeration of the elements of the domain.  The standard
presentation of the natural numbers is simply the numbers in ascending
order beginning with 0. There are many permutations of the natural
numbers, each of which can be considered a presentation.  Computation
with these non-standard presentations may be impossible; that is a
computable function on the standard presentation may be non-computable
over a non-standard presentation.  Therefore, an {\it effective
presentation for a domain\/} is defined as a presentation which makes
the required information computable. 

\subsection{Effective Presentations}
Information about elements in a
domain can be characterized completely by looking at the finite
elements and their relationships. Thus a
presentation must enumerate the finite elements and allow the
consistency and lub relationships on these elements to
be computed to allow this style of computation.

 The consistency relation and the lub 
relation depend on each other.  For example, if a set of elements
is consistent, a lub must exist for the set.  Given that
a set is consistent, the lub can be found in finite time
by just enumerating the elements and checking to see if this element
is the lub.  However, if the set is inconsistent, the
enumeration will not reveal this fact.  Thus, the consistency relation
must be assumed to be recursive in an effective presentation.
Exercise~\ref{7.13} provides a description of presentations that
should clarify the assumptions made.
Formally, a presentation is defined as follows: 
\begin{mdef}{Effective Presentation}\label{7.1}
The {\it presentation\/} of a finitary basis \bas{D} is a function
$\pi:\nat\far\bas{D}$ such that $\pi(0)=\Delta_D$ and the range of
$\pi$ is the set of finite elements of \bas{D}. The definition holds
for a domain $\domd$ as well.

A presentation $\pi$ is {\it effective\/} iff
\begin{enumerate}
\item The consistency relation ($\exists
k.\pi_i\appx\pi_k\mand\pi_j\appx\pi_k$) for elements $\pi_i$ and
$\pi_j$ is recursive\footnote{Recursive in this context means that the
relation is decidable.} over $i$ and $j$.
\item The lub relation ($\pi_k=\pi_i\llub\pi_j$) is
recursive over $i$, $j$, and $k$.
\end{enumerate}
\end{mdef}

This definition supports our intuition about
domains; we have stated that the important information about a domain
is the set of finite elements, the ordering and consistency
relationships between the elements and the lub relation.  
Thus, an effective presentation provides, in a suitable (that is
computable) form, the basic information about the structure and
elements of a domain. A presentation can also be viewed as an
enumeration of the elements of the domain with the position of an
element in the
enumeration given by the index corresponding to the integer input for
that element in the presentation function with the 0 element
representing $\bottom$. This perspective is used in
the majority of the proofs.

\subsection{Computability}
Now that the presentation of a domain has been formalized, the notion
of computability can be formally defined.  Thus,
\begin{mdef}{Computable Mappings}\label{7.2}
Given two domains, $\domd$ and $\dome$ with effective presentations
$\pi_1$ and $\pi_2$ respectively, an approximable mapping
$f:\bas{D}\far\bas{E}$ is {\it computable\/} iff the relation
\[x_n\:f\:y_m\]
is recursively enumerable in $n$ and $m$.
\end{mdef}
 
By considering the domain $\domd$ to be a single element domain, the
above definition applies not only to computable functions but also to
computable elements.  For $d\in\domd$ where $d$ is the only element in
the domain, the element 
\[e=f(d)\in\dome\]
defines an element in $\dome$.  The definition states that $e$ is a
computable iff the set
\[\{m\in\nat\such y_m\appx e\}\]
is a recursively enumerable set of integers.  Clearly if the set of
elements approximating another is finite, the set is recursive.  The
notion of a recursively enumerable set simply requires that all elements
approximating the element in question be listed eventually.  The
computation then proceeds by accepting an enumeration representing the
input element and enumerating the elements that approximate
the desired output element.

Now that the notions of computability and effective presentations have
been formalized, the methods of constructing domains and functions
will be addressed.

The proof of the next theorem is trivial and is left to the reader.
\begin{thm}\label{7.3}
The identity map on an effectively given domain is computable.  The
composition of computable mappings on effectively given domains are
also computable.
\end{thm}
The following corollary is a consequence of this theorem:
\begin{cor}
For computable function $f:\domd\far\dome$ and a computable element
$x\in\domd$, the element $f(x)\in\dome$ is computable.
\end{cor}

In addition, the standard domain constructors maintain effective
presentations.
\begin{thm}\label{7.4}
For domains $\domd_0$ and $\domd_1$ with effective presentations, the
domains 
\[\domd_0+\domd_1~{\rm and}~ \domd_0\times\domd_1\]
are also effectively
given. In addition, the projection functions are all computable.
Finally, if $f$ and $g$ are computable maps, then so are $f+g$ and
$f\times g$.
\end{thm}
\proof Let $\{X_i\such i\in\nat\}$ be the enumeration of $\domd_0$
and $\{Y_i\such i\in\nat\}$ be the enumeration of $\domd_1$. Another
method of sum construction is to use two distinguishing elements in
the first position to specify the element type.  Thus, a sum domain
can be defined as follows:
\[\domd_0+\domd_1=\{(\Delta_0,\Delta_1)\}\cup\{(0,x)\such
x\in\domd_0\}\cup \{(1,y)\such y\in\domd_1\}\]
The enumeration can then be defined as follows for $n\in\nat$:
\[\begin{array}{lcl}
Z_0&=&(\Delta_0,\Delta_1)\\
Z_{2n+1}&=&(0,X_n)\\
Z_{2n+2}&=&(1,Y_n)
\end{array}\]
The proof that $Z_i$ is an effective presentation is left as an
exercise.

For the product construction, the domain appears as follows:
\[\domd_0\times\domd_1=\{(x,y)\such x\in\domd_0,y\in\domd_1\}\]
The enumeration can be defined in terms of the functions
$p:\nat\far\nat$, $q:\nat\far\nat$, and $r:(\nat\times\nat)\far\nat$
where for $m$, $n$, $k\in\nat$:
\[\begin{array}{lcl}
p(r(n,m))&=&n\\
q(r(n,m))&=&m\\
r(p(k),q(k))&=&k
\end{array}\]

Thus, $r$ is a one-to-one pairing function (see Exercise~\ref{5.13})
of which there are several.  The functions $p$ and $q$ extract the
indices from the result of the pairing function.  The enumeration for
the product 
domain is then defined as follows:
\[W_i = (X_{p(i)},Y_{q(i)})\]
The proof that this is an effective presentation is also left as an
exercise. 

For the combinators, the relations will be defined in terms of the
enumeration indices.  For example,
\[\begin{array}{lcl}
X_n\:in_0\:Z_m&\iff& m=0~{\rm or}\\
&&\exists k.m=2k+1\mand X_k\appx X_n\\
W_k\:proj_1\:Y_m&\iff& Y_m\appx Y_{q(k)}
\end{array}\]
The reader should verify that these sets are recursively enumerable.
For this proof, recall that recursively enumerable sets are closed
under conjunction, disjunction, substituting recursive functions, and
applying an existential quantifier to the front of a recursive
predicate.  The proof for the 
other combinators is left as an exercise. \eop

Product spaces formalize the notion of computable functions
of several variables.  Note that the proof of Theorem~\ref{3.7} shows
that substitution of computable functions of severable variables into
other computable functions are still computable.  The next step is to
show that the function space constructor preserves effectiveness.
\begin{thm}\label{7.5}
For domains $\domd_0$ and $\domd_1$ with effective presentations, the
domain $\domd_0\far\domd_1$ also has an effective presentation.  The
combinators $apply$ and $curry$ are computable if all input domains are
effectively given.  The computable elements of the domain
$\domd_0\far\domd_1$ are the computable maps for
$\bas{D_0}\far\bas{D_1}$. 
\end{thm}
\proof Let $\domd_0=\{X_i\such i\in\nat\}$ and $\domd_1=\{Y_i\such
i\in\nat\}$  be the presentations for the domains.  The elements of
$\bas{D_0}\far\bas{D_1}$ are finite step functions which respect the
mapping of some subset of $\bas{D_0}\times\bas{D_1}$.  Given the
enumeration, each element can be associated with a set
\[\{(X_{n_i},Y_{m_i})\such\exists q. 1\leq i\leq q\}\]
Thus, there is a finite set of integers pairs that determine the
element. Given the definition of consistency from Theorem~\ref{3.9}
for elements in the function space domain and the decidability of
consistency in $\domd_0$ and $\domd_1$, consistency of any finite
set of this form is decidable (tedious but decidable since all
elements must be checked with all others, etc).  Since consistency is
decidable, a systematic enumeration of pair sets which are consistent
can be made; this enumeration is simply the enumeration of
$\domd_0\far\domd_1$. Finding the lub consists of making
a finite series of tests to find the element that is the lub, which
must exist since the set is consistent and we have 
closure on lubs of finite consistent subsets.  Finding
the lub requires a finite series of checks in both $\domd_0$
and $\domd_1$ but these checks are decidable. Thus, the lub
 relation is also decidable in $\domd_0\far\domd_1$.  This shows
that $\domd_0\far\domd_1$ is effectively given.

To show that $apply$ and $curry$ are computable, the mappings need to
be examined.  The mapping defined for apply is
\[(F,a)\:apply\: b\iff a\:F\:b\]
The function $F$ is the lub of all the finite step
functions that are consistent with it.  As such, $F$ can be viewed as
the canonical representative of this set.  Since $F$ is a finite step
function, this relation is decidable.  As such, the $apply$ relation
is recursive and not just recursively enumerable and $apply$ is a
computable function.

The reasoning for $curry$ is similar in that the relations are
studied.  Given the increase in the number of domains, the
construction is more tedious and is left for the exercises.

To see that the computable elements correspond to the computable maps,
recall the relationship shown in Theorem~\ref{3.10} between the maps
and the elements in the function space.  Thus, we have
\[a\:f\:b \iff b\in f(\ideal_a)~{\rm or}~\ideal_b\appx f(\ideal_a)\]
Since $f$ is a computable map, we know that the pairs in the map are
recursively enumerable.  Using the previous techniques for deciding
consistency of finite sets, the set of elements consistent with $f$ can be
enumerated.  But this set is simply the ideal for $f$ in the function
space.  The converse direction is trivial. \eop

The final combinator to be discussed, and perhaps the most important,
is the fixed point combinator.
\begin{thm}\label{7.6}
For any effectively given domain, $\domd$, the combinator
$fix:(\domd\far\domd)\far\domd$ is computable.
\end{thm}
\proof Let $\{X_n\such n\in\nat\}$ be the presentation of the  domain
$\domd$. Recall that for $f\in\domd\far\domd$, 
\[f\:fix\:X\iff \exists k\in\nat.
\Delta\:f\:X_1\:f\ldots f\:X_k\mand X_k=X\]
All of the checks in this finite sequence are decidable since $\domd$
is effectively given.  In addition, existential quantification of a
decidable predicate gives a recursively enumerable predicate.  Thus,
$fix$ is computable.  \eop
\subsection{Recap}
Now that this has been formalized, what has been accomplished?  The
major consequence of the theorems to this point is that any expression
over effectively given domains (that is effectively given types)
combined with computable constants 
using the $\lambda$-notation and the fixed point combinator is a
computable function of its free variables.  Such functions, applied to
computable arguments, yield computable values.  These functions also
have computable least fixed points.  All this gives us a
mathematical programming language for defining computable operations.  
Combining this language with the specification of types with domain
equations gives a powerful language.\\

As an example, the effectiveness of the domain $\dom{T}$ from
Example~\ref{6.1n} is studied.  The complete proof is left as an
exercise.
\begin{exm}
Recall the domain $\dom{T}$ from the previous section.  This domain
is characterized by the domain equation
\[\dom{T}=\dom{A}+(\dom{T}\times\dom{T})\]
for some domain $\dom{A}$.  If $\dom{A}$ is effectively given, we wish
to show that $\dom{T}$ is effectively given as well.  The elements
are either atomic elements from $\dom{A}$ or are pairs from $\dom{T}$.
Let $A=\{A_i\such i\in\nat\}$ be the enumeration for $\dom{A}$.  An
enumeration for $\dom{T}$ can be defined as follows:
\[\begin{array}{lcl}
T_0&=&\bottom_T\\
T_{2n+1}&=&3*A_n\\
T_{2n+2}&=&3*T_{p(n)}+1\cup 3*T_{q(n)}+2
\end{array}\]
where for $A$, a set of indices, $m*A+k=\{m*n+k\such n\in A\}$.  The
functions $p$ and $q$ here are the inverses of the pairing function
$r$ defined in Theorem~\ref{7.4}.  These functions must be defined
such that $p(n)\leq n$ and $q(n)\leq n$ so that the recursion is well
defined by taking smaller indices.  The rest of the proof is left to
the exercises.  Specifically, the claim that $\dom{T}=\{T_i\}$ should
be verified as well as the effectiveness of the enumeration.  These
proofs rely either on the effectiveness of $\dom{A}$, on the
effectiveness of elements in $\dom{T}$ with smaller indices, or are
trivial.   
\end{exm}

The final example uses the powerset construction.  We have
repeatedly used the fact that a powerset is a domain.  Its effectiveness is
now verified.
\begin{exm}\label{7.8}
Specifically, the powerset of the natural numbers, ${\cal{P}(\nat)}$ is
considered. In this domain, all elements are consistent, and there is
a top element, denoted $\omega$, which is the set of all natural
numbers.  The ordering is the subset relation.  The lub
 of two subsets is the union of the two subsets, which is
decidable. To enumerate the finite subsets, the following enumeration
is used:
\[E_n=\{k\such \exists i,j. i< 2^k\mand n=i+2^k+j*2^{k+1}\}\]
This says that $k\in E_n$ if the $k$ bit in the binary expansion of
$n$ is a $1$.  All finite subsets of $\nat$ are of the form $E_n$ for
some $n$. Various combinators for ${\cal P}(\nat)$ are presented in
Exercise~\ref{7.23}.
\end{exm}
\vspace{.25in}
\begin{center}
{\bf Exercises}
\end{center}
\begin{exer}\label{7.13}
Show that an effectively given domain can always be identified with a
relation
\[INCL(n,m)\]
on integers where the derived relations
\[\begin{array}{lcl}
CONS(n,m)&\iff&\exists k.INCL(k,n)\mand INCL(k,m)\\
MEET(n,m,k)&\iff&\forall j.[INCL(j,k)\iff INCL(j,n)\mand INCL(j,m)]
\end{array}\]
are recursively decidable and where the following axioms hold:
\begin{enumerate}
\item $\forall n.INCL(n,n)$
\item $\forall n,m,k. INCL(n,m)\mand INCL(m,k)\imp INCL(n,k)$
\item $\exists m.\forall n. INCL(n,m)$
\item $\forall n,m. CONS(n,m)\imp\exists k.MEET(n,m,k)$
\end{enumerate}
\end{exer}
\begin{exer}\label{7.15}
Finish the proof of Theorem~\ref{7.4}.
\end{exer}
\begin{exer}\label{7.16}
Complete the proof of Theorem~\ref{7.5} by defining $curry$ as a
relation and showing it computable.  Is the set recursively
enumerable or is it recursive?
\end{exer}
\begin{exer}\label{7.18}
Two effectively given domains are {\it effectively isomorphic\/} iff
$\ldots$ Complete the statement of the theorem and prove it.
\end{exer}
\begin{exer}\label{7.23}
Complete the proof about the powerset in Example~\ref{7.8}.  Show that
the combinators $fun$ and $graph$ from Exercise~\ref{5.14} are
computable. Show the same for
\begin{enumerate}
\item $\lambda x,y.x\cap y$
\item $\lambda x,y.x\cup y$
\item $\lambda x,y.x+ y$
\end{enumerate}
where for $x,y\in{\cal P}(\nat)$,
\[x+y=\{n+m\such n\in x, m\in y\}\]
What are the computable elements of ${\cal P}(\nat)$?
\end{exer}
\newpage
\section{Sub-Spaces of the Universal Domain}\label{univdom}
To have a flexible method of solving domain equations and yielding
effectively given domains as the solutions, the domains will be
embedded in a {\it universal\/} domain which is ``big'' enough to hold
all other domains as sub-domains.  This universal domain is shown
to be effectively presented, and the mappings which define the
sub-spaces are shown to be computable.  First, the correspondence
between sub-spaces and mappings called {\it retractions\/} is
investigated, leading us to the definition of mappings called {\it projections\/}.  It is then shown that these definitions can be written
out using the $\lambda$-calculus notation, demonstrating the power of
our mathematical programming language.

\subsection{Retractions and Projections}
We start with the definition of {\it retractions\/}.
\begin{mdef}{Retractions}\label{8.1}
A {\it retraction\/} of a given domain $\dome$ is an approximable
mapping $a:\bas{E}\far\bas{E}$ such that $a\circ a=a$. 
\end{mdef}
Thus, a retraction is the identity function on objects in the
range of the retraction and maps other elements into range. The next
theorem relates these sets to sub-spaces. 
\begin{thm}\label{8.2}
If $\domd\subd\dome$ and if $a:\bas{E}\far\bas{E}$ is defined such that
\[X\:a\:Z \iff \exists Y\in\domd. Z\appx Y\appx X\]
for all $X,Z\in\bas{E}$, then $a$ is a retraction and $\domd$ is
isomorphic to the fixed point set of $a$, the set $\{y\in\dome\such
a(y)=y\}$, ordered under inclusion.
\end{thm}
\proof That $a$ is an approximable map is a direct consequence of the
definition of  sub-space (Definition~\ref{6.10}).  By
Theorem~\ref{6.12}, a projection pair, $i$ and $j$, exist for $\domd$
and this tells us that $a=i\circ j$ (also showing $a$ approximable
since approximable mappings are closed under composition).
Theorem~\ref{6.12} also tells us that $j\circ i=\ident_D$.  To show
that $a$ is a retraction, $a\circ a=a$ must be established.  Thus,
\[a\circ a = i\circ j\circ i\circ j = i\circ \ident_D\circ j = i\circ j =
a\]
holds, showing that $a$ is a retraction.

We now need to show the isomorphism to $\domd$.  For $x\in\domd$,
$i(x)\in\dome$ and we can calculate:
\[a(i(x))=i\circ j\circ i(x) = i\circ \ident_D(x) = i(x)\]
Thus, $i(x)$ is in the fixed point set of $a$.  For the other direction,
let $a(y)=y$.  Then $i(j(y)) = y$ holds.  But, $j(y)\in\domd$, so $i$
must map $\domd$ one-to-one and onto the fixed point set of $a$.
Since $i$ and $j$ are approximable, they are certainly monotonic, and
thus the map is an isomorphism with respect to set inclusion.  \eop

Not all retractions are associated with a sub-domain relationship.
The retractions defined in the above theorem are all subsets as
relations of the identity relation.  The retractions for sub-domains
are characterized by the following definition:
\begin{mdef}{Projections}\label{8.3}
A retraction $a:\dome\far\dome$ is a {\it projection\/} if $a\subseteq
\ident_E$ as relations.  The retraction is {\it finitary\/} iff its fixed
point set is isomorphic to some domain.
\end{mdef}

An example is in order.
\begin{exm}\label{8.4}
Consider a two element system, $\bas{O}$  with objects $\Delta$ and $0$.
For any basis $\bas{D}$ that is not trivial (has more than one
element), $\bas{O}$ comes from a 
retraction on $\bas{D}$.  Define a combinator
$check:\bas{D}\far\bas{O}$ by the relation
\[x\:check\: y \iff y=\Delta~{\rm or}~x\neq \Delta_D\]
Thus, $check(x)=\bottom_O\iff x=\bottom_D$. Another combinator can be
defined, 
\[fade:\bas{O}\times\bas{D}\far\bas{D}\]
such that for $t\in\dom{O}$ and $x\in\domd$
\[\begin{array}{lcll}
fade(t,x)&=&\bottom_D&{\rm if}~t=\bottom_O\\
&=&x&otherwise
\end{array}\]
For $u\in\domd$ and $u\neq\bottom_D$, the mapping $a$ is defined as
\[a(x)=fade(check(x),u)\]
It can be seen that $a$ is a retraction, but not a projection in
general, and the range of $a$ is isomorphic to $\bas{O}$.  

These combinators can also be used to define the subset of functions
in $\bas{D}\far\bas{E}$ that are strict.  Define a combinator
$strict:(\bas{D}\far\bas{E})\far(\bas{D}\far\bas{E})$ by the equation
\[strict(f)=\lambda x.fade(check(x),f(x))\]
with $fade$ defined as $fade:\bas{O}\times\bas{E}\far\bas{E}$.  The
range of $strict$ is all the strict functions; $strict$ is a projection
whose range is a domain.
\end{exm}

The next theorem characterizes projections.
\begin{thm}\label{8.5}
For approximable mapping $a:\bas{E}\far\bas{E}$, the following are
equivalent:
\begin{enumerate}
\item $a$ is a finitary projection
\item $a(x)=\{y\in\bas{E}\such\exists x'\in I_x. x'\:a\:x'\mand
y\appx x'\}$ for all $x\in\bas{E}$.
\end{enumerate}
\end{thm}
\proof Assume that (2) holds.  We want to show that $a$ is a finitary
projection.  By the closure properties on ideals, we know that for all
$x\in\dome$,
\[x'\in x\mand y\appx x'\imp y\in x\]
Thus, $a(x)\subseteq x$ must hold.  In addition, the following
trivially holds:
\[x'\in x\mand x'\:a\: x'\imp x'\in a(x)\]
thus $a(x)\subseteq a(a(x))$ holds for all $x\in\dome$.  This shows
that $a$ is indeed a projection.  Let $D=\{x\in\bas{E}\such
x\:a\:x\}$.  It is easy to show that $\bas{D}\subd\bas{E}$ and that
$a$ is determined from $\bas{D}$ as required in Theorem~\ref{8.2}.
Thus, the fixed point set of $a$ is isomorphic to a domain from the
previous proofs.  Thus, (2)$\imp$(1).

For the converse, assume that $a$ is a finitary projection.  Let
$\domd$ be isomorphic to the fixed point set of $a$.  This means there
is a projection pair $i$ and $j$ such that $j\circ i=\ident_D$ and $i\circ
j = a$ and $a\subseteq \ident_E$.  From Theorem~\ref{6.15} then we have
that $\domd\iso\domd'$ and $\domd'\subd \dome$.  We want to identify
$\domd'$ as follows:
\[\domd'=\{x\in\dome\such x\:a\:x\}\]
From the proof of Theorem~\ref{6.15}, the basis elements of $\bas{D'}$
are the finite elements of $\bas{D}$.  Each of these elements is in
the fixed point set of $a$.  Thus, 
\[x\in\bas{D'}\imp a({\ideal}_x) = {\ideal}_x \imp x\:a\:x\]
Since $a$ is a projection, ${\ideal}_x$ must also be a fixed point.
Since $i(j({\ideal}_x)) = {\ideal}_x$ implies that $j({\ideal}_x)$ is
a finite element of $\domd$, $x\in\domd'$ must hold.  Thus, the
identification of $\domd'$ holds.  

Finally, using $a=i\circ j$ in the formula in Theorem~\ref{6.12}, the
formula in (2) is obtained, proving the converse. \eop

This characterization of projections provides a new and interesting
combinator.
\begin{thm}\label{8.6}
For any domain $\dome$, define
$sub:(\dome\far\dome)\far(\dome\far\dome)$ using the relation
\[x\: sub(f)\: z \iff \exists y\in\bas{E}.y\:f\:y\mand y\appx
x\mand z\appx y\]
for all $x,z\in\bas{E}$ and all $f:\bas{E}\far\bas{E}$.  Then the
range of $sub$ is exactly the set of finitary projections on $\dome$.
In addition, $sub$ is a finitary projection on $\dome\far\dome$.  If
$\dome$ is effectively given, then $sub$ is computable.
\end{thm}
\proof Clearly, $sub(f)$ is approximable.  It is obvious from the
definition that $f\mapsto sub(f)$ preserves lubs and
thus is approximable as well.  Thus,
\[y\:f\:y\mand y\appx x\mand z\appx y\imp x\:f\:z\]
obviously holds.  Thus, $sub(f)\subseteq f$ holds.  Also
\[y\:f\:y\imp y\:sub(f)\:y\]
thus, $sub(f)\subseteq sub(sub(f))$ holds as well.  Thus, $sub$ is a
projection on $\dome\far\dome$.  The definition of the relation shows
that it is computable when $\dome$ is effectively given.

Since $sub$ is a projection, its range is the same as its fixed point
set.  If $sub(a)=a$, it is easy to see that clause (2) of
Theorem~\ref{8.5} holds and {\it conversely\/}.  Thus, the range of
$sub$ is the finitary projections.  

To see that $sub$ is a finitary projection, we use Theorem~\ref{6.11}
and Theorem~\ref{8.2} to say that the fixed point set of $sub$ is in a
one-to-one inclusion preserving correspondence with the domain
$\{D\such D\subd\dome\}$. \eop

\subsection{Universal Domain $\domu$}
With these results and the universal domain to be defined next, the
theory of sub-domains is translated into the $\lambda$-calculus
notation using the $sub$ combinator.  The universal domain is defined
by first defining a domain which has the desired structure but has a
{\it top\/} element.  The top element is then removed to give the
universal domain.
\begin{mdef}{Universal Domain}\label{8.7n}
As in the section on domain equations, an inductive definition for a
domain $\domv$ is given as follows:
\begin{enumerate}
\item $\Delta,\top\in\bas{V}$
\item $\langle u,v\rangle \in\bas{V}$ whenever $u,v\in\bas{V}$
\end{enumerate}
Thus, we are starting with two objects, a bottom element and a top
element, and making two flavors of copies of these objects.
Intuitively, we end up with finite binary trees with either the top or
the bottom element as the leaves. To simplify the definitions below,
the pairs should be reduced such that:
\begin{enumerate}
\item All occurrences of $\langle \Delta,\Delta\rangle $ are replaced by $\Delta$ and
\item All occurrences of $\langle \top,\top\rangle $ are replaced by $\top$.
\end{enumerate}
These rewrite rules are easily shown to be finite
Church-Rosser.\footnote{The finitary basis should be defined as the
equivalence classes induced by the reduction.  The presentation is
simplified by considering only reduced trees.}
As an example of the reduction the pair
\[\langle \langle \langle \top,\langle \top,\top\rangle \rangle ,\langle \top,\Delta\rangle \rangle ,\langle \langle \Delta,\Delta\rangle ,\langle \top,\top\rangle \rangle \rangle \]
reduces to 
\[\langle \langle \top,\langle \top,\Delta\rangle \rangle ,\langle \Delta,\top\rangle \rangle \].
The approximation ordering is defined as follows:
\begin{enumerate}
\item $\Delta\appx v$ for all $v\in\bas{V}$
\item $v\appx \top$ for all $v\in\bas{V}$.
\item $\langle u,v\rangle \appx \langle u',v'\rangle $ iff $u\appx u'$ and $v\appx v'$
\end{enumerate}

Since the top element is approximated by everything, all finite sets
of trees are consistent.  The lub for a pair of trees is
defined as follows:
\begin{enumerate}
\item $u\llub\top=\top$ for $u\in\bas{V}$
\item $\top\llub u=\top$ for $u\in\bas{V}$
\item $u\llub\Delta=u$ for $u\in\bas{V}$
\item $\Delta\llub u=u$ for $u\in\bas{V}$
\item $\langle u,v\rangle \llub \langle u',v'\rangle =\langle u\llub u',v\llub v'\rangle $ for $u,v\in\bas{V}$
\end{enumerate}

The proof that this forms a finitary basis follows the same guidelines
as the proofs in Section~\ref{sec6}.  In addition, it should be clear
that the presentation is effective.  

To form the universal domain, the top element is simply removed.
Thus, the system $\bas{U}=\bas{V}-\{\top\}$ is the basis used to form
the universal domain.  The proof that this is still a finitary basis
with an effective presentation is also straightforward and left to the
exercises. Note that inconsistent sets can now exist since there is no
top element.  A set is inconsistent iff its lub is
$\top$. 
\end{mdef}

We shall now prove the claims made for the universal domain.
\begin{thm}\label{8.8}
The domain $\domu$ is universal, in the sense that for every domain
$\domd$ we have $\domd\subd\domu$. If $\domd$ is
effectively given, then the projection pair for the embedding is
computable. In fact, there is a correspondence between the effectively
presented domains and the computable finitary projections of $\domu$.
\end{thm}
\proof Recall that $\bas{D}$ must be countable to be a finitary basis.
Thus, we can assume that the basis has an enumeration
\[D=\{X_n\such n\in\nat\}\]
where $X_0=\Delta$.  The effective and general cases are considered
together in the proof; comments about computability are included for
the effective case as required.  Thus, if $\domd$ is effectively
given, the enumeration above is assumed to be computable.

To prove that the domain can be embedded in $\domu$, the embedding
will be shown.  To start, for each finite element $d_i$ in the basis, define
two sets, $d_i^+$ and $d_i^-$ as follows:
\[\begin{array}{lcl}
d_i^+&=&\{d\in\bas{D}\such d_i\appx d\}\\
d_i^-&=&D-d_i^+
\end{array}\]
The $d_i^+$ set contains all the elements that $d_i$ approximates,
while the $d_i^-$ set contains all the other elements, partitioning
$\bas{D}$ into two disjoint sets.  Sets for different elements can be
intersected to form finer partitions of $\bas{D}$.  For $k>0$, let
$R\in\{+,-\}^k$, let $R_i$ be the $ith$ symbol in the string $R$,  and
define a region $D_R$ as 
\[D_R=\bigcap\limits_{i=1}^k d_i^{R_i}\]
where $k$ is the length of $R$.
The set $\{D_{R}\such R\in\{+,-\}^k\}$  of regions partitions $\bas{D}$
into $2^k$ disjoint sets.  Thus, for each element $e_i$ in the
enumeration there is a corresponding partition of the basis given by
the family of sets $\{D_{R}\such R\in\{+,-\}^i\}$. For strings
$R,S\in\{+,-\}^*$ such that $R$ is a prefix of $S$, denoted $R\leq S$,
$D_S\subseteq D_R$.   It is important to realize that the composition
of these sets is dependent on the order in which the elements are enumerated.
Some of these regions are empty, but it is  
decidable if a given intersection is empty if $\domd$ is effectively
presented. It is also decidable if a given element is in a particular
region. 

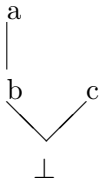
\begin{figure}
\begin{center}
\begin{picture}(100,100)
\put(30,60){a}
\put(30,30){b}
\put(60,30){c}
\put(40,0){$\bottom$}
\put(30,59){\line(0,-1){18}}
\put(30,29){\line(1,-1){15}}
\put(60,29){\line(-1,-1){15}}
\end{picture}
\end{center}
\caption{Example Finite Domain\label{findom}}
\end{figure}
To see the function these regions are serving, consider the 
finite domain in Figure~\ref{findom}.\footnote{This example is taken
from Cartwright and Demers~\cite{tt}.} Consider the enumeration with
\[d_0=\bottom, d_1=b, d_2=c, d_3=a.\]
The $d_i^+$ and $d_i^-$ sets are as follows:
\[\begin{array}{lcl}
d_1^+&=&\{a,b\}\\
d_1^-&=&\{c,\bottom\}\\
d_2^+&=&\{c\}\\
d_2^-&=&\{a,b,\bottom\}\\
d_3^+&=&\{a\}\\
d_3^-&=&\{b,c,\bottom\}
\end{array}\]
The regions are as follows:
\[\begin{array}{lclclcl}
D_+ &=&\{a,b\}&\:\:\:\:&D_{+++} &=&\{\}\\
D_- &=&\{\bottom,c\}&&D_{++-} &=&\{\}\\
D_{++} &=&\{\}&&D_{+-+} &=&\{a\}\\
D_{+-} &=&\{a,b\}&&D_{+--} &=&\{b\}\\
D_{-+} &=&\{c\}&&D_{-++} &=&\{\}\\
D_{--} &=&\{\bottom\}&&D_{-+-} &=&\{c\}\\
&&&&D_{--+} &=&\{\}\\
&&&&D_{---} &=&\{\bottom\}
\end{array}\]
The regions generated by each successive element encode the
relationships induced by the approximation ordering between the new
element and all elements previously added.  
The reader is encouraged to try this example with other enumerations
of this basis and compare the results. 

The embedding of the elements proceeds by building a tree based on
the regions corresponding to the element.  The regions are used to
find locations in the tree and to determine whether a $\top$ or a
$\Delta$ element is placed in the location.  These trees preserve the
relationships specified by the regions and thus, the tree embedding is
isomorphic to the domain in question.  Once the tree is built,
the reduction rules are applied until a non-reducible tree is reached.
This tree is the representative element in the universal domain, and
the set of these trees form the sub-space.  

The function to determine the location in the tree for a given domain,
\[Loc_D:\{+,-\}^*\far\{l,r\}^*\]
 takes strings used to generate regions
and outputs a path in a tree where $l$ stands for left sub-tree and $r$
stands for right sub-tree.  This path is computed using the following
inductive definition:
\[\begin{array}{lcll}
Loc_D(\epsilon)&=&\epsilon.\\
Loc_D(R+)&=&Loc_D(R)l&{\rm if }~D_{R+}\not=\emptyset~ {\rm and }~
D_{R-}\not=\emptyset.\\
&=&Loc_D(R)&{\rm otherwise}.\\
Loc_D(R-)&=&Loc_D(R)r&{\rm if }~D_{R+}\not=\emptyset~{\rm  and }~
D_{R-}\not=\emptyset.\\
&=&Loc_D(R)&{\rm otherwise}.
\end{array}\]

The set of locations for each non-empty region is the set of paths to
all leaves of some finite binary tree. An induction argument is used
to show the following properties of $Loc_D$ that ensure this:
\begin{enumerate}
\item If $R\leq S$ for $R,S\subseteq \{+,-\}^*$, then $Loc_D(R)\leq
Loc_D(S)$. 
\item Let $S=\{Loc_D(R)\such R\in\{+,-\}^k\mand D_R\neq \emptyset\}$ for
$k>0$ be a set of location paths for a given $k$.  For any
$p\in\{l,r\}^*$ there exists $q\in S$ such that either $p\leq q$ or
$q\leq p$.  That is, every potential path is represented by some
finite path.
\item  Finally, for all $p,q\in S$ if $p\leq q$ then $p=q$.
This means that a unique leaf is associated with each location. 
\end{enumerate}

To find the tree for a given element $d_k$ in the enumeration, apply
the following rules to each $R\in\{+,-\}^{k-1}$.
\begin{enumerate}
\item If $D_{R-}\neq \emptyset$ then the leaf for path $Loc_D(R-)$ is
labeled $\top$.
\item If $D_{R+}\neq \emptyset$ then the leaf for path $Loc_D(R+)$ is
labeled $\Delta$.
\end{enumerate}

These rules are used to assign a tree in $\bas{U}$, which is then reduced
using the reduction rules, for each element in
the enumeration of $\bas{D}$.  To see that the top element is never
assigned by these rules, note that some region of the form $R+$ for
every length $k$ 
must be non-empty since it must contain the element $e_k$ being embedded.

Returning to the example, the location function defines paths for
these elements as follows:
\[\begin{array}{lclclcl}
Loc_D(+)&=&l&\:\:\:\:&Loc_D(+-+)&=&ll\\
Loc_D(-)&=&r&&Loc_D(+--)&=&lr\\
Loc_D(+-)&=&l&&Loc_D(-+-)&=&rl\\
Loc_D(-+)&=&rl&&Loc_D(---)&=&rr\\
Loc_D(--)&=&rr
\end{array}\]

The trees generated for each of the elements are:
\[\begin{array}{lcl}
d_0&\mapsto &\Delta\\
d_1&\mapsto &\langle \Delta,\top\rangle \\
d_2&\mapsto &\langle \top,\langle \Delta,\top\rangle \rangle \\
d_3&\mapsto& \langle \langle \Delta,\top\rangle ,\langle \top,\top\rangle \rangle \\
&\mapsto&  \langle \langle \Delta,\top\rangle ,\top\rangle 
\end{array}\]

To verify that the space generated is a valid sub-space, we must
verify that the bottom element is mapped to $\bottom_U$ and that the
consistency 
and lub relations are maintained. The tree $\Delta$ is
clearly assigned to $X_0$, the bottom element for 
the basis being embedded, since there are no strings of length $-1$. 
The embedding preserves inconsistency of elements by forcing the lub
 of the embedded elements to be $\top$.  The $D_{R-}$
regions represent the elements that the element being
embedded does
not approximate.  Note that the $D_{R-}$ sets cause the $\top$ element
to be added as the leaf.  Since the $D_R$ sets are built using the
approximation ordering, it is straightforward to see that the
approximation ordering is preserved by the embedding.  Lubs
 are also maintained by the embedding, although the reduction is
required to see that this is the case.
It should be clear that, if the domain $\domd$ is effectively given,
the sub-space can be computed since the embedding procedure uses the
relationships given in the presentation.  

Finally, suppose that $a$ is a computable, finitary projection on
$\domu$.  From the proof of Theorem~\ref{8.5}, the domain of this
projection is characterized by the set
\[\{y\in\bas{U}\such y\:a\:y\}\]
If $a$ is computable, the set of pairs for $a$ is recursively
enumerable.  Thus, the set above is also recursively enumerable since
equality among basis elements is decidable.  Thus, the domain given by
the projection must also be effectively given. \eop

Thus, the domain $\domu$ is an effectively presented {\it universal\/}
domain in which all other domains can be embedded.  The
sub-domains of $\domu$ include $\domu\far\domu$, $\domu\times\domu$,
etc. These domains must be sub-domains of $\domu$ since they are
effectively presented based on our earlier theorems.
\subsection{Domain Constructors in $\domu$}
 The next step is
to see how to define the constructors commonly used.
\begin{mdef}{Domain Constructors}\label{8.9}
Let the computable projection pair, 
\[i_+:\domu +\domu\far\domu~{\rm and}~j_+:\domu\far\domu +\domu\]
be fixed.  Fix suitable projection pairs
$i_\times,j_\times,i_\far$, and $j_\far$ as well.  Define
\[\begin{array}{lcl}
a+b&=&cond\circ \langle which,i_+\circ in_0\circ a\circ out_0, i_+\circ
in_1\circ b\circ out_1\rangle \circ j_+\\
a\times b&=&i_x\circ \langle a\circ proj_0,b\circ proj_1\rangle \circ j_x\\
a\far b&=&i_\far\circ (\lambda f.b\circ f\circ a)\circ j_\far
\end{array}\]
for all $a,b:\domu\far\domu$.
\end{mdef}
From earlier theorems, we know that these combinators are all
computable over an effectively presented domain.  The next theorem
characterizes the effect these combinators have on projection
functions.
\begin{thm}\label{8.10}
If $a,b:\domu\far\domu$ are projections, then so are $a+b$, $a\times
b$, and $a\far b$.  If $a$ and $b$ are finitary, then so are the
compound projections.  
\end{thm}

\proof Since $a$ and $b$ are retractions, $a=a\circ a$ and $b=b\circ
b$. Then for $a\times b$ using the definition of $\times$,
\[\begin{array}{lcl}
(a\times b)\circ (a\times b)&=&i_x\circ \langle a\circ proj_0,b\circ
proj_1\rangle \circ \langle a\circ proj_0,b\circ proj_1\rangle \circ j_x\\
&=&i_x\circ \langle a\circ a\circ proj_0,b\circ b\circ proj_1\rangle \circ j_x\\
&=& a\times b
\end{array}\]
Thus, $a\times b$ is a retraction.  The other cases follow similarly.

Since $a$ and $b$ are projections, $a,b\subseteq \ident_U$ (denoted simply
$\ident$ for the remainder of the proof).  Using the definition for $+$
along with the above relation and the definition of projection pairs,
we can see that 
\[a+b\subseteq \ident+\ident=i_+\circ j_+ \subseteq \ident\]
Thus, $a+b$ is a projection. The other cases follow similarly.

To show that the projections are finitary, we must show that the fixed
point sets are isomorphic to a domain.  Since $a$ and $b$ are assumed
finitary, their fixed point sets are isomorphic to 
\[\begin{array}{lcl}
D_a&=&\{x\in\bas{U}\such x\:a\: x\}\\
D_b&=&\{y\in\bas{U}\such y\:b\: y\}
\end{array}\]
 
We wish to show that $\domd_a\far\domd_b\iso\domd_{a\far b}$.  By the
definition of the $\far$ constructor, the fixed point set of $a\far b$
over $\domu$ is the same as the fixed point set of $\lambda f.b\circ
f\circ a$ on $\domu\far\domu$. (Hint: $i_\far$ and $j_\far$ set up the
isomorphism.) So, the fixed points for $f:\domu\far\domu$ are of the form:
\[f=b\circ f\circ a\]
We can think of $a$ as a function in $\domu\far\domd_a$ and define the
other half of the projection pair as $i_a:\domd_a\far\domu$ where
$i_a\circ a = a$ and $a\circ i_a=i_a$.  Define a function $i_b$ for
the projection pair for $b$ similarly.  For some
$g:\domd_a\far\domd_b$ let
\[f=i_b\circ g\circ a\]
Substituting this definition for $f$ yields
\[b\circ f\circ a = b\circ i_b\circ g\circ a\circ a = i_b\circ g \circ
a = f\]
by the definition of $i_b$ and since $a$ is a retraction by
assumption.  Conversely, for a function $f$ such that $i_b\circ g\circ a
= f$, let
\[g=b\circ f\circ i_a\]
Substituting again,
\[i_b\circ g\circ a = i_b\circ g\circ f\circ i_a\circ a = b\circ
f\circ a = f\]
Thus, there is an order preserving isomorphism between
$g:\domd_a\far\domd_b$ and the functions $f=b\circ f\circ a$.  The
proofs of the isomorphisms for the other constructs are similar. \eop

Thus, the sub-domain relationship with the universal domain has been
stated in terms of finitary projections over the universal domain.  In
addition, all the domain constructors have been shown to be computable
combinators on the domain of these finitary projections.  Recalling
that all computable maps have computable fixed points, the standard
fixed point method can be used to solve domain equations of all kinds
if they can be defined on projections.  

Returning to the $\lambda$-calculus for a moment, all objects in the
$\lambda$-calculus are considered functions.  Since $\domu\far\domu$
is a part of $\domu$, every object in the $\lambda$-calculus is also 
an object of $\domu$.  Transposing some of the familiar
notation, where the old notation appears on the left, the new
combinators are defined as follows:
\[\begin{array}{lcl}
which(z)=which(j_+(z))\\
in_i(x)=i_+(in_i(x))~{\rm where}~i=0,1\\
out_i(x)=out_i(j_+(x))~{\rm where}~i=0,1\\
\langle x,y\rangle =i_x(\langle x,y\rangle )\\
proj_i=proj_i(j_x(z))~{\rm where}~i=0,1\\
u(x) = j_\far(u)(x)\\
\lambda x.\tau=i_\far(\lambda x.\tau)
\end{array}\]

Thus, all functions, all constants, all combinators, and all
constructs are elements of $\domu$.  Indeed, {\bf everything} computable is an
element of $\domu$.  Elements in $\domu$ play multiple roles by
representing different objects under different projections.  While
this notion may be difficult to get used to, there are many
advantages, both notational and conceptual.
\vspace{.25in}
\begin{center}
{\bf Exercises}
\end{center}
\begin{exer}\label{8.14}
A retraction $a:\domd\far\domd$ is a {\it closure operator\/} iff
$\ident_D\subseteq a$ as relations.  On a domain like ${\cal P}(\nat)$, give
some examples of closure operators.  (Hint: Close up the integers
under addition.  Is this continuous on ${\cal P}(\nat)$?) Prove in
general that for any closure $a:\domd\far\domd$, the fixed point set
of $a$ is always a finitary domain.  (Hint: Show that the fixed point
set is closed as required for a domain.) What are the finite elements
of the fixed point set?
\end{exer}
\begin{exer}\label{8.15}
Give a direct proof that the domain $\{X\such X\subd \domd\}$ is
effectively presented if $\domd$ is.  (Hint: The finite elements of
the domain correspond exactly to the finite domains $X\subd \domd$.)
In the case of $\domd=\domu$, show that the computable elements of the
domain correspond exactly to the effectively presented domains (up to
effective isomorphism).
\end{exer}
\begin{exer}\label{8.16}
For finitary projections $a:\dome\far\dome$, write
\[\domd_a=\{x\in\dome\such x\:a\:x\}\]
Show that for any two such projections $a$ and $b$, that
\[a\subseteq b \iff \domd_a\subd\domd_b\]
\end{exer}
\begin{exer}\label{8.17}
Find another universal domain that is not isomorphic to $\domu$.
\end{exer}
\begin{exer}\label{8.18}
Prove the remaining cases in Theorem~\ref{8.10}.
\end{exer}
\begin{exer}\label{8.24}
Suppose $S$ and $T$ are two binary constructors on domains that can be
made into computable operators on projections over the universal
domain.  Show that we can find a pair of effectively presented domains
such that
\[D\iso S(D,E)~{\rm and}~E\iso T(D,E).\]
\end{exer}
\begin{exer}\label{8.26}
Using the translations shown after the proof of Theorem~\ref{8.10},
show how the whole {\it typed\/}-$\lambda$-calculus can be translated
into $\domu$.  (Hint: for $f:\domd_a\far\domb$, write $f=b\circ f\circ
a$ for finitary projections $a$ and $b$.  For $\lambda
x^{\domd_a}.\sigma$, write $\lambda x.b(\sigma'[a(x)/x])$ where
$\sigma'$ is the translation of $\sigma$ into the untyped
$\lambda$-calculus.  Be sure that the resulting term has the right
type.)
\end{exer}
\begin{exer}
Show that the basis presented for the universal domain $\bas{U}$ is
indeed a finitary basis and that it has an effective presentation.
\end{exer}
\begin{exer}
Work out the embedding for the other enumerations for the example
given in the proof of Theorem~\ref{8.8}.
\end{exer}


\newpage
\begin{thebibliography}{9}
\bibitem{blc}Barendregt, H.P. {\em The Lambda Calculus: Its Syntax and
Semantics\/}, Revised Edition.  Studies in Logic and the Foundations
of Mathematics 103. North-Holland, Amsterdam, 1984.
\bibitem{tt}Cartwright, Robert, and Demers, Alan. The Topology of
 Program Termination. Proceedings of the Third Annual Symposium on Logic in Computer Science ({LICS} '88), 296--308, 1988.
\bibitem{ns}Scott, Dana.  Lectures on a Mathematical Theory of Computation.
Technical Monograph PRG-19, Oxford University Computing Laboratory, 
Oxford, 1981.
\bibitem{is}Scott, Dana.  Domains for Denotational Semantics.  Technical 
Report, Computer Science Department, Carnegie Mellon University, 1983.
\bibitem{stoy}Stoy, Joseph E.  {\em Denotational Semantics: The Scott-Strachey Approach to Programming Language Theory.}  MIT Press, 1977.
\bibitem{rogers}Hartley Rogers, J. {\em Theory of Recursive Functions and Effective Computability.} McGraw-Hill Book Company, 1967.

\end{thebibliography}
\end{document}